\documentclass[12pt]{article}

\usepackage{amsmath}
\usepackage{mathrsfs}
\usepackage[round]{natbib}
\usepackage{amssymb,color}
\usepackage{appendix}
\usepackage{amsthm}
\usepackage{enumerate}
\usepackage{hyperref}
\usepackage{mathtools} 
\mathtoolsset{showonlyrefs=true}  

\usepackage{caption}

\usepackage[margin=0.7in]{geometry}

\numberwithin{equation}{section}
\usepackage{comment}
\usepackage{enumitem}
\newlist{steps}{enumerate}{1}
\setlist[steps, 1]{label = Step \Roman*:}
\newlist{driftcondi}{enumerate}{1}
\setlist[driftcondi]{label=(C\arabic*)}

\usepackage{titling}
\newtheorem{thm}{Theorem}[section]
\newtheorem{cor}[thm]{Corollary}

\newtheorem{assume}{Assumption}[section]
\newtheorem{prop}[thm]{Proposition}

\newtheorem{lemma}[thm]{Lemma}

\theoremstyle{definition}
\newtheorem{remark}{Remark}[section]
\newtheorem{defi}{Definition}[section]

\newcommand{\upperRomannumeral}[1]{\uppercase\expandafter{\romannumeral#1}}

\DeclarePairedDelimiter{\norm}{\lVert}{\rVert}

\newcommand{\mbp}{\mathbb{P}}
\newcommand{\mbq}{\mathbb{Q}}

\newcommand{\tmbp}{\widetilde{\mathbb{P}}}
\newcommand{\hmbp}{\widehat{\mathbb{P}}}

\newcommand{\bmbp}{\overline{\mathbb{P}}}
\newcommand{\mbr}{\mathbb{R}}
\newcommand{\EP}{\mathbb{E}^{\mathbb{P}}}
\newcommand{\EQ}{\mathbb{E}^{\mathbb{Q}}}
\newcommand{\ETP}{\mathbb{E}^{\tmbp}}

\newcommand{\EHP}{\mathbb{E}^{\hmbp}}

\newcommand{\EBP}{\mathbb{E}^{\bmbp}}

\newcommand{\hlambda}{\hat{\lambda}}

\newcommand{\tB}{\widetilde{B}}

\newcommand{\de}{\frac{\partial}{\partial\epsilon}\biggl|_{\epsilon=0}}
\newcommand{\sube}{_{\epsilon}}
\newcommand{\supe}{^{\epsilon}}

\newcommand{\ETPe}{\mathbb{E}^{\tmbp\supe}}
\newcommand{\EHPe}{\mathbb{E}^{\hmbp\supe}}

\newcommand{\hY}{\hat{Y}}

\newcommand{\RN}[1]{
	\textup{\uppercase\expandafter{\romannumeral#1}}
}

\begin{document}

\author{Hyungbin Park\thanks{hyungbin@snu.ac.kr, hyungbin2015@gmail.com} and  Heejun Yeo\thanks{yeohj@snu.ac.kr}  \\ \\ \normalsize{Department of Mathematical Sciences} \\ 
		\normalsize{Seoul National University}\\
		\normalsize{1, Gwanak-ro, Gwanak-gu, Seoul, Republic of Korea} 
	}
\title{Dynamic and static fund separations and their stability for long-term optimal investments}
	
\maketitle	
	
\begin{abstract}
	This paper investigates
	dynamic and static fund separations and their stability for long-term optimal investments under three model classes. An investor maximizes the expected utility with constant relative risk aversion under an incomplete market consisting of a safe asset, several risky assets, and a single state variable. The state variables in two of the model classes follow a 3/2 process and an inverse Bessel process, respectively. The other market model has the partially observed state variable modeled as an Ornstein-Uhlenbeck state process.
	We show that the dynamic optimal portfolio of this utility maximization consists of $m+3$ portfolios: the safe asset,  the myopic portfolio, the \textit{m} time-independent portfolios, and the intertemporal portfolio.
	Over time, the intertemporal portfolio eventually vanishes, leading the dynamic portfolio to converge to $m+2$ portfolios, referred to as the static portfolio.
	We also prove that the convergence is stable under model parameter perturbations.
	In addition, sensitivities of the intertemporal portfolio with respect to small parameters perturbations also vanish in the long run.
	The convergence rate for the intertemporal portfolio and its sensitivities are computed explicitly for the presented models.
\end{abstract}

{\small\hspace{5mm}\textbf{\textsc{Key Words: }} fund separation, long-term investment, non-affine model, partial observation}

\section{Introduction}

\subsection{Overview}
An optimal investment problem is typically formulated as a problem of finding an optimal trading strategy that maximizes the expected utility of an investor. If a market depends on \textit{d} state variables that models investment opportunities, classical \textit{dynamic fund separation} asserts that the optimal portfolio of the investor consists $d+2$ funds: the safe asset, the myopic portfolio, and \textit{d} intertemporal portfolios (\citet{merton1973intertemporal}). This paper refines the dynamic fund separation for an investor with constant relative risk aversion (CRRA) under markets with a single state variable by further decomposing the intertemporal portfolio into a time-independent portfolio and another intertemporal portfolio. In addition, we show that the weight of the new intertemporal portfolio vanishes over time. The result provides \textit{static fund separation} in long-term investments. Moreover, the new intertemporal portfolio becomes insensitive to small parameter perturbations in the long run, which ensures stability of the optimal portfolio under parameter perturbations.

We outline the classical dynamic separation theorem for a market with a single state variable in the following simple form. Let $T>0$ be a terminal time and $\hat{\pi}_T(t,y)$ denote the optimal portfolio at current time \textit{t} and current state \textit{y}. Then $\hat{\pi}_T$ can be expressed as
\begin{equation} \label{port:optimal}
	\hat{\pi}_T(t,y)=w_0(p)M(y)+\tilde{w}(p,t,T,y)I(y),
\end{equation}
where $M,w,I,\tilde{w}$ denote the myopic portfolio and its weight, the intertemporal portfolio and its weight, respectively, and \textit{p} refers to the preference parameter. Our dynamic separation further decomposes $\tilde{w}$ to the form
\begin{equation}
	\hat{\pi}_T(t,y)=w_0(p)M(y)+\tilde{w}_1(p,y)I(y)+\tilde{w}_2(p,t,T,y)I(y).
\end{equation}
Such an additional decomposition of the intertemporal portfolio can be performed by applying the Hansen--Scheinkman factorization developed by \citet{hansen2009long} (see also \citet{hansen2012dynamic}, \citet{qin2016positive}, and \citet{park2018sensitivity}). The vanishing of $\tilde{w}_2$ in the long run can be argued by employing the ergodic theorems for diffusion processes (see, e.g., \citet{pinchover2004large}). Thus, the optimal portfolio eventually converges to the time-independent portfolio in the long-run
\begin{equation} \label{port:static}
	\hat{\pi}_{\infty}(y)=w_0(p)M(y)+\tilde{w}_1(p,y)I(y).
\end{equation}
We call this portfolio the static optimal portfolio, in that a long-term investor does not have to adjust her funds over time.

The term \textit{static fund separation} was first used by \citet{guasoni2015static}; their static fund separation theorem involves more than that mentioned in the previous paragraph. For two classes of affine models, they separated variables \textit{p} and \textit{y} of the weight $\tilde{w}_1$ in \eqref{port:static} as \begin{equation}
	\tilde{w}_1(p,y)=\sum_{i=1}^m w_i(p)\psi_i(y)
\end{equation}
for some positive number \textit{m}. Then the optimal portfolio \eqref{port:optimal} and the static optimal portfolio \eqref{port:static} are further decomposed into
\begin{equation}
	\hat{\pi}_{T}(t,y)=w_0(p)M(y)+\sum_{i=1}^m w_i(p)\psi_i(y)I(y)+\tilde{w}_2(p,t,T,y)I(y)
\end{equation}
and
\begin{equation}
	\hat{\pi}_{\infty}(y)=w_0(p)M(y)+\sum_{i=1}^m w_i(p)\psi_i(y)I(y),
\end{equation}
respectively. Thus, the optimal portfolio now consists of the safe asset, $m+1$ preference-free static funds $M(y),(\psi_iI)(y)$ with weights $w_i(p)$ that depend on the risk aversion but not on the state variable \textit{y}, and one intertemporal portfolio. In addition, the static optimal portfolio consists of the safe asset and $m+1$ preference-free static funds $M(y),(\psi_iI)(y)$ with weights $w_i(p)$. They call the latter separation as the \textit{static fund separation theorem}. We also derive the same type of separations for two non-affine mode classes of models and one affine class of model with partially observable state variable.

The main contributions of this paper include two additional results compared to \citet{guasoni2015static}. First, the exponential convergence rate is computed explicitly for each model, whereas \citet{guasoni2015static} provided only convergence of the optimal portfolio without explicit convergence rate. The Hansen--Scheinkman decomposition plays an important role in finding the convergence rate. Second, sensitivity analysis of the optimal portfolio is conducted. We argue that the convergence of the optimal portfolio is stable under model parameter perturbations by showing that sensitivities of the intertemporal weight to parameter perturbations vanish in the long run. Mathematically, the sensitivity of a portfolio with respect to perturbations of a specific parameter is measured by the partial derivative of the portfolio with respect to the parameter. \citet{park2018sensitivity} provides a method to calculate the partial derivative explicitly.

One major difference compared to \citet{guasoni2015static} is that we demonstrate our arguments for various types of state variable models: a 3/2 state process and an inverse Bessel state process, which are non-affine, and one partially observable Ornstein-Uhlenbeck state process. However, \citet{guasoni2015static} worked with only two affine state processes. While it is more difficult to handle non-affine models than affine models, it is worth studying the former because several empirical studies argue that they outperform affine models as factor models in a financial market. Investigating the market with a partially observable state variable is also meaningful in that the state variable are non-tradable; thus, investors cannot actually observe it directly. We employ the Kalman-Bucy filter to estimate the state variable through given stock prices.

The remainder of this paper is organized as follows. Section \ref{sec:setup} describes the market model and the main arguments of this work.
Explicit results for three market models charaterized by a 3/2 state process, an inverse Bessel state process, and a partially observable Ornstein-Uhlenbek state process are presented in Sections \ref{sec:3/2}, \ref{sec:invB}, and \ref{sec:FOU} respectively. Section \ref{sec:outline} outlines the approach to proofs for the main arguments. Finally, the last section summarizes the results of this paper. Detailed proofs are provided in the appendices.

\subsection{Related literature}
The optimal long-term investment problem has been studied in various ways. \citet{fleming2000risk} and \citet{fleming2002risk} studied the problem of maximizing the long-term growth rate of the expected utility of an investor with hyperbolic absolute risk aversion. \citet{guasoni2012portfolios} suggested a method to derive long-term optimal porfolios explicitly for investors with CRRA in an incomplete market. \citet{robertson2017long} proved the convergence of finite horizon optimal portfolios to their long-term counterparts for investors with CRRA in matrix-valued factor models. They also proved turnpike theorems in their models, which suggest that for long horizons, the optimal portfolio for a generic utility approaches that of a CRRA utility.

Sensitivity analysis of optimal investment is also an important area of study in financial mathematics. \citet{larsen2007stability} investigated the stability properties of utility-maximization for complete and incomplete markets by observing how small perturbations of the market parameters affect an investor's optimal behavior. \citet{larsen2018expansion} analyzed the effect of perturbations in the market price of risk on the optimal expected utility for an investor with CRRA. Subsequently, \citet{mostovyi2019sensitivity} extended the results to the case of general utility functions. \citet{mostovyi2020asymptotic} investigated the sensitivity of the optimal expected utility under small perturbations of the num\'{e}raire in an incomplete market. \citet{park2019sensitivity} conducted sensitivity analysis of the optimal expected utility with respect to small fluctuations in the initial value, the drift function, and the volatility function of a factor process in an incomplete market.

The 3/2 process is one of the most popular such non-affine processes. Empirical studies conducted by \citet{drimus2012options} show that as a model for stochastic volatility, the 3/2 model captures volatility smiles better than the square-root process. \citet{goard2013stochastic} also empirically compared eight stochastic volatility models including several affine models and the 3/2 model in their ability to capture the evolution of the volatility index and found that the 3/2 model showed better performance. At the same time, the 3/2 model is known to be analytically tractable to some extent(see, e.g., \citet{lewis2000option} and \citet{zheng2016pricing}). In addition to stochastic volatility modeling, the 3/2 process has also been used to model short interest rates. Using statistical testing, \citet{ahn1999parametric} demonstrated that the 3/2 process outperforms the existing affine-class models in that it better captures the actual behavior of the interest rate. The inverse Bessel process, as the name suggests, is the inverse of a Bessel process with constant drift that appears as the heavy-traffic limit in queueing theory \citep{coffman1998polling}. To the best of our knowledge, the inverse Bessel process has not been used or analyzed before, but it is worth comparing the performance in different situations with other non-affine models, for it has a notable feature that the diffusion term has a power of 2.

The idea of introducing a partially observable state variable was motivated by \citet{lee2016pairs}. They assumed that the drift of a riskier asset among two co-integrated assets is unobservable, and then solved the portfolio optimization problem for investors who do not observe the drift. The Kalman-Bucy filter was the main tool for estimating the uncertain drift, and it also plays a similar role in this work. The Kalman-Bucy filter, widely used in various fields, provides a powerful way to solve linear filtering problem. One commonly cited reference for the Kalman-Bucy filter is \citet{kallianpur1980stochastic}, which includes general stochastic filtering theory.

\section{Problem settings and results} \label{sec:setup}

\subsection{Model setup}
Let $(\Omega,\mathcal{F},\{\mathcal{F}_t\}_{t\geq 0},\mathbb{P})$ be a filtered probability space. We consider an incomplete market having one safe asset and \textit{n} risky assets with \textit{d} state variables driving the risk-free rate \textit{r}, the excess returns $\mu$, and the volatility matrix $\Sigma$, which are all continuous. The $\mbp$-dynamics of the adapted state variables $Y=(Y_1,\ldots,Y_d)'$ is given by
\begin{equation} \label{state_var}
dY_t=b(Y_t)\,dt+\sigma(Y_t)\,dB_t
\end{equation}
with a state space $E\subset\mbr^d$, where $B=(B^1,\ldots,B^d)'$ is a \textit{d}-dimensional Brownian motion, $b:E\to\mbr^d$ and $\sigma:E\to\mbr^{d\times d}$ are continuously differentiable, $\sigma(y)$ is positive definite for each $y\in E$, and \textit{E} is open. Throughout this paper, the symbol $'$ denotes the transpose of a matrix. The dynamics of the safe asset $S^0$ and risky assets $S=(S^1,\ldots,S^n)'$ are given by
\begin{gather}
\frac{dS_t^0}{S_t^0}=r(Y_t)\,dt, \label{safe} \\
dS_t=\mbox{diag}(S_t)(r(Y_t)\underline{1}+\mu(Y_t))\,dt+\mbox{diag}(S_t)\Sigma(Y_t)\,dW_t\,, \label{risky}
\end{gather}
where diag($S_t$) stands for the diagonal matrix whose \textit{i}'s diagonal entry is $S_t^i$ for each $i=1,\ldots,n$, $\underline{1}$ represents the $n\times 1$ column vector whose entries are identically 1, and $W=(W^1,\ldots,W^n)'$ is an \textit{n}-dimensional Brownian motion. The Brownian motions \textit{W} and \textit{B} are correlated, and $\rho$ denotes the $n\times d$ cross-correlation matrix, that is,
\begin{equation} \label{cross_corr}
\langle W^j,B^l\rangle_t=\rho_{jl}t,\quad 1\leq j\leq n,\,1\leq l\leq d\,.
\end{equation}

Let $\pi=(\pi^1,\ldots,\pi^n)'$ denote the proportions of wealth in each risky asset. Then, the corresponding wealth process $X^{\pi}$ is given by
\begin{equation} \label{wealth}
\frac{dX_t^{\pi}}{X_t^{\pi}}=(r(Y_t)+\pi_t'\,\mu(Y_t))\,dt+\pi_t'\,\Sigma(Y_t)\,dW_t\,.
\end{equation}
For a fixed future time $T>0$ and current time $0\leq t<T$ with given values $Y_t=y$ and $X_t=x$, we aim to maximize the expected utility of wealth. We adopt investor's preference as the power utility
\begin{equation} \label{power_utility}
U(x)=\frac{x^p}{p},\quad p<1,\,p\neq 0\,.
\end{equation}
The value function \textit{V} is then written in the form
\begin{equation} \label{valuefunc}
V(t,x,y)=\sup_{(\pi_u)_{t\leq u\leq T}}\frac{1}{p}\,\EP[(X_T^{\pi})^p\,|\,X_t^{\pi}=x,Y_t=y]\,,
\end{equation}
and is defined on $[0,T]\times[0,\infty)\times E$. The corresponding Hamilton--Jacobi--Bellman (HJB) equation is
\begin{equation} \label{eqn:HJB}
\left\{\begin{array}{ll}
V_t+b(y)'\nabla_yV+\dfrac{1}{2}\text{Tr}\Bigl((\sigma\sigma')(y)\nabla_y^2V\Bigl)+r(y)xV_x\\ \qquad +\sup_{\pi}\left(x\,\pi'\Bigl(\mu(y) V_x+\Sigma(y)\rho\,\sigma(y)'\nabla_yV_x\Bigl)+\dfrac{x^2}{2}\pi'(\Sigma\Sigma')(y)\pi V_{xx}\right)=0 & \mbox{in }(0,T)\times[0,\infty)\times E \\ V(t,x,y)=\dfrac{x^p}{p} & \mbox{on } \{T\}\times[0,\infty)\times E
\end{array} \right. .
\end{equation}
Here the symbols $\nabla_y$ and $\nabla_y^2$ denote the gradient $\nabla_yV=(V_{y_1},\ldots,V_{y_d})'$ and Hessian matrix $\nabla_y^2V=(V_{y_{ij}})_{i,j=1}^{d}$ with respect to the vector \textit{y}, respectively, and ``Tr'' denotes the trace of a matrix. The supremum in the HJB equation can be achieved. Indeed, since the value function \textit{V} is concave in \textit{x}, the matrix $V_{xx}\Sigma\Sigma'$ is negative definite; hence, the HJB equation becomes
\begin{align}
V_t+&b(y)'\nabla_yV+\dfrac{1}{2}\text{Tr}\Bigl((\sigma\sigma')(y)\nabla_y^2V\Bigl)+r(y)xV_x\\ &-\dfrac{1}{2V_{xx}}\Bigl(\mu(y)V_x+\Sigma(y)\rho\,\sigma(y)'\nabla_yV_x\Bigl)'(\Sigma\Sigma')(y)^{-1}\Bigl(\mu(y)V_x+\Sigma(y)\rho\,\sigma(y)'\nabla_yV_x\Bigl)=0
\end{align}
in $[0,T)\times[0,\infty)\times E$, and the corresponding optimal portfolio is
\begin{equation} \label{eqn:optimal_port}
\pi(t,x,y\,;T)=-\frac{1}{xV_{xx}(t,x,y)}(\Sigma\Sigma')(y)^{-1}\Bigl(\mu(y)V_x(t,x,y)+\Sigma(y)\rho\,\sigma(y)'\nabla_yV_x(t,x,y)\Bigl)\,.
\end{equation}
As discussed in \citet{guasoni2015static}, we might set $V(t,x,y)=v(t,y)x^p/p$. Then the function \textit{v} satisfies the following PDE:
\begin{equation} \label{eqn:reduced}
\left\{\begin{array}{ll}
v_t+\Bigl(b(y)'-q(\Sigma(y)^{-1}\mu(y))'\rho\,\sigma(y)'\Bigl)\nabla v+\dfrac{1}{2}\text{Tr}\Bigl((\sigma\sigma')(y)\nabla^2v\Bigl)\\ \qquad+\left(p\,r(y)-\dfrac{q}{2}\lVert\Sigma(y)^{-1}\mu(y)\rVert^2\right) v-\dfrac{q}{2v}\lVert \rho\,\sigma(y)'\nabla v\rVert^2=0 & \mbox{in }(0,T)\times E \\ v=1 & \mbox{on } \{T\}\times E
\end{array} \right.,
\end{equation}
where $\lVert\cdot\rVert$ denotes the Euclidean norm on $\mathbb{R}^n$ or $\mathbb{R}^d$, and $q=p/(p-1)$. Accordingly, the optimal portfolio \eqref{eqn:optimal_port} can be represented by the function of variables \textit{t} and \textit{y}, as follows:
\begin{equation} \label{eqn:optimal_port2}
\hat{\pi}_T(t,y)=\frac{1}{1-p}\left(\left((\Sigma\Sigma')^{-1}\mu\right)(y)+(\Sigma(y)')^{-1}\rho\,\sigma(y)'\frac{\nabla v(t,y)}{v(t,y)}\right)\,.
\end{equation}
We shall now investigate the large-time behavior of \eqref{eqn:optimal_port2}. We will conduct the large-time sensitivity analysis of \eqref{eqn:optimal_port2} with respect to the perturbed state process.

In the one-dimensional case, that is, if the model has only one state variable, the problem can be reduced further. Put $\delta:=1/(1-q\rho'\rho)$. Substituting $u(t,y):=v(T-t,y)^{1/\delta}$, as in \citet{zariphopulou2001solution}, reduces the semilinear PDE \eqref{eqn:reduced} to a linear PDE:
\begin{equation} \label{eqn:1d:reduced}
\left\{\begin{array}{ll}
u_t=\Bigl(b(y)-q((\Sigma^{-1}\mu)(y))'\rho\,\sigma(y)\Bigl)u_y+\dfrac{1}{2}\sigma^2(y)u_{yy}\\ \hspace{4cm} +\dfrac{1}{\delta}\left(p\,r(y)-\dfrac{q}{2}\lVert(\Sigma^{-1}\mu)(y)\rVert^2\right)u & \mbox{in }(0,T)\times E \\ u=1 & \mbox{on } \{0\}\times E,
\end{array} \right.
\end{equation}
and the optimal portfolio \eqref{eqn:optimal_port2} is then
	\begin{equation} \label{eqn:optimal_port3}
	\hat{\pi}_T(t,y)=\frac{1}{1-p}\left((\Sigma\Sigma')^{-1}\mu\right)(y)+\frac{\delta}{1-p}(\Sigma(y)')^{-1}\rho\,\sigma(y)\frac{ u_y(T-t,y)}{u(T-t,y)}\,.
\end{equation}
In this paper, we primarily focus on the case where $d=1$.

\subsection{Hansen--Scheinkman decomposition}
Let \textit{u} be a function on $[0,\infty)\times E$ defined as
\begin{equation}\label{eqn:1d:FK_forward}
	u(t,z)=\EP_z\left[\exp{\left\{-\int_{0}^{t}\mathcal{V}(Z_s)\,ds\right\}}\right],
\end{equation}
where $\EP_z[\,\cdot\,]$ is the expectation under a probability measure $\mbp_z$ under which \textit{Z} satisfies
\begin{equation} \label{SDE:FK}
	dZ_t=\left(b-q(\Sigma^{-1}\mu)'\rho\sigma\right)\!(Z_t)\,dt+\sigma(Z_t)\,dB_t,\quad Z_0=z
\end{equation}
and takes values in \textit{E} for each $z\in E$ (here, weak existence and uniqueness in law of the SDE are implicitly assumed), and
\begin{equation}
	\mathcal{V}(z)=-\dfrac{1}{\delta}\left(p\,r(z)-\dfrac{q}{2}\lVert(\Sigma^{-1}\mu)(z)\rVert^2\right).
\end{equation}
It is natural to expect that \textit{u} is a solution to the PDE \eqref{eqn:1d:reduced}. The aim in this section is to express $u$ in the form
\begin{align} \label{eqn:u:decomp}
	u(t,z)=e^{-\lambda t}\phi(z)f(t,z) 
\end{align}
for some 
$\lambda\in\mbr,$   a   twice continuously differentiable positive  function $\phi$ and a continuous function $f.$ Such a type of decomposition is a crucial tool for finding a static optimal portfolio and the convergence rate. We follow the decomposition method developed by \citet{hansen2009long} to get the form \eqref{eqn:u:decomp}.

Define an operator
\begin{equation} \label{operator:H-S}
	\mathscr{L}:=\frac{1}{2}\sigma^2(z)\frac{d^2}{dz^2}+\left(b-q(\Sigma^{-1}\mu)'\rho\,\sigma\right)(z)\frac{d}{dz}-\mathcal{V}(z)\cdot.
\end{equation}
We say that a pair $(\lambda,\phi)$ of $\lambda\in\mbr$ and a positive and twice continuously differentiable function $\phi$ is an \textit{eigenpair} of $\mathscr{L}$ if
\begin{equation}
	\mathscr{L}\phi=-\lambda\phi\,.
\end{equation}
Suppose that such an eigenpair exists. Then it is straightforward to show that a process
\begin{equation} \label{eqn:H-S:martingale}
	M_t:=e^{\lambda t-\int_0^t\mathcal{V}(Z_u)du}\frac{\phi(Z_t)}{\phi(z)},\quad t\geq 0
\end{equation} 
is a positive local martingale such that $M_0=1$. If it is not only a local martingale but a true martingale, we can define a new probability measure $\tmbp_t$ on $\mathcal{F}_t$ by		
\begin{equation}
	\frac{d\tmbp_t}{d\mbp_z}=M_t
\end{equation}
for each $t\geq 0$, and $\tmbp_t$-dynamics of $(Z_s)_{0\leq s\leq t}$ is written as
\begin{equation} \label{eqn:1d:SDEphi}
	dZ_s=\left(b-q(\Sigma^{-1}\mu)'\rho\,\sigma+\sigma^2\frac{\phi'}{\phi}\right)(Z_s)\,ds+\sigma(Z_s)\,d\widetilde{B}_s,\,0<s\leq t,\quad Z_0=z,\quad\tmbp_t - a.s.
\end{equation}
where $(\widetilde{B}_s)_{0\leq s\leq t}$ is a $\tmbp_t$-Brownian motion with the Girsanov transformation
\[ \widetilde{B}_s=B_s-\int_{0}^{s}\left(\sigma\frac{\phi'}{\phi}\right)(Z_u)du,\,0\leq s\leq t.\]
Then \textit{u} can be decomposed into the form
\begin{align}
	u(t,z)&=\EP_z\left[\exp{\left\{-\int_{0}^{t}\mathcal{V}(Z_s)\,ds\right\}}\right]\\
	&=e^{-\lambda t}\phi(z)\mathbb{E}^{\tmbp_t}\left[\frac{1}{\phi(Z_t)}\right]\,.
\end{align}
Therefore, the decomposition \eqref{eqn:u:decomp} has been accomplished with
\begin{equation}
	f(t,z)=\mathbb{E}^{\tmbp_t}\left[\frac{1}{\phi(Z_t)}\right].
\end{equation}

As discussed in Appendix \ref{sec:consistency}, the family of probability measures $(\tmbp_t)_{t\geq0}$ is \textit{consistent}. Thus, we can use the universal notation $\tmbp_z$ to represent the family $(\tmbp_t)_{t\geq0}$. In the language of Definition \ref{defi:consistency:BM_SDE}, the process $\tB$ is a Brownian motion on the consistent probability space and \textit{Z} is a solution to the SDE
$$dZ_s=\left(b-q(\Sigma^{-1}\mu)'\rho\,\sigma+\sigma^2\frac{\phi'}{\phi}\right)(Z_s)\,ds+\sigma(Z_s)\,d\widetilde{B}_s,\quad Z_0=z,$$
on the consistent probability space $(\Omega,\mathcal{F},(\mathcal{F}_t)_{t\geq 0},(\tmbp_t)_{t\geq 0})$. Hereinafter, we work on the consistent probability space but use universal notation. Therefore, we denote \textit{f} and \textit{u} as
\begin{equation} \label{eqn:remainder}
	f(t,z)=\ETP_z\left[\frac{1}{\phi(Z_t)}\right],
\end{equation}
and
\begin{equation} \label{eqn:1d:u:HSdecomp}
	u(t,z)=e^{-\lambda t}\phi(z)f(t,z)\,,
\end{equation}
respectively.

\subsection{Dynamic and static fund separations}

We conduct the following three types of fund separation theorems for non-affine models. The main ideas in this section are indebted to \citet{guasoni2015static}. \\

(i) The optimal portfolio \eqref{eqn:optimal_port3} consists of three funds: the safe asset, the myopic portfolio $(\Sigma\Sigma')^{-1}\mu$, and the intertemporal hedging portfolio $(\Sigma')^{-1}\rho\,\sigma$. This is the classical \textit{dynamic} three-fund theorem. In this separation, the intertemporal weight $u_z/u$, which depends on both risk aversion, time, and the state variable, complicates an investor's trading strategy.  \\

(ii) We further refine the dynamic separation. Using the Hansen--Scheinkman decomposition \eqref{eqn:1d:u:HSdecomp}, the optimal portfolio can be further decomposed
into
\begin{equation}
	\hat{\pi}_{T}(t,z)=\frac{1}{1-p}((\Sigma\Sigma')^{-1}\mu)(z)+\frac{\delta}{1-p}((\Sigma')^{-1}\rho\,\sigma)(z)\frac{\phi'(z)}{\phi(z)}+\frac{\delta}{1-p}((\Sigma')^{-1}\rho\,\sigma)(z)\frac{f_z(t,z)}{f(t,z)}.
\end{equation}
We show that the intertemporal weight $f_z/f$ vanishes in the long-run. In more detail, there exist a positive constant $\hat{\lambda}$ and positive functions $C_1,C_2:E\to\mbr$ such that
\begin{equation} \label{convergence}
	C_1(z)e^{-\hlambda t} \leq \left|\frac{f_z(t,z)}{f(t,z)}\right|\leq C_2(z)e^{-\hlambda t}\,,\quad \forall t\in(0,T),\,\forall z\in E.
\end{equation}
Define a \textit{static} optimal portfolio as    
\begin{equation} \label{eqn:optimal_port_asymp}
	\hat{\pi}_{\infty}(z):=\frac{1}{1-p}\left((\Sigma\Sigma')^{-1}\mu\right)(z)+\frac{\delta}{1-p}(\Sigma(z)')^{-1}\rho\,\sigma(z)\frac{\phi'(z)}{\phi(z)}.
\end{equation}
Clearly, the inequalities in \eqref{convergence} shows that
\begin{equation} \label{convergence:optimal_port}
C_1(z)e^{-\hlambda(T-t)}\leq\left| \hat{\pi}_T(t,z)-\hat{\pi}_{\infty}(z)\right|\leq C_2(z)e^{-\hlambda(T-t)}
\end{equation}
for all $0\leq t<T$, which means that the optimal portfolio $\hat{\pi}_T(t,z)$ converges exponentially fast to $\hat{\pi}_{\infty}(z)$ as $T\to\infty$ with convergence rate $\hlambda$ for each $z\in E$(here we abuse the notations $C_1$ and $C_2$, which is definitely harmless).
\citet{guasoni2015static} also proved the convergence
$$\lim_{T\to\infty} \frac{u_z(t,z)}{u(t,z)}=\frac{\phi'(z)}{\phi(z)}$$ for two classes of affine models, but they did not provided the exponential convergence rate $\hlambda.$ \\

(iii) We compare 
fund separations of 
the dynamic portfolio and static optimal portfolio.
For two non-affine models, the log-derivative of $\phi$ has the decomposition of the form 
\begin{equation}  \label{eqn:log_phi':decomp}
	\frac{\phi'(z)}{\phi(z)}=\sum_{i=1}^{m}w_i(p)\psi_i(z),
\end{equation}
where $w_i$ is independent of \textit{z}  and $\psi_i$ is independent of \textit{p} for each $i=1,\ldots,m$.
This, together with the decomposition \eqref{eqn:1d:u:HSdecomp} of \textit{u}, shows that
\begin{equation}
	\frac{u_z(t,z)}{u(t,z)}=\sum_{i=1}^{m}w_i(p)\psi_i(z)+\frac{f_z(t,z)}{f(t,z)},
\end{equation}
and eventually we have
	\begin{align}
	\hat{\pi}_T(t,z)&=\frac{1}{1-p}\left((\Sigma\Sigma')^{-1}\mu\right)(z)+\sum_{i=1}^{m}\frac{\delta\,w_i(p)}{1-p}(\Sigma(z)')^{-1}\rho\,\sigma(z)\psi_i(z)+\frac{\delta}{1-p}(\Sigma(z)')^{-1}\rho\,\sigma(z)\frac{f_z(T-t,z)}{f(T-t,z)}\\
	&=\hat{\pi}_{\infty}(z)+\frac{\delta}{1-p}(\Sigma(z)')^{-1}\rho\,\sigma(z)\frac{f_z(T-t,z)}{f(T-t,z)}. \label{port:optimal:asymp_remain}
\end{align}
Now the dynamic optimal portfolio consists of the safe asset, the myopic portfolio $(\Sigma\Sigma')^{-1}\mu$, the \textit{m} time-independent hedging portfolios $\left((\Sigma')^{-1}\rho\,\sigma\psi_i\right)(z),\,i=1,\ldots,m$, and the intertemporal hedging portfolio $\left((\Sigma')^{-1}\rho\,\sigma\right)(z)$ whose intertemporal weight $f_z/f$ vanishes as terminal time \textit{T} goes to infinity. This result leads to \textit{static fund separation}, which implies that for long-term investments, the funds are preference free and only depend on the state variable, while their weights depend on each investor's risk aversion and are independent of the state variable.

Static fund separation is a major result of \citet{guasoni2015static}, a study that has greatly inspired this work. The result of our static fund separation is the same as theirs; however, compared to their approach, we also find the expression \eqref{port:optimal:asymp_remain} that shows that the optimal portfolio can be explicitly separated into the static optimal portfolio and intertemporal hedging portfolio vanishing at infinity. Moreover, we prove static fund separation theorems corresponding to non-affine models, whereas \citet{guasoni2015static} presented the theorems for affine models only.

\subsection{Long-term stability under parameter perturbations} \label{subsec:stability}

The other main result of this paper is that it demonstrates the stability of the optimal portfolio under parameter perturbations in the long run. 
Such an analysis is important for two reasons.
First, model parameters estimated from market data can have errors from true parameter values.
Second, the parameters change as time flows. 
Today's model parameters obtained from market will be 
different from yesterday's.
Hence, the static optimal portfolio would be less useful for the long-term investment 
if the parameter stability is not guaranteed.

Stability of the optimal portfolio under parameter perturbations can be confirmed by showing that sensitivities of the optimal portfolio with respect to small parameter perturbations converge as the terminal time goes to infinity. To be precise, we shall show that
\begin{equation} \label{conv:sensitivity} \frac{\partial}{\partial\chi}\hat{\pi}_T(t,z)\xrightarrow{T\to\infty}\frac{\partial}{\partial\chi}\hat{\pi}_{\infty}(z),\quad\forall t\geq 0,\,\forall z\in E,
\end{equation}
where $\chi$ stands for one of the parameters in the state variable and $\hat{\pi}_{\infty}$ is the static optimal portfolio \eqref{eqn:optimal_port_asymp}. We also note from the decomposition \eqref{port:optimal:asymp_remain} that this is equivalent to the sensitivities of the intertemporal hedging portfolio with respect to model parameters of the state variable vanishing in the long-run, namely
\begin{equation}
	\frac{\partial}{\partial\chi}\left((\Sigma(z)')^{-1}\rho\,\sigma(z)\frac{f_z(t,z)}{f(t,z)}\right)\xrightarrow{T\to\infty} 0,\quad\forall t\geq 0,\,\forall z\in E.
\end{equation}
In addition, we shall provide a sharp bound on the rate of convergence for each parameter.

Two types of perturbations are considered in this paper. The first is the sensitivity with respect to the current state variable, which is the case of $\chi=z.$ We show that
\[ \left|\frac{\partial}{\partial z}\frac{f_z(t,z)}{f(t,z)}\right|\leq C(z)e^{-\hlambda t},\]
which directly derives (here, the positive function \textit{C} may vary from each inequality)
\begin{equation}
	\left|\frac{\partial}{\partial z}\hat{\pi}_T(t,z)-\frac{d}{dz}\hat{\pi}_{\infty}(z)\right|\leq C(z) e^{-\hlambda(T-t)}.
\end{equation}
The second case includes sensitivities with respect to drift and diffusion terms of the state process. In such cases where the parameter $\chi$ is included in the drift term or the diffusion term, we show that
$$\displaystyle \left|\frac{\partial}{\partial \chi}\hat{\pi}_T(t,z)-\frac{\partial}{\partial \chi}\hat{\pi}_{\infty}(z)\right|\leq C(z;\chi)(1+T)e^{-\hlambda(T-t)}.$$
Such type of differentiations are not straightforward, especially, the second case, which requires a more elaborate setting. The methods are described in detail in Appendix \ref{sec:diff}.

In conclusion, our analysis indicates that the intertemporal part of the optimal portfolio is highly insensitive to the model parameters and the current state variable in long-term investments. As a result, the optimal portfolio remains stable under parameter perturbations in the long-run.

\section{3/2 state process model} \label{sec:3/2}
This section deals with a market model with a 3/2 state process. We present the main results here, while the detailed proofs can be found in Appendix \ref{sec:3/2:append}. Recall that $S=(S^1,\ldots S^n)',\,\underline{1}=(1,\ldots,1)',\,q=p/(p-1)$, and $\delta=1/(1-q\rho'\rho)$.

Let a state variable \textit{Y} follow a solution to the following stochastic differential equation:
\begin{equation} \label{eqn:3/2}
dY_t=(b-aY_t)Y_tdt+\sigma Y_t^{3/2}dB_t\,,
\end{equation}
where $b,\sigma>0$ and $a>-\sigma^2/2$. We assume that the market model follows
\begin{align}
r(Y_t)&\equiv r, \\ dS_t&=\mbox{diag}(S_t)(r\underline{1}+\Sigma\mu Y_t)\,dt+\sqrt{Y_t}\,\mbox{diag}(S_t)\Sigma \,dW_t, \label{SDE:risky:3/2}\\ \frac{d\langle W,B\rangle_t}{dt}&=\rho\,,
\end{align}
where $\Sigma\in\mathbb{R}^{n\times n},\,\mu,\rho\in\mathbb{R}^n$ have constant entries and $\Sigma$ is invertible. We also set $p<0$ for the investor's preference. The PDE \eqref{eqn:1d:reduced} corresponding to this model is
\begin{equation} \label{eqn:3/2:PDE}
\left\{\begin{array}{ll}
u_t=\dfrac{1}{2}\sigma^2y^3u_{yy}+(by-\theta y^2)u_y+\dfrac{1}{\delta}\left(p\,r-\dfrac{q}{2}\lVert\mu\rVert^2y\right)u & \mbox{in } (0,T)\times(0,\infty) \\ u=1 & \mbox{on }\{0\}\times(0,\infty)\,,
\end{array} \right. 
\end{equation}
where $\theta:=a+q\sigma\rho'\mu$.

\begin{assume} \label{assume:3/2}
	$\theta=a+q\sigma\rho'\mu>-\sigma^2/2$.
\end{assume}
\noindent All below theorems are based on Assumption \ref{assume:3/2}.

\begin{thm} \label{thm:static_fund:3/2}
	Let Assumption \ref{assume:3/2} hold, and define a real number $\eta$ as
	\begin{equation} \label{const:lambda_eta:3/2}
		\eta:=-\left(\frac{1}{2}+\frac{\theta}{\sigma^2}\right)+\sqrt{\left(\frac{1}{2}+\frac{\theta}{\sigma^2}\right)^2\!+\frac{q\lVert\mu\rVert^2}{\delta\,\sigma^2}}.
	\end{equation}
	Then,
	\begin{enumerate}[label=(\roman*)]
		\item a function
		\begin{equation}
			u(t,z)=\EP_z\left[\exp{\left\{\int_{0}^{t}\dfrac{1}{\delta}\left(p\,r-\dfrac{q}{2}\lVert\mu\rVert^2Z_s\right)\,ds\right\}} \right]\, 
		\end{equation}
			is a $C^{1,2}((0,T)\times (0,\infty))$ solution to the PDE \eqref{eqn:3/2:PDE},
		where \textit{Z} is the unique strong solution to the SDE
		 \begin{equation}
			dZ_t=(b-\theta Z_t)Z_tdt+\sigma Z_t^{3/2}dB_t,\quad Z_0=z>0,\quad \mbp_z\mbox{-a.s.}
		\end{equation}
		\item the function $u=u(t,z)$ admits the Hansen--Scheinkman decomposition 
		\begin{equation}
			u(t,z)=e^{-(b\,\eta+pr/\delta)t}z^{-\eta}\ETP_z\left[Z_t^{\eta}\right],\quad (t,z)\in[0,T)\times(0,\infty),
		\end{equation}
		where $(Z_s)_{0\leq s\leq t}$ is the solution to the SDE 
		\begin{equation}
			dZ_s=\Bigl(b-(\theta+\sigma^2\eta)Z_s\Bigl)Z_s\,ds+\sigma Z_s^{3/2}d\widetilde{B}_s,\quad 0\leq s\leq t,\quad Z_0=z,\quad \mbp_z\mbox{-a.s.}
		\end{equation}
		for $\tmbp_z$-Brownian motion $(\tB_s)_{0\leq s\leq t},$
		\item the value function \textit{V} and the optimal portfolio $\hat{\pi}_T$ satisfy
		\begin{equation}
			V(t,x,z)=\frac{x^p}{p}u(T-t,z)^{\delta}=\frac{x^p}{p}e^{-(b\,\eta+pr/\delta)(T-t)}z^{-\eta}f(T-t,z),
		\end{equation}
		where $f(t,z)=\ETP_z\left[Z_{t}^{\eta}\right]$, and
		\begin{align} 
			\hat{\pi}_T(t,z)&=\frac{1}{1-p}(\Sigma')^{-1}\mu+\frac{\delta}{1-p}(\Sigma')^{-1}\rho\,\sigma z\frac{u_z(T-t,z)}{u(T-t,z)}\\
			&=\frac{1}{1-p}(\Sigma')^{-1}\mu-\frac{\delta\eta}{1-p}(\Sigma')^{-1}\rho\,\sigma+\frac{\delta}{1-p}(\Sigma')^{-1}\rho\,\sigma\frac{f_z(T-t,z)}{f(T-t,z)},\label{eqn:optimal_port:3/2}
		\end{align}
		\item the partial derivative $f_z(t,z)$ has the representation
		 \begin{equation}
			f_z(t,z)=\frac{\eta}{z^{3/2}}\ETP_z\left[Z_t^{\eta+1/2}\exp{\left\{-\left(\frac{\theta+\sigma^2\eta}{2}+\frac{3\sigma^2}{8}\right)\int_0^tZ_s\,ds-\frac{b}{2}t \right\}}\right]\,,
		\end{equation}
		\item the static optimal portfolio is given by
		\begin{equation} \label{port:asymp:3/2}
			\hat{\pi}_{\infty}(z)=\frac{1}{1-p}(\Sigma')^{-1}\mu-\frac{\delta\eta}{1-p}(\Sigma')^{-1}\rho\,\sigma,
		\end{equation}
		and there exist positive functions $C_1,C_2$ such that
		\begin{equation} \label{ineqn:optimal_conv:3/2}
			C_1(z)e^{-b(T-t)}\leq\left| \hat{\pi}_T(t,z)-\hat{\pi}_{\infty}(z)\right|\leq C_2(z)e^{-b(T-t)}\,,
		\end{equation}
		for any $T>0,0\leq t<T$, and $z\in (0,\infty)$,
		\item the dynamic fund separation holds with four funds: the safe asset, the myopic portfolio $(\Sigma')^{-1}\mu$, one static hedging portfolio $(\Sigma')^{-1}\rho\,\sigma$, and one intertemporal hedging portfolio $(\Sigma')^{-1}\rho\,\sigma$ with intertemporal weight $f_z(T-t,z)/f(T-t,z)$. The static fund separation holds with three funds: the safe asset, the myopic portfolio $(\Sigma')^{-1}\mu$, and one hedging portfolio $(\Sigma')^{-1}\rho\,\sigma$.
	\end{enumerate}
\end{thm}

\begin{remark} \label{rmk:3/2}
	\begin{itemize}
		\item[(i)]  Equations \eqref{eqn:optimal_port:3/2} and \eqref{port:asymp:3/2} show that the optimal portfolio is decomposed into the sum of static optimal and intertemporal hedging portfolios:
		\begin{equation}
			\hat{\pi}_T(t,z)=\hat{\pi}_{\infty}(z)+\frac{\delta}{1-p}(\Sigma')^{-1}\rho\,\sigma\frac{f_z(T-t,z)}{f(T-t,z)}.
		\end{equation}
		
		\item[(ii)] In this market model, both the myopic and static hedging portfolios are independent of the current state variable. This is not because the state variable follows the 3/2 process, but because the dynamics of the risky assets \textit{S} are modeled as \eqref{SDE:risky:3/2}. Thus, we can derive a case that each fund depends on the current state variable by simply adjusting the power of the state process in the dynamics of \textit{S}. For example, set
		\begin{equation}
			dS_t=\mbox{diag}(S_t)(r\underline{1}+\Sigma\mu Y_t^{3/2})\,dt+Y_t\,\mbox{diag}(S_t)\Sigma \,dW_t.
		\end{equation}
		Then, the PDE that \textit{u} solves is still \eqref{eqn:3/2:PDE}, and the optimal portfolio is given by
		\begin{equation}
				\hat{\pi}_T(t,z)=\frac{1}{1-p}(\Sigma')^{-1}\mu\sqrt{z}-\frac{\delta\eta}{1-p}(\Sigma')^{-1}\rho\,\sigma\sqrt{z}+\frac{\delta}{1-p}(\Sigma')^{-1}\rho\,\sigma\sqrt{z}\frac{f_z(T-t,z)}{f(T-t,z)},
		\end{equation}
		where $f(t,z)=\ETP_z\left[Z_{t}^{\eta}\right]$.
	\end{itemize}
\end{remark}

Note that the optimal portfolio \eqref{eqn:optimal_port:3/2} clearly depends on the parameters \textit{b,a}, and $\sigma$. To emphasize the dependence, we denote as
$$ \hat{\pi}_{\cdot}(t,z)=\hat{\pi}_{\cdot}(t,z\,;\,b,a,\sigma).$$ 
The following theorem establishes the convergence of partial derivatives of the optimal portfolio \eqref{eqn:optimal_port:3/2} with respect to $z,b,a,\sigma$, respectively, to those of the static optimal portfolio. As mentioned in Subsection \ref{subsec:stability}, this implies that the intertemporal component of the optimal portfolio is largely insensitive to both the model parameters and the current state variable in long-term investments. Consequently, the optimal portfolio remains stable under perturbations in the parameters over the long run.

\begin{thm} \label{thm:perturb:3/2}
	Let Assumption \ref{assume:3/2} hold. For the optimal portfolio \eqref{eqn:optimal_port:3/2} and static optimal portfolio \eqref{port:asymp:3/2}, there exists a positive function $C:(0,\infty)\to\mbr$ depending on $b,a,\sigma$ such that
	\begin{enumerate}[label=(\roman*)]
		\item $\displaystyle \left|\frac{\partial}{\partial z}\hat{\pi}_T(t,z\,;\,b,a,\sigma)-\frac{\partial}{\partial z}\hat{\pi}_{\infty}(z\,;\,b,a,\sigma)\right|=\left|\frac{\partial}{\partial z}\hat{\pi}_T(t,z\,;\,b,a,\sigma)\right|\leq C(z\,;\,b,a,\sigma)e^{-b(T-t)},$
		\item $\displaystyle \left|\frac{\partial}{\partial b}\hat{\pi}_T(t,z\,;\,b,a,\sigma)-\frac{\partial}{\partial b}\hat{\pi}_{\infty}(z\,;\,b,a,\sigma)\right|=\left|\frac{\partial}{\partial b}\hat{\pi}_T(t,z\,;\,b,a,\sigma)\right|\leq C(z\,;\,b,a,\sigma)(1+T)e^{-b(T-t)}$
		\item $\displaystyle \left|\frac{\partial}{\partial a}\hat{\pi}_T(t,z\,;\,b,a,\sigma)-\frac{\partial}{\partial a}\hat{\pi}_{\infty}(z\,;\,b,a,\sigma)\right|\leq C(z\,;\,b,a,\sigma)(1+T)e^{-b(T-t)},$
		\item $\displaystyle \left|\frac{\partial}{\partial\sigma}\hat{\pi}_T(t,z\,;\,b,a,\sigma)-\frac{\partial}{\partial\sigma}\hat{\pi}_{\infty}(z\,;\,b,a,\sigma)\right|\leq C(z\,;\,b,a,\sigma)(1+T)e^{-b(T-t)},$
	\end{enumerate}
	for any $T>0,0\leq t<T$, and $z\in (0,\infty).$
\end{thm}
\noindent In this model, the exponential convergence rate $\hlambda$ is \textit{b}. Compared to $\hlambda$ of the inverse Bessel state process model in the next section, this has a very simple form. In addition, \textit{b} does not affects the static optimal portfolio.

\begin{remark}
	The approach in this paper is also valid for the two affine models presented in \citet{guasoni2015static}. Thus, similar results
	(HS decomposition, dynamic/static fund separations, etc)
	given in Theorem \ref{thm:static_fund:3/2} and \ref{thm:perturb:3/2} hold for the two models therein.
\end{remark}

\section{Inverse Bessel state process model} \label{sec:invB}
Let a state variable \textit{Y} follow a solution to the following SDE:
\begin{equation} \label{eqn:invB}
dY_t=(b-aY_t)Y_t^2dt+\sigma Y_t^2dB_t\,,
\end{equation}
where $b,\sigma>0$, and $a>-\sigma^2/2$. For a given deterministic $Y_0>0$, we show that the SDE \eqref{eqn:invB} has a unique strong solution, and the solution remains in $(0,\infty)$ for all time $t\geq 0$.

We set the risk-free rate and volatility matrix as constants, while the excess returns are set to depend on the state variable \textit{Y}:
\begin{align}
r(Y_t)&\equiv r, \\ dS_t&=\mbox{diag}(S_t)(r\underline{1}+\Sigma\mu Y_t^2)\,dt+\mbox{diag}(S_t)\Sigma Y_t\,dW_t,\\ \frac{d\langle W,B\rangle_t}{dt}&=\rho\,,
\end{align}
where $\Sigma\in\mathbb{R}^{n\times n},\,\mu,\rho\in\mathbb{R}^n$ have constant entries and $\Sigma$ is invertible. We also set $p<0$ for the investor's preference. The PDE \eqref{eqn:1d:reduced} corresponding to this model is
\begin{equation} \label{eqn:PDE:invB}
\left\{\begin{array}{ll}
u_t=\dfrac{1}{2}\sigma^2y^4u_{yy}+(by^2-\theta y^3)u_y+\dfrac{1}{\delta}\left(p\,r-\dfrac{q}{2}\lVert\mu\rVert^2y^2\right)u & \mbox{in } (0,T)\times(0,\infty) \\ u=1 & \mbox{on }\{0\}\times(0,\infty)\, ,
\end{array} \right.
\end{equation}
where $\theta:=a+q\sigma\rho'\mu$.

\begin{assume} \label{assume:invB}
	$\theta=a+q\sigma\rho'\mu>-\sigma^2/2$.
\end{assume}
\noindent All below theorems are based on Assumption \ref{assume:invB}. Their proofs are provided in Appendix \ref{sec:invB:append}.

\begin{thm} \label{thm:static_fund:invB}
	Let Assumption \ref{assume:invB} hold, and define real numbers $\eta,\xi,\lambda$, and $\hlambda$ as
	\begin{equation}
		\eta:=\frac{-\left(\sigma^2/2+\theta\right)+\sqrt{\left(\sigma^2/2+\theta\right)^2+q\sigma^2\lVert\mu\rVert^2/\delta}}{\sigma^2},\,\xi:=\dfrac{b\,\eta}{\sigma^2(\eta+1)+\theta},\, \lambda:=-\dfrac{1}{2}\sigma^2\xi^2+b\,\xi-\dfrac{p\,r}{\delta},
	\end{equation}
	and
	\begin{equation} \label{const:hlambda:invB}
		\hlambda:=-\frac{1}{2}\sigma^2\left(\frac{b-\sigma^2\xi}{\theta+\sigma^2(\eta+2)}\right)^2+\frac{(b-\sigma^2\xi)^2}{\theta+\sigma^2(\eta+2)}>0.
	\end{equation}
	Then,
	\begin{enumerate}[label=(\roman*)]
		\item a function
		 \begin{equation}
			u(t,z)=\EP_z\left[\exp{\left\{\int_{0}^{t}\dfrac{1}{\delta}\left(p\,r-\dfrac{q}{2}\lVert\mu\rVert^2Z_s^2\right)\,ds\right\}} \right]\,,
		\end{equation}
		where \textit{Z} is the unique strong solution to the SDE
		\begin{equation}
			dZ_t=(b-\theta Z_t)Z_t^2dt+\sigma Z_t^2dB_t,\quad z>0,\quad \mbp_z\mbox{-a.s.}
		\end{equation}
		is a $C^{1,2}((0,T)\times (0,\infty))$ solution to the PDE \eqref{eqn:PDE:invB},
		
		\item the function $u=u(t,z)$ admits the Hansen--Scheinkman decomposition form
		\begin{equation}
			u(t,z)=e^{-\lambda t}z^{-\eta}e^{\xi/z}\ETP_z\left[Z_t^{\eta}e^{-\xi/Z_t}\right],\quad (t,z)\in[0,T)\times(0,\infty),
		\end{equation}
		where $\tmbp_z$-dynamics of $(Z_s)_{0\leq s\leq t}$ is given by
		\begin{equation}
			dZ_s=\Bigl((b-\sigma^2\xi)-(\theta+\sigma^2\eta)Z_s\Bigl)Z_s^2ds+\sigma Z_s^2d\widetilde{B}_s,\,0<s\leq t\,,\quad Z_0=z,
		\end{equation}
	
		\item the value function \textit{V} and optimal portfolio $\hat{\pi}_T$ satisfy
		\begin{equation}
			V(t,x,z)=\frac{x^p}{p}u(T-t,z)^{\delta}=\frac{x^p}{p}e^{-\lambda (T-t)}z^{-\eta}e^{\xi/z}f(T-t,z),
		\end{equation}
		where $f(t,z)=\ETP_z\left[Z_{t}^{\eta}e^{-\xi/Z_{t}}\right]$, and
		\begin{align} 
				\hat{\pi}_T(t,z)&=\frac{1}{1-p}\left((\Sigma')^{-1}\mu+\delta(\Sigma')^{-1}\rho\,\sigma z\frac{u_z(T-t,z)}{u(T-t,z)}\right)\\	&=\frac{1}{1-p}(\Sigma')^{-1}\mu-\frac{\delta\eta}{1-p}(\Sigma')^{-1}\rho\,\sigma-\frac{\delta\xi}{1-p}(\Sigma')^{-1}\rho\,\sigma\frac{1}{z}\\
				&\qquad\qquad\qquad+\frac{\delta}{1-p}(\Sigma')^{-1}\rho\,\sigma z\frac{f_z(T-t,z)}{f(T-t,z)}, \label{eqn:optimal_port:invB}
		\end{align}
		\item the partial derivative $f_z(t,z)$ has the representation
		\begin{equation}
			f_z(t,z)=\frac{1}{z^2}\ETP_z\left[\left(\eta Z_t^{\eta+1}+\xi Z_t^{\eta}\right)\exp{\left\{-\frac{\xi}{Z_t}-\left(\sigma^2(1+\eta)+\theta\right)\int_{0}^{t}Z_s^2ds\right\}}\right]\,,
		\end{equation}
	
		\item the static optimal portfolio is given by
		\begin{equation} \label{port:asymp:invB}
			\hat{\pi}_{\infty}(z)=\frac{1}{1-p}(\Sigma')^{-1}\mu-\frac{\delta\eta}{1-p}(\Sigma')^{-1}\rho\,\sigma-\frac{\delta\xi}{1-p}(\Sigma')^{-1}\rho\,\sigma\frac{1}{z},
		\end{equation}
		and there exist positive functions $C_1,C_2$ such that
		\begin{equation} \label{ineqn:optimal_conv:invB}
			C_1(z)e^{-\hlambda(T-t)}\leq\left| \hat{\pi}_T(t,z)-\hat{\pi}_{\infty}(z)\right|\leq C_2(z)e^{-\hlambda(T-t)}\,,
		\end{equation}
		for any $T>0,0\leq t<T$, and $z\in (0,\infty)$,
		\item the dynamic fund separation holds with five funds: the safe asset, the myopic portfolio $(\Sigma')^{-1}\mu$, two static hedging portfolios $(\Sigma')^{-1}\rho\,\sigma$ and $(\Sigma')^{-1}\rho\,\sigma/z$, and one dynamic hedging portfolio $(\Sigma')^{-1}\rho\,\sigma z$ with intertemporal weight $f_z(T-t,z)/f(T-t,z)$. The static fund separation holds with four funds: the safe asset, the myopic portfolio $(\Sigma')^{-1}\mu$, and two hedging portfolios $(\Sigma')^{-1}\rho\,\sigma$ and $(\Sigma')^{-1}\rho\,\sigma/z$.
	\end{enumerate}
\end{thm}
\begin{remark} \label{rmk:invB}
	As in case of the 3/2 state process model, the optimal portfolio is decomposed into
	\begin{equation}
		\hat{\pi}_T(t,z)=\hat{\pi}_{\infty}(z)+\frac{\delta}{1-p}(\Sigma')^{-1}\rho\,\sigma z\frac{f_z(T-t,z)}{f(T-t,z)},
	\end{equation}
	and the myopic portfolio is independent of the current state variable. On the contrary, there is one hedging portfolio that depends on the current state variable.
\end{remark}

As explained in the previous section, we also denote $\hat{\pi}_T(t,z)=\hat{\pi}_T(t,z\,;\,b,a,\sigma)$ and $\hat{\pi}_{\infty}(z)=\hat{\pi}_{\infty}(z\,;\,b,a,\sigma)$. The following theorem shows that the optimal portfolio in this market model is stable under parameter perturbations in long-term investments.
 
\begin{thm} \label{thm:perturb:invB}
	Let Assumption \ref{assume:invB} hold. For the optimal portfolio \eqref{eqn:optimal_port:invB}, static optimal portfolio \eqref{port:asymp:invB}, and $\hlambda$ defined as \eqref{const:hlambda:invB}, there exists a positive function $C:(0,\infty)\to\mbr$ depending on $b,a$, and $\sigma$ such that
	\begin{enumerate}[label=(\roman*)]
		\item $\displaystyle \left|\frac{\partial}{\partial z}\hat{\pi}_T(t,z\,;\,b,a,\sigma)-\frac{\partial}{\partial z}\hat{\pi}_{\infty}(z\,;\,b,a,\sigma)\right|\leq C(z\,;\,b,a,\sigma)e^{-\hlambda(T-t)},$
		\item $\displaystyle \left|\frac{\partial}{\partial b}\hat{\pi}_T(t,z\,;\,b,a,\sigma)-\frac{\partial}{\partial b}\hat{\pi}_{\infty}(z\,;\,b,a,\sigma)\right|\leq C(z\,;\,b,a,\sigma)(1+T)e^{-\hlambda(T-t)}$
		\item $\displaystyle \left|\frac{\partial}{\partial a}\hat{\pi}_T(t,z\,;\,b,a,\sigma)-\frac{\partial}{\partial a}\hat{\pi}_{\infty}(z\,;\,b,a,\sigma)\right|\leq C(z\,;\,b,a,\sigma)(1+T)e^{-\hlambda(T-t)},$
		\item $\displaystyle \left|\frac{\partial}{\partial\sigma}\hat{\pi}_T(t,z\,;\,b,a,\sigma)-\frac{\partial}{\partial\sigma}\hat{\pi}_{\infty}(z\,;\,b,a,\sigma)\right|\leq C(z\,;\,b,a,\sigma)(1+T)e^{-\hlambda(T-t)},$
	\end{enumerate}
	for any $T>0,0\leq t<T$, and $z\in (0,\infty).$
\end{thm}

\section{Filtered Ornstein-Uhlenbeck state process} \label{sec:FOU}
In reality, it is hard to observe the state variable directly. Instead, an investor can estimate the value of the state variable only through stock prices. Such real-world limitation leads us to consider the conditional mean
\begin{equation} \label{eqn:filtered_state}
	\hat{Y}_t:=\EP[Y_t|\mathcal{F}_t^S],\quad 0\leq t\leq T,
\end{equation}
where $(\mathcal{F}_t^S)_{0\leq t\leq T}$ denotes the filtration generated by $(S_t)_{0\leq t\leq T}$. Here $\hat{Y}_t$ is clearly the projection of $Y_t$ onto the subspace consisting of $\mathcal{F}_t^S$-adapted random variables.

Let the market follow
\begin{align}
	dY_t&=(b-aY_t)dt+\sigma dB_t\,, \label{eqn:OU} \\
	r(Y_t)&\equiv r, \\ dS_t&=\mbox{diag}(S_t)(r\underline{1}+\Sigma\mu Y_t)\,dt+\mbox{diag}(S_t)\Sigma \,dW_t, \label{eqn:stock:OU} \\ \frac{d\langle W,B\rangle_t}{dt}&=\rho\,,
\end{align}
where $\Sigma\in\mathbb{R}^{n\times n},\,\mu,\rho\in\mathbb{R}^n$ have constant entries and $\Sigma$ is invertible. We assume $p<0$, and that the investor cannot observe the state process \eqref{eqn:OU} directly but filters it through the dynamics of stock prices \eqref{eqn:stock:OU}. Then the dynamics of \eqref{eqn:filtered_state} in this market can be obtained by solving the linear filtering problem of finding the best estimate $\hat{Y}$ of \textit{Y} given the observations $dS_t^i/S_t^i,\,i=1,\ldots,n$. In particular, applying \citealp[Theorem 10.5.1]{kallianpur1980stochastic} directly shows that the Kalman filter $(\hat{Y}_t,P_t)$ is given by the solution to the following system of equations
\begin{align}
	d\hY_t&=(b-a\,\hY_t)dt+(P_t\,\mu'+\sigma\rho')\Sigma^{-1}d\nu_t, \label{eqn:FOU} \\
	\frac{dP_t}{dt}&=-2(a+\sigma\rho'\mu)P_t+\sigma^2(1-\lVert\rho\rVert^2)-\lVert\mu\rVert^2P_t^2, \label{eqn:cond_var}
\end{align}
where $\nu$ is the innovation process
\begin{equation}
	d\nu_t=\Sigma\,\mu(Y_t-\hY_t)dt+\Sigma\,dW_t,
\end{equation}
and $P_t$ is the conditional variance
\begin{equation}
	P_t=\EP[(Y_t-\hY_t)^2].
\end{equation}
The Riccati equation \eqref{eqn:cond_var} has a constant solution
\begin{equation}
	P_t\equiv\frac{-(a+\sigma\rho'\mu)+\sqrt{(a+\sigma\rho'\mu)^2+\sigma^2(1-\norm{\rho}^2)\norm{\mu}^2}}{\norm{\mu}^2}
\end{equation}
if we set the initial value $P_0$ as the constant in the equation. Since observations does not increase or decrease over time, it is reasonable to treat \textit{P} as a constant. Thus, in view of the investor, the market follows
\begin{align}
	d\hY_t&=(b-a\,\hY_t)dt+(P_0\,\mu'+\sigma\rho')\Sigma^{-1}d\nu_t\,,\\
	r(Y_t)&\equiv r, \\ dS_t&=\mbox{diag}(S_t)(r\underline{1}+\Sigma\mu Y_t)\,dt+\mbox{diag}(S_t)\Sigma \,dW_t, \\ \frac{d\langle W,\nu\rangle_t}{dt}&=\Sigma\,,
\end{align}
where
\begin{equation}
	P_0=\frac{-(a+\sigma\rho'\mu)+\sqrt{(a+\sigma\rho'\mu)^2+\sigma^2(1-\norm{\rho}^2)\norm{\mu}^2}}{\norm{\mu}^2}.
\end{equation}
The maximum expected utility of the investor is given by
\begin{equation}
	V(t,x,y)=\sup_{(\pi_u)_{t\leq u\leq T}}\frac{1}{p}\,\EP[(X_T^{\pi})^p\,|\,X_t^{\pi}=x,\hY_t=y]\,,
\end{equation}
and the corresponding HJB equation is
\begin{equation}
	\left\{\begin{array}{ll}
		V_t+(b-ay)'V_y+\dfrac{1}{2}\norm{P_0\mu+\sigma\rho}^2V_{yy}+rxV_x\\ \qquad +\sup_{\pi}\left(x\,\pi'\Bigl(\mu y V_x+\Sigma(P_0\mu+\sigma\rho)V_{xy}\Bigl)+\dfrac{x^2}{2}\pi'\Sigma\Sigma'\pi V_{xx}\right)=0 & \mbox{in }(0,T)\times[0,\infty)\times \mbr \\ V(t,x,y)=\dfrac{x^p}{p} & \mbox{on } \{T\}\times[0,\infty)\times \mbr.
	\end{array} \right. 
\end{equation}
Following the same procedure as in Section \ref{sec:setup}, we obtain the PDE corresponding to \eqref{eqn:1d:reduced}
\begin{equation} \label{eqn:PDE:FOU}
	\left\{\begin{array}{ll}
		u_t=\dfrac{1}{2}\norm{\theta}^2u_{yy}+(b-(a+q\theta'\mu)y)u_y+\dfrac{1}{\delta}\left(p\,r-\dfrac{q}{2}\lVert\mu\rVert^2y^2\right)u & \mbox{in } (0,T)\times\mbr \\ u=1 & \mbox{on }\{0\}\times\mbr\,.
	\end{array} \right. 
\end{equation}
and the optimal portfolio
\begin{equation}
	\hat{\pi}_T(t,y)=\frac{1}{1-p}(\Sigma')^{-1}\mu\,y+\frac{\delta}{1-p}(\Sigma')^{-1}\theta\frac{u_y(T-t,y)}{u(T-t,y)},
\end{equation}
where $\theta:=P_0\,\mu+\sigma\rho$.

\begin{assume} \label{assume:FOU}
	$a+q\theta'\mu>0$.
\end{assume}
\noindent Theorem \ref{thm:static_fund:FOU} and Theorem \ref{thm:perturb:FOU} below are based on Assumption \ref{assume:FOU}. Their proofs are provided in Appendix \ref{sec:FOU:append}.

\begin{thm} \label{thm:static_fund:FOU}
	Let Assumption \ref{assume:FOU} hold, and define real numbers $\eta,\xi,\lambda$, and $\hlambda$ as
	\begin{equation}
		\eta:=\frac{-(a+q\theta'\mu)+\sqrt{(a+q\theta'\mu)^2+\dfrac{q}{\delta}\norm{\theta}^2\norm{\mu}^2}}{2\norm{\theta}^2},\,\xi:=\dfrac{2b\,\eta}{a+q\theta'\mu+2\norm{\theta}^2\eta},\, \lambda:=\left(\eta-\frac{1}{2}\xi^2\right)\norm{\theta}^2+b\,\xi-\dfrac{p\,r}{\delta},
	\end{equation}
	and
	\begin{equation} \label{const:hlambda:FOU}
		\hlambda:=a+q\theta'\mu+2\norm{\theta}^2\eta.
	\end{equation}
	Then,
	\begin{enumerate}[label=(\roman*)]
		\item a function
		\begin{equation}
			u(t,z)=\EP_z\left[\exp{\left\{\int_{0}^{t}\dfrac{1}{\delta}\left(p\,r-\dfrac{q}{2}\lVert\mu\rVert^2Z_s^2\right)\,ds\right\}} \right]\, 
		\end{equation}
		is a $C^{1,2}((0,T)\times (0,\infty))$ solution to the PDE \eqref{eqn:PDE:FOU}
		where \textit{Z} is the unique strong solution to the SDE
		\begin{equation}
			dZ_t=(b-(a+q\theta'\mu) Z_t)dt+\norm{\theta}dB_t,\quad Z_0=z\in\mbr,\quad \mbp_z\mbox{-a.s.}
		\end{equation}
		\item the function $u=u(t,z)$ admits the Hansen--Scheinkman decomposition 
		\begin{equation}
			u(t,z)=e^{-\lambda\,t-\eta\,z^2-\xi\,z}\ETP_z\left[e^{\eta Z_t^2+\xi Z_t}\right],\quad (t,z)\in[0,T)\times\mbr,
		\end{equation}
		where $(Z_s)_{0\leq s\leq t}$ is the solution to the SDE 
		\begin{equation}
			dZ_s=(b-\norm{\theta}^2\xi-(a+q\theta'\mu+2\norm{\theta}^2\eta) Z_s)ds+\norm{\theta}d\tB_s,\quad 0\leq s\leq t,\quad Z_0=z,\quad \mbp_z\mbox{-a.s.}
		\end{equation}
		for $\tmbp_z$-Brownian motion $(\tB_s)_{0\leq s\leq t},$
		\item the value function \textit{V} and optimal portfolio $\hat{\pi}_T$ satisfy
		\begin{equation}
			V(t,x,z)=\frac{x^p}{p}u(T-t,z)^{\delta}=\frac{x^p}{p}e^{-\lambda(T-t)-\eta\,z^2-\xi\,z}f(T-t,z),
		\end{equation}
		where $f(t,z)=\ETP_z\left[e^{\eta Z_t^2+\xi Z_t}\right]$, and
		\begin{align} 
			\hat{\pi}_T(t,z)&=\frac{1}{1-p}(\Sigma')^{-1}\mu \,z+\frac{\delta}{1-p}(\Sigma')^{-1}\theta\frac{u_z(T-t,z)}{u(T-t,z)}\\
			&=\frac{1}{1-p}(\Sigma')^{-1}\mu z-\frac{2\delta\eta}{1-p}(\Sigma')^{-1}\theta z-\frac{\delta\,\xi}{1-p}(\Sigma')^{-1}\theta +\frac{\delta}{1-p}(\Sigma')^{-1}\theta\frac{f_z(T-t,z)}{f(T-t,z)},\label{eqn:optimal_port:FOU}
		\end{align}
		\item the partial derivative $f_z(t,z)$ has the representation
		\begin{equation}
			f_z(t,z)=e^{-\hlambda\,t}\ETP_z\left[(2\eta Z_t+\xi)e^{\eta Z_t^2+\xi Z_t}\right]\,,
		\end{equation}
		\item the static optimal portfolio is given by
		\begin{equation} \label{port:asymp:FOU}
			\hat{\pi}_{\infty}(z)=\frac{1}{1-p}(\Sigma')^{-1}\mu\,z-\frac{2\delta\eta}{1-p}(\Sigma')^{-1}\theta\,z-\frac{\delta\,\xi}{1-p}(\Sigma')^{-1}\theta,
		\end{equation}
		and there exist positive functions $C_1,C_2$ such that
		\begin{equation}
			C_1(z)e^{-\hlambda(T-t)}\leq\left| \hat{\pi}_T(t,z)-\hat{\pi}_{\infty}(z)\right|\leq C_2(z)e^{-\hlambda(T-t)}\,,
		\end{equation}
		for any $T>0,0\leq t<T$, and $z\in\mbr$,
		\item the dynamic fund separation holds with five funds: the safe asset, the myopic portfolio $(\Sigma')^{-1}\mu\,z$, two static hedging portfolios $(\Sigma')^{-1}\theta\,z$ and $(\Sigma')^{-1}\theta$, and one intertemporal hedging portfolio $(\Sigma')^{-1}\theta$ with intertemporal weight $f_z(T-t,z)/f(T-t,z)$. The static fund separation holds with four funds: the safe asset, the myopic portfolio $(\Sigma')^{-1}\mu\,z$, and two hedging portfolios $(\Sigma')^{-1}\theta\,z$ and $(\Sigma')^{-1}\theta$.
	\end{enumerate}
\end{thm}

\begin{remark}
	\citet{guasoni2015static} also treat an Ornstein-Uhlenbeck state process as an example. Their $\psi_i$s in the decomposition \eqref{eqn:log_phi':decomp} are exactly the same as ours. Moreover, if we assume that the state variable is fully observable and adjust the two market models to be identical, it can be shown that the results of static fund separation obtained from the two different approaches are the same.
\end{remark}

\begin{thm} \label{thm:perturb:FOU}
	Let Assumption \ref{assume:FOU} hold. For the optimal portfolio \eqref{eqn:optimal_port:FOU}, static optimal portfolio \eqref{port:asymp:FOU}, and $\hlambda$ defined as \eqref{const:hlambda:FOU}, there exists a positive function $C:\mbr\to\mbr$ depending on $b,a$, and $\sigma$ such that
	\begin{enumerate}[label=(\roman*)]
		\item $\displaystyle \left|\frac{\partial}{\partial z}\hat{\pi}_T(t,z\,;\,b,a,\sigma)-\frac{\partial}{\partial z}\hat{\pi}_{\infty}(z\,;\,b,a,\sigma)\right|\leq C(z\,;\,b,a,\sigma)e^{-2\hlambda(T-t)},$
		\item $\displaystyle \left|\frac{\partial}{\partial b}\hat{\pi}_T(t,z\,;\,b,a,\sigma)-\frac{\partial}{\partial b}\hat{\pi}_{\infty}(z\,;\,b,a,\sigma)\right|\leq C(z\,;\,b,a,\sigma)(1+\sqrt{T})e^{-\hlambda(T-t)}$
		\item $\displaystyle \left|\frac{\partial}{\partial a}\hat{\pi}_T(t,z\,;\,b,a,\sigma)-\frac{\partial}{\partial a}\hat{\pi}_{\infty}(z\,;\,b,a,\sigma)\right|\leq C(z\,;\,b,a,\sigma)(1+T)e^{-\hlambda(T-t)},$
		\item $\displaystyle \left|\frac{\partial}{\partial\sigma}\hat{\pi}_T(t,z\,;\,b,a,\sigma)-\frac{\partial}{\partial\sigma}\hat{\pi}_{\infty}(z\,;\,b,a,\sigma)\right|\leq C(z\,;\,b,a,\sigma)(1+T)e^{-\hlambda(T-t)},$
	\end{enumerate}
	for any $T>0,0\leq t<T$, and $z\in\mbr.$
\end{thm}

\section{Outline of proofs} \label{sec:outline}
This section describes the general approach to proofs of the main results. Detailed proofs for the three specific models in Sections \ref{sec:3/2}, \ref{sec:invB}, and \ref{sec:FOU} are provided in Appendices \ref{sec:3/2:append}, \ref{sec:invB:append}, and \ref{sec:FOU:append}, respectively. Clearly, it should be assumed that the SDE \eqref{SDE:FK} has a strong solution and pathwise uniqueness holds.

\subsection{Verification argument} \label{subsec:verification}
To justify our arguments, it is necessary to demonstrate that the function \textit{u} indeed satisfies
\begin{equation} \label{eqn:valfunc:verify}
	V(t,x,z)=\frac{x^p}{p}u(T-t,z)^{\delta}
\end{equation}
in $(0,T]\times[0,\infty)\times E$, where \textit{V} is the value function \eqref{valuefunc}, and the optimal portfolio is given by \eqref{eqn:optimal_port3}.
These can be shown based on the verification lemma presented in \citealp[Appendix A]{guasoni2015static}. The main condition to be focused upon is the well-posedness of an SDE
\begin{equation} \label{temp232}
	dZ_s=\left(b(Z_s)-q(\Sigma^{-1}\mu)(Z_s)'\rho\,\sigma(Z_s)+\sigma(Z_s)^2\frac{u_z(T-s,Z_s)}{u(T-s,Z_s)}\right)\,ds+\sigma(Z_s)\,dB_s,\,0\leq s\leq T,\quad Z_0=z
\end{equation}
in the weak sense. Once this condition turns out to be true, the verification is immediately established by \citealp[Lemma A.3]{guasoni2015static}.

Through the assumptions imposed on the drift and diffusion terms, the uniqueness immediately follows (e.g., \citealp[Theorem 1.12.1]{pinsky1995positive}); hence, it suffices to show the existence of a weak solution. The process of applying the Hansen--Scheinkman decomposition to \eqref{eqn:1d:FK_forward} shows that the SDE \eqref{eqn:1d:SDEphi} has a weak solution. Moreover, using \eqref{eqn:1d:u:HSdecomp}, we can directly see that
\begin{equation} \label{eqn:1d:deriv:logu}
	\frac{u_z(t,z)}{u(t,z)}=\frac{\phi'(z)}{\phi(z)}+\frac{f_z(t,z)}{f(t,z)}.
\end{equation}
Thus, if we can show that the Dol\'{e}ans-Dade exponential
\begin{equation}
	\mathcal{E}\left(\int_0^{\cdot}\sigma(Z_s)\frac{f_z(T-s,Z_s)}{f(T-s,Z_s)}\,d\widetilde{B}_s\right)_t
\end{equation}
is a martingale, then Girsanov's theorem and identity \eqref{eqn:1d:deriv:logu} will lead to the weak existence of \eqref{temp232}. This can be easily verified by showing that the process $(f(T-t,Z_t))_{0\leq t\leq T}$ is a martingale and
\begin{equation}\label{temp:verify}
	\frac{f(T-t,Z_t)}{f(T,z)}=\mathcal{E}\left(\int_0^{\cdot}\sigma(Z_s)\frac{f_z(T-s,Z_s)}{f(T-s,Z_s)}\,d\widetilde{B}_s\right)_t.
\end{equation}
Indeed, by the time-homogeneous Markov property
\begin{equation}
	f(T-t,Z_t)=\ETP_{Z_t}\left[\frac{1}{\phi(Z_{T-t})}\right]=\ETP_z\left[\frac{1}{\phi(Z_{T})}\biggl|\mathcal{F}_t\right],
\end{equation}
the process $(f(T-t,Z_t))_{0\leq t\leq T}$ is a martingale. In addition, using It\^{o}'s formula we have 
\begin{equation}
	df(T-t,Z_t)=\sigma(Z_t)f_z(T-t,Z_t)\,d\tB_t=\sigma(Z_t)\frac{f_z(T-t,Z_t)}{f(T-t,Z_t)}f(T-t,Z_t)\,d\tB_t,
\end{equation}
which immediately shows \eqref{temp:verify}.

\subsection{Convergence of optimal portfolio}
The decomposition \eqref{eqn:log_phi':decomp} is not a general result but can be directly seen for each non-affine model to be introduced. Thus, here we focus on a methodology of proving \eqref{convergence} and, hence, completing \eqref{convergence:optimal_port}. We note from \eqref{eqn:1d:deriv:logu} that showing \eqref{convergence} is equivalent to proving
\begin{equation}
	C_1(z)e^{-\hlambda t}\leq\left|\frac{f_z(t,z)}{f(t,z)}\right|\leq C_2(z)e^{-\hlambda t},\quad\forall t\in(0,T),\,\forall z\in E.
\end{equation}
If the process \eqref{eqn:1d:SDEphi} is positive recurrent and if $1/\phi$ is integrable with respect to the invariant measure of \textit{Z}, then the ergodic theorem ensures that $f(t,z)$ converges to 0 (see Appendix \ref{sec:recur}). Next, using Corollary \ref{cor:perturb:flow}, the partial derivative $f_z(t,z)$ is given by
\begin{align}
	f_z(t,z)&=-\ETP_z\left[\frac{\partial Z_t^z}{\partial z}\frac{\phi'(Z_t)}{\phi(Z_t)^2}\right],\\
	&=-\ETP_z\left[e^{\int_0^t\!\left(b-q(\Sigma^{-1}\mu)'\rho\,\sigma+\sigma^2\frac{\phi'}{\phi}\right)'\!(Z_s)-\frac{1}{2}\sigma'(Z_s)^2ds+\int_0^t\sigma'(Z_s)\,d\widetilde{B}_s}\frac{\phi'(Z_t)}{\phi(Z_t)^2}\right]\,,
\end{align}
and this may be written in the form
\begin{equation} \label{eqn:1d:f_z}
	f_z(t,z)=g(z)\EQ_z\left[\exp{\left\{-\int_0^t\mathcal{W}(Z_s)\,ds\right\}}k(Z_t)\right] 
\end{equation}
for some measurable functions $\mathcal{W},k,g:E\to\mbr$ and some probability measure $\mbq_z$ (e.g., one way to write $f_z$ in the form \eqref{eqn:1d:f_z} is changing the measure via $d\mbq_z/d\tmbp_z=\mathcal{E}\left(\int_0^{\cdot}\sigma'(Z_s)\,d\tB_s\right)_t$). We apply Hansen--Scheinkman decomposition to \eqref{eqn:1d:f_z}. Let $(\hlambda,\hat{\phi})$ be an eigenpair of the operator
\begin{equation} \label{operator:1d:Lphi}
	\mathscr{L}^{\phi}:=\frac{1}{2}\sigma^2(z)\frac{d^2}{dz^2}+\left(b-q(\Sigma^{-1}\mu)'\rho\,\sigma+\sigma\frac{\phi'}{\phi}\right)(z)\frac{d}{dz}-\mathcal{W}(z)\cdot.
\end{equation}
Applying the Hansen--Scheinkman decomposition to $f_z$ yields the form
\begin{equation} \label{eqn:1d:remainder_z}
	f_z(t,z)=e^{-\hat{\lambda} t}g(z)\hat{\phi}(z)\EHP_{z}\left[\frac{k(Z_t)}{\hat{\phi}(Z_t)}\right]\,.
\end{equation}
If the solution to the SDE
\begin{equation}
	dZ_t=\left(b-q(\Sigma^{-1}\mu)'\rho\,\sigma+\sigma^2\left(\frac{\phi'}{\phi}+\frac{\hat{\phi}'}{\hat{\phi}}\right)\right)(Z_t)\,dt+\sigma(Z_t)\,d\widehat{B}_t,\quad Z_0=z
\end{equation}
is positive recurrent, then by the ergodic property, the remainder term($\EHP_z[\cdot]$) converges to some positive constant for each $z\in E$. If, in addition, $f_z$ never touches zero, then $e^{\hlambda t}f_z(t,z)$ is positive on $[0,\infty)\times E$ and bounded away from zero for each $z\in E$ by continuity. Therefore, there exist positive functions $C_1,C_2:E\to\mbr$ such that
\begin{equation} \label{ineqn:optimal_conv:estm}
	C_1(z)e^{-\hlambda t} \leq \left|\frac{f_z(t,z)}{f(t,z)}\right|\leq C_2(z)e^{-\hlambda t}\,.
\end{equation}
If $\hlambda>0$, we conclude that the optimal portfolio converges to the static optimal portfolio exponentially fast with convergence rate $\hlambda$ for each $z\in E$.

\begin{remark} \label{rmk:1d:positive_recur}
	\begin{itemize}
		\item[(i)] It appears that $\hlambda$ may not be positive. However, it is known that if $\mathcal{W}\geq 0$ and $\mathcal{W}\not\equiv0$, then $\hlambda>0$. We refer to \citealp[Theorem 4.2.1 (i),(iii), Proposition 4.2.2, and Theorem 4.3.2]{pinsky1995positive} for this fact.
		
		\item[(ii)] An eigenpair letting the process under the transformed measure be positive recurrent is unique if it exists\citep[Theorem 4.3.2]{pinsky1995positive}. We call the eigenpair the \textit{positive recurrent eigenpair}.
	\end{itemize}
	
\end{remark}

\subsection{Long-term stability under parameter perturbations}

\subsubsection{Sensitivity with respect to current state variable}

Since
\begin{equation}
	\frac{\partial}{\partial z}\left(\frac{f_z(t,z)}{f(t,z)}\right)=\frac{f_{zz}(t,z)}{f(t,z)}-\left(\frac{f_z(t,z)}{f(t,z)}\right)^2
\end{equation}
and
\begin{equation}
	\left|\frac{f_z(t,z)}{f(t,z)}\right|^2\leq C^2(z)e^{-2\hlambda t}
\end{equation}
by \eqref{ineqn:optimal_conv:estm}, all that remains is to determine an upper bound of $|f_{zz}(t,z)/f(t,z)|$ converging to 0 as $t\to\infty$. This can be done by differentiating $f_z(t,z)$ with respect to \textit{z} once again:
\begin{equation} \label{eqn:initial:temp}
	f_{zz}(t,z)=e^{-\hat{\lambda} t}\left(\frac{d}{dz}\left(g(z)\hat{\phi}(z)\right)\EHP_{z}\left[\frac{h(Z_t)}{\hat{\phi}(Z_t)}\right]+g(z)\hat{\phi}(z)\frac{\partial}{\partial z}\EHP_{z}\left[\frac{h(Z_t)}{\hat{\phi}(Z_t)}\right]\right)\,.
\end{equation}
The partial derivative of $\EHP_{z}\left[h(Z_t)/\hat{\phi}(Z_t)\right]$ appearing in \eqref{eqn:initial:temp} can be handled as in the method in which $f_z$ is transformed into the form \eqref{eqn:1d:remainder_z}. Then, as a consequence of ergodicity, the terms in the bracket in \eqref{eqn:initial:temp} are bounded for each \textit{z}. Now, it can be easily deduced that
\begin{equation}
	\left|\frac{\partial}{\partial z}\left(\frac{f_z(t,z)}{f(t,z)}\right)\right|\leq C(z)e^{-\hlambda t}
\end{equation}
for some positive function $C:E\to\mbr$. Therefore, the inequality
\begin{equation}
	\left|\frac{\partial}{\partial z}\hat{\pi}_T(t,z)-\frac{d}{dz}\hat{\pi}_{\infty}(z)\right|\leq C(z)e^{-\hlambda (T-t)}
\end{equation}
is established for any $0\leq t<T$ and some positive function $C:E\to\mbr$.

\subsubsection{Sensitivity with respect to drift and diffusion terms}
Let us denote a parameter included in the drift or diffusion terms by $\chi$. We want to prove the convergence \eqref{conv:sensitivity} and find a reasonable bound of the convergence rate. To emphasize dependence on $\chi$, we rewrite the sate process \eqref{state_var} in the form
\begin{equation} 
	dY_t=b(Y_t;\chi)\,dt+\sigma(Y_t;\chi)\,dB_t.
\end{equation}
Let $I\subset\mbr$ be an open interval containing 0. For each $\epsilon\in I$, define perturbed drift and perturbed diffusion by $b\sube(y;\chi):=b(y;\chi+\epsilon)$ and $\sigma\sube(y;\chi):=\sigma(y;\chi+\epsilon)$, respectively. We assume that $\epsilon\mapsto b\supe$ and $\epsilon\mapsto\sigma\sube$ are differentiable at $\epsilon=0$ for each \textit{y} and $\chi$. Namely, $b(y;\chi)$ and $\sigma(y;\chi)$ are assumed to be differentiable with respect to $\chi$ for each \textit{y}. For simplicity, we omit $\chi$ when writing $b\sube$ and $\sigma\sube$.

For each $\epsilon\in I$, let $u\supe$ be a function on $[0,\infty)\times E$ defined as
\begin{equation}\label{eqn:FK:perturbed}
	u\supe(t,z)=\mathbb{E}_z^{\mbp\supe}\left[\exp{\left\{-\int_{0}^{t}\mathcal{V}(Z_s\supe)\,ds\right\}}\right],
\end{equation}
where $(Z\supe,\mbp\supe)$ satisfies the SDE
\begin{equation} \label{SDE:FK:perturbed}
	dZ_t\supe=\left(b\sube-q(\Sigma^{-1}\mu)'\rho\,\sigma\sube\right)\!(Z_t\supe)\,dt+\sigma\sube(Z_t\supe)\,dB_t\supe,\quad Z_0\supe=z,
\end{equation}
with state space \textit{E} (here, $B\supe$ is clearly a $\mbp\supe$-Brownian motion; the other objects constituting a weak solution such as sample space and sigma-algebra, are omitted, which is definitely harmless). The Hansen--Scheinkman decomposition applied to $u\supe$ yields
\begin{align} \label{eqn:u:HS:perturbed}
	u\supe(t,z)&=e^{-\lambda\sube t}\phi\sube(z)\ETPe_z\left[\frac{1}{\phi\sube(Z_t\supe)}\right]\\
	&=e^{-\lambda\sube t}\phi\sube(z)f\supe(t,z),
\end{align}
where $\lambda\sube,\phi\sube,\widetilde{\mbp}\supe,f\supe$ are clearly self-explanatory, and $\tmbp\supe$-dynamics of $(Z_s\supe)_{0\leq s\leq t}$ is
\begin{equation} \label{eqn:drift:SDEphi}
	dZ_s\supe=\left(b\sube-q(\Sigma^{-1}\mu)'\rho\,\sigma\sube+\sigma\sube^2\frac{\phi\sube'}{\phi\sube}\right)\!(Z_s\supe)\,ds+\sigma\sube(Z_s\supe)\,d\widetilde{B}_s\supe,\,0<s\leq t,\quad Z_0\supe=z\,.
\end{equation}
Then, the perturbed optimal portfolio $\hat{\pi}_T$ is written in the form
\begin{equation} \label{sens:optimal}
	\hat{\pi}_T\supe(t,z)=\frac{1}{1-p}\left((\Sigma\Sigma')^{-1}\mu\right)(z)+\frac{\delta}{1-p}(\Sigma(z)')^{-1}\rho\,\sigma\sube(z)\frac{ u_z\supe(T-t,z)}{u\supe(T-t,z)}\,,
\end{equation}
and the sensitivity with respect to $\chi$ can be expressed as
\begin{equation} \label{port:sensitivity}
	\frac{\partial}{\partial \chi}\hat{\pi}_T=\frac{\partial}{\partial\epsilon}\biggl|_{\epsilon=0}\!\hat{\pi}_T\supe.
\end{equation}
The perturbed static optimal portfolio and its sensitivity with respect to $\chi$ are also given similarly:
\begin{equation}
	\hat{\pi}_{\infty}\supe(z)=\frac{1}{1-p}\left((\Sigma\Sigma')^{-1}\mu\right)(z)+\frac{\delta}{1-p}(\Sigma(z)')^{-1}\rho\,\sigma\sube(z)\frac{\phi\sube'(z)}{\phi\sube(z)}.
\end{equation}
and
\begin{equation} \label{port:sensitivity}
	\frac{\partial}{\partial \chi}\hat{\pi}_{\infty}=\frac{\partial}{\partial\epsilon}\biggl|_{\epsilon=0}\!\hat{\pi}_{\infty}\supe.
\end{equation}
Therefore, it is evident that the followings are equivalent:
\begin{itemize}
	\item[(i)] $\displaystyle\frac{\partial}{\partial \chi}\hat{\pi}_T\to\frac{\partial}{\partial \chi}\hat{\pi}_{\infty}$ as $T\to\infty$.
	
	\item[(ii)] $\displaystyle\frac{\partial}{\partial\epsilon}\biggl|_{\epsilon=0}\hat{\pi}_T\supe\to\frac{\partial}{\partial\epsilon}\biggl|_{\epsilon=0}\hat{\pi}_{\infty}\supe$ as $T\to\infty$.
	
	\item[(iii)] $\displaystyle\de\frac{u_z\supe(t,z)}{u\supe(t,z)}\to\de\frac{\phi\sube'(z)}{\phi\sube(z)}$ as $t\to\infty$.
\end{itemize}
We show (iii). Direct calculation shows that
\begin{align}
	\left|\de\frac{u_z\supe(t,z)}{u\supe(t,z)}-\de\frac{\phi\sube'(z)}{\phi\sube(z)}\right|&=\left|\frac{f_z(t,z)}{f(t,z)}\right|\left|\de\log{f\supe(t,z)}-\de\log{|f_z\supe(t,z)|}\right|\\
	&\leq C(z)e^{-\hlambda t}\left(\left|\de\log{f\supe(t,z)}\right|+\left|\de\log{|f_z\supe(t,z)|}\right|\right)\,.
\end{align}
Thus, it remains to observe the growth rates of
\begin{equation}
	\de\log{f\supe(t,z)}\quad\mbox{and}\quad\de\log{|f_z\supe(t,z)|}\,,
\end{equation}
with respect to time \textit{t}. This can be done by exploiting Corollary \ref{cor:perturb:united}.

\begin{itemize}
	\item[(i)] Since $f(t,z)$ converges to some positive constant as $t\to\infty$ for each $z\in E$, the growth rate of (i) is the growth rate of
	\begin{equation}
		\de f_z\supe(t,z)=\de\ETPe_z\left[\frac{1}{\phi\sube(Z_t\supe)}\right]\,.
	\end{equation}
	In light of Corollary \ref{cor:perturb:united}, we have
\begin{align}
	\frac{\partial}{\partial\epsilon}\biggl|_{\epsilon=0}\mathbb{E}_z^{\tmbp\supe}\left[\frac{1}{\phi\sube(Z_t\supe)}\right]&=\ETP_{z}\!\left[\frac{1}{\phi(Z_t)}\int_0^t\left(\frac{\partial}{\partial\epsilon}\biggl|_{\epsilon=0}\!\tilde{b}\sube\right)\left(h(Z_s)\right)d\tB_s\right]+\ETP_{z}\left[\left(\frac{\partial }{\partial\epsilon}\biggl|_{\epsilon=0}\!\frac{1}{\phi\sube}\circ h\sube^{-1}\right)\left(h(Z_t)\right)\right] \\ &\quad+\left(\de\!\tilde{z}\supe\right)\ETP_{z}\!\left[\left(\frac{1}{\phi}\circ h^{-1}\right)'\left(h(Z_t)\right)\frac{\sigma(z)}{\sigma(Z_t)}\frac{\partial Z_t}{\partial z}\right]\,,
\end{align}
where
\begin{equation}
	\tilde{b}\sube=\left(\frac{b\sube}{\sigma\sube}+\sigma\sube\frac{\phi\sube'}{\phi\sube}-\frac{1}{2}\sigma\sube'\right)\circ h\sube^{-1},\quad h\sube(x):=\int_{\cdot}^x\frac{1}{\sigma\sube(y)}dy,\quad h=h_0,\quad\tilde{z}\supe=h\sube(z).
\end{equation}

	Now the growth rate can be obtained by observing estimates on the right-hand side.
	
	\item[(ii)] We express the function $f_z\supe(t,z)$ in the form
	\begin{equation}
		f_z\supe(t,z)=e^{-\hlambda\sube t}g\sube(z)\hat{\phi}\sube(z)\EHPe_{z}\!\left[\frac{k\sube(Z_t\supe)}{\hat{\phi}\sube(Z_t\supe)}\right]\,,
	\end{equation}
	as in \eqref{eqn:1d:remainder_z}, and take log-derivative at $\epsilon=0$:
	\begin{equation}
		\de\!\log{|f_z\supe(t,z)|}=-\de\hlambda\sube  t+\frac{\frac{\partial}{\partial\epsilon}\bigl|_{\epsilon=0}(g\sube(z)\hat{\phi}\sube(z))}{g(z)\hat{\phi}(z)}+\frac{\frac{\partial}{\partial\epsilon}\bigl|_{\epsilon=0}\EHPe_{z}\left[k\sube(Z_t\supe)/\hat{\phi}\sube(Z_t\supe)\right]}{\EHP_{z}\left[k(Z_t)/\hat{\phi}(Z_t)\right]}\,.
	\end{equation}
	As in (i), Corollary \ref{cor:perturb:united} can be applied to the second term of the right-hand side.
\end{itemize}
	Accordingly, we can expect to derive the inequality
\begin{equation}
	\left|\de\!\frac{u_z\supe(t,z)}{u\supe(t,z)}-\de\!\frac{\phi\sube'(z)}{\phi\sube(z)}\right|\leq C(z;\chi)(1+T)e^{-\hlambda t}\,,
\end{equation}
for some positive function $C:E\to\mbr$ depending on $\chi$.

\section{Conclusion}
This paper investigated dynamic and static fund separations and their stability for long-term investors with CRRA under markets with a single state variable. Two non-affine models and one partially observable affine model were adopted as models for state variables. For the 3/2 state process and the inverse Bessel state process, we showed that the optimal portfolio on a finite horizon has refined dynamic separation consisting of the safe asset, myopic portfolio, several preference-free static funds that depend on the risk aversion but not on the state variable, and one intertemporal portfolio.
Then, we proved that the intertemporal weight vanishes exponotially fast in the long run. In particular, the explicit value of convergence rate was obtained for each model. Thus, the optimal portfolio converges exponentially fast to the static optimal portfolio. The Hansen--Scheinkman decomposition was widely used in all steps. We also demonstrated the stability of the optimal portfolio under model parameter perturbations by showing that the sensitivities of the intertemporal portfolio with respect to parameter perturbations vanish exponentially fast over time.

For the partially observable Ornstein-Uhlenbeck state process, the role of the state process is replaced by a process which is the projection of the state process onto the space of processes adapted to a filtration generated by stock price processes. We proved refined dynamic separation theorem and static separation theorem for this model, too. We also got similar convergence and stability results as in the two non-affine classes of models. \\

\noindent\textbf{Acknowledgement.}\\ 
Hyungbin Park was supported by the National Research Foundation of Korea (NRF) grants funded by the Ministry of Science and ICT (No. 2017R1A5A1015626, No. 2018R1C1B5085491 and No. 2021R1C1C1011675) 
and the Ministry of Education  (No. 2019R1A6A1A10073437) through the Basic Science Research Program.
In addition, Hyungbin Park was supported by Research Resettlement Fund for the new faculty of Seoul National University, South Korea. 
Financial support from the Institute for Research in Finance and Economics of Seoul National University is gratefully acknowledged.

\appendix
\section{Consistent probability spaces} \label{sec:consistency}
This section reviews the notion of a consistent family of probability measures introduced by \cite{park2021convergence}. Let $(\Omega,\mathcal{F},(\mathcal{F}_t)_{t\geq 0},\mbp)$ be a filtered probability space. If there exists a positive martingale \textit{M} such that $M_0=1$, then we can define a probability measure $\tmbp_t$ on $\mathcal{F}_t$ by
\begin{equation}
	\frac{d\tmbp_t}{d\mbp}=M_t,
\end{equation}
for each $t\geq 0$. It can be easily shown that
\begin{equation}
	\tmbp_{t'}=\tmbp_t\quad \mbox{on}\,\,\mathcal{F}_t,
\end{equation}
whenever $0\leq t\leq t'$. Nevertheless, in general it may not be possible to extend the probability measures $(\tmbp_t)_{t\geq 0}$ consistently to a measure on $\mathcal{F}_{\infty}:=\bigvee_{t\geq 0}\mathcal{F}_t$, the sigma-algebra generated by the filtration $\mathop{(\mathcal{F}_t)}_{t\geq 0}$. In other words, there is no guarantee that there exists a measure $\tmbp$ on $\mathcal{F}_{\infty}$ such that
\begin{equation}
	\tmbp=\tmbp_t,\quad\mbox{on}\,\,\mathcal{F}_t,
\end{equation}
for all $t\geq 0$. This appears to cause a problem when dealing with a kind of long-run analysis which necessarily observes the whole infinite horizon $[0,\infty)$. 

Motivated by the previous discussion, we suggest a rather elaborate probability space. In the following setting, though we cannot define a universal probability measure on $\mathcal{F}_{\infty}$ extending the family $(\tmbp_t)_{t\geq 0}$, such a problem can be avoided.

\begin{defi} \label{defi:consistency}
	Let $(\Omega,\mathcal{F})$ be a measurable space and $(\mathcal{F}_t)_{t\geq 0}$ be a filtration. A family of probability measures $(\mbp_t)_{t\geq 0}$ is called \textit{consistent} if $\mbp_t$ is a probability measure on $\mathcal{F}_t$ for each $t\geq 0$ and if
	\begin{equation}
		\mbp_{t'}=\mbp_t\quad \mbox{on}\,\,\mathcal{F}_t,
	\end{equation}
	whenever $0\leq t\leq t'$. For a consistent family of probability measures $(\mbp_t)_{t\geq 0}$, a quadruple\\ $(\Omega,\mathcal{F},(\mathcal{F}_t)_{t\geq 0},(\mbp_t)_{t\geq 0})$ is called a \textit{consistent probability space}.
\end{defi}
\noindent Since the aforementioned family $(\tmbp_t)_{t\geq 0}$ is consistent, $(\Omega,\mathcal{F},(\mathcal{F}_t)_{t\geq 0},(\tmbp_t)_{t\geq 0})$ is a consistent probability space. Thus, we can abuse the notation $\tmbp$ as an alias of $\tmbp_t$ for any $t\geq0$, when it comes to integration of a $\mathcal{F}_t$-measurable random variable. That is, the notation $\tmbp$ can be universally used to express the expectation $\mathbb{E}^{\tmbp_t}[X]$ as follows:
$$\mathbb{E}^{\tmbp}[X]=\mathbb{E}^{\tmbp_t}[X],$$ 
for any $t\geq 0$ and $\mathcal{F}_t$-measurable random variable \textit{X}.

Certain key concepts in stochastic calculus theory such as a Brownian motion and stochastic differential equations (SDE), should be redefined in the environment of a given consistent probability space $(\Omega,\mathcal{F},(\mathcal{F}_t)_{t\geq 0},(\tmbp_t)_{t\geq 0})$.
\begin{defi} \label{defi:consistency:BM_SDE}
	Let $(\Omega,\mathcal{F},(\mathcal{F}_t)_{t\geq 0},(\tmbp_t)_{t\geq 0})$ be a consistent probability space.
	\begin{enumerate}[label=(\roman*)]
		\item A process $B=(B_t)_{t\geq 0}$ is called a Brownian motion on the consistent probability space if, for each $t\geq 0$, the process $(B_s)_{0\leq s\leq t}$ is a Brownian motion on the filtered probability space $(\Omega,\mathcal{F}_t,(\mathcal{F}_s)_{0\leq s\leq t},\mbp_t)$.
		\item Let a process \textit{B} be a Brownian motion on the consistent probability space. A process $Z=(Z_t)_{t\geq 0}$ is called a strong solution (respectively, weak solution) to the SDE
		\begin{equation}
			dZ_t=b(Z_t)\,dt+\sigma(Z_t)\,dB_t,\quad Z_0=z,
		\end{equation}
		on the consistent probability space if, for each $t\geq0$, the process $(Z_s)_{0\leq s\leq t}$ is a strong solution (respectively, weak solution) of the SDE 
		\begin{equation}
			dZ_s=b(Z_s)\,ds+\sigma(Z_s)\,dB_s,\,0\leq s\leq t,\quad Z_0=z,\quad \tmbp_t - a.s.
		\end{equation}
		on the filtered probability space $(\Omega,\mathcal{F}_t,(\mathcal{F}_s)_{0\leq s\leq t},\mbp_t)$.
		The notion of uniqueness (pathwise and in law) can be similarly defined.
	\end{enumerate}
\end{defi}

\section{Recurrence and Ergodicity} \label{sec:recur}
Let us briefly review the concept of recurrence. There are several different definitions of recurrence, but here we follow the definition of \citet{pinsky1995positive}. The diffusion process $\left(Z,(\mbp_z)_{z\in E}\right)$ satisfying
\begin{equation} \label{SDE:recur}
	dZ_t=b(Z_t)\,dt+\sigma(Z_t)\,dB_t,\quad Z_0=z,\quad \mbp_z\mbox{-a.s.}
\end{equation}
with state space $E\subset\mbr^d$ is called \textit{recurrent} if for all $y,z\in E$ and $r>0,\,\mbp_z\{\tau_{\bar{B}_r(y)}<\infty\}=1$, where $\bar{B}_r(y)=\{x\in E\,:\,|y-x|\leq r\}$ and $\tau_{\bar{B}_r(y)}=\inf{\{t\geq 0\,:\,Z_t\in\bar{B}_r(y) \}}$. In particular, in the one-dimensional case it is easy to see that \textit{Y} is recurrent if and only if for all $y,z\in E,\,\mbp_z\{ Z_t=y\mbox{ for some } t\in[0,\infty) \}=1$. Let us denote the generator corresponding to \textit{Y} by
\begin{equation} \label{operator:L0}
	\mathscr{L}_0:=\frac{1}{2}\sigma^2(z)\frac{d^2}{dz^2}+b(z)\frac{d}{dz}\,,
\end{equation}
and the formal adjoint of $\mathscr{L}_0$ by $$\widetilde{\mathscr{L}}_0:=\frac{1}{2}\frac{d^2}{dz}(\sigma^2(z)\cdot)-\frac{d}{dz}(b(z)\cdot).$$
If a diffusion process is recurrent, there exists a unique (up to positive multiples) positive and twice continuously differentiable function $\tilde{\phi}$ such that $\tilde{\phi}''$ is locally H\"{o}lder continuous on \textit{E} and $\widetilde{\mathscr{L}}_0\tilde{\phi}=0$\citep[Theorem 4.3.3.(v) and Theorem 4.3.4]{pinsky1995positive}. If, in addition, $\int_E\tilde{\phi}(z)dz<\infty$, then the diffusion process is called \textit{positive recurrent}. The positive recurrent processes are of our interest because of certain ergodic properties they have, which we now introduce briefly. Let $p(t,z,dy)$ be the transition measure of $Z_t,\,t\geq 0$. Then $\tilde{\phi}$ is an invariant density for $p(t,z,dy)$, that is, $\int_Ep(t,z,dy)\tilde{\phi}(y)dz=\tilde{\phi}(z)dz$ for all $t\geq 0$\citep[Theorem 4.8.6]{pinsky1995positive}. We can assume, by normalizing if necessary, that $\int_E\tilde{\phi}(z)dz=1$ and therefore $\tilde{\phi}$ becomes the invariant probability density. Then for $h\in L^1(\tilde{\phi})$, that is, $\int_E|h(z)|\tilde{\phi}(z)dz<\infty$, \begin{enumerate}
	\item $\EP_z[|h(Z_t)|]<\infty$ for all $z\in E$ and $t>0$, and
	\item $\lim_{t\to\infty}\EP_z[h(Z_t)]=\int_Eh(z)\tilde{\phi}(z)dz$.
\end{enumerate} 
(e.g., \citealp[Remark 4.2]{robertson2015large}).

In the one-dimensional case, there is a useful criterion for positive recurrence. Though this can be found in many literatures (e.g., \citealp[Corollary 5.1.11]{pinsky1995positive}), we record the criterion here for convenience.
\begin{prop} \label{prop:pos_recur}
	The diffusion process \eqref{SDE:recur} with state space $E=(\alpha,\beta),\,-\infty\leq\alpha<\beta\leq\infty$ is positive recurrent if and only if the following two conditions 
	\begin{itemize}
		\item[(i)] $\int_{\alpha}^{c}\exp{\left\{-\int_c^x\frac{2b}{\sigma^2}(y)dy\right\}}dx=\int_{c}^{\beta}\exp{\left\{-\int_c^x\frac{2b}{\sigma^2}(y)dy\right\}}dx=\infty$,
		\item[(ii)] $\int_{\alpha}^{\beta}\frac{1}{\sigma^2(x)}\exp{\left\{\int_c^x\frac{2b}{\sigma^2}(y)dy\right\}}dx<\infty$,
	\end{itemize}
	hold for $c\in(\alpha,\beta)$.
\end{prop}

\section{Differentiation with respect to parameters} \label{sec:diff}
This section provides methods for differentiating functions of the form
\begin{equation} \label{ftn:expec(init,para)}
	(z,\chi)\mapsto\EP_z[f(Z_t;\chi)]
\end{equation}
where \textit{Z} is a solution to the SDE
\begin{equation} 
	dZ_t=b(Z_t;\chi)\,dt+\sigma(Z_t;\chi)\,dB_t,\quad Z_0=z,\quad\mbp_z-a.s.,
\end{equation}
with respect to the parameters \textit{z} and $\chi$. First, the following proposition, which
is a special case of \citealp[Theorem \upperRomannumeral{5}.7.39]{protter2004stochastic}, furnishes a method of differentiating with respect to the initial value \textit{z}.

\begin{prop} \label{prop:perturb:flow}
	Let a process \textit{Z} with a state space $E\subset\mbr$ be a unique strong solution to the SDE
	\begin{equation} \label{SDE:perturb:flow}
		dZ_t=b(Z_t)\,dt+\sigma(Z_t)\,dB_t,\quad Z_0=z\in E\,,
	\end{equation}
	where \textit{b} and $\sigma$ have locally Lipschitz first derivatives on \textit{E}. Then, for almost all $\omega\in\Omega$ and for all $t<\tau(\omega)$, where $\tau:=\inf{\{s\geq 0\,:\,Z_s\not\in E\}}$, $z\mapsto Z_t(\omega)$ is continuously differentiable. Furthermore, the partial derivative $\partial Z_t/\partial z$ is given by
	\begin{equation} \label{deriv:perturb:flow}
		\frac{\partial Z_t}{\partial z}=\exp{\left\{\int_0^tb'(Z_s)-\frac{1}{2}\sigma'(Z_s)^2ds+\int_0^t\sigma'(Z_s)dB_s\right\}}\,.
	\end{equation}
\end{prop}
\noindent To emphasize the dependence on the initial value, we also denote \textit{Z} as $Z^z$. The following corollary is now nearly evident.
\begin{cor} \label{cor:perturb:flow}
	Let the assumptions in Proposition \ref{prop:perturb:flow} hold. Suppose $f:E\to\mbr$ is a continuously differentiable function such that 
	\begin{equation}
		\EP_z[|f(Z_t)|]<\infty
	\end{equation}
	and
	\begin{equation} \label{cond:cor:perturb:flow}
		\mbox{a family }\left(f'(Z_t^{z+\epsilon})\frac{\partial Z_t^{z+\epsilon}}{\partial z}\right)_{\epsilon\in I}\mbox{ is uniformly integrable,}
	\end{equation}
	for a sufficiently small open interval \textit{I} containing 0 so that $z+\epsilon\in E$ for all $\epsilon\in I$. Then,
	\begin{align}
		\frac{\partial}{\partial z}\EP_z[f(Z_t)]&=\EP_z\left[\frac{\partial Z_t^{z}}{\partial z}f'(Z_t^{z})\right]\\
		&=\EP_z\left[\exp{\left\{\int_0^tb'(Z_s)-\frac{1}{2}\sigma'(Z_s)^2ds+\int_0^t\sigma'(Z_s)dB_s\right\}}f'(Z_t^{z})\right]\,.
	\end{align}
\end{cor}
\begin{proof}
	Consider a unique strong solution to the following SDE
	\begin{equation}
		dZ_t^{z+\epsilon}=b(Z_t^{z+\epsilon})\,dt+\sigma(Z_t^{z+\epsilon})\,dB_t,\quad Z_0^{z+\epsilon}=z+\epsilon\,.
	\end{equation}
	in the probability space $(\Omega,\mathcal{F},\mbp_z)$ with a Brownian motion \textit{B}.
	Due to uniqueness in law, $\mbp_{z+\epsilon}$-distribution of $Z$ and $\mbp_z$-distribution of $Z^{z+\epsilon}$ are identical. Thus,
	\begin{equation}
		\EP_{z+\epsilon}[f(Z_t)]=\EP_z[f(Z_t^{z+\epsilon})]\,,
	\end{equation}
	and
	\begin{equation}
		\frac{1}{\epsilon}\left(\EP_{z+\epsilon}[f(Z_t)]-\EP_{z}[f(Z_t)] \right)=\EP_z\left[\frac{f(Z_t^{z+\epsilon})-f(Z_t)}{\epsilon}\right]\,.
	\end{equation}
	Using condition \eqref{cond:cor:perturb:flow}, letting $\epsilon\to 0$ yields the desired result.
\end{proof}

For the differentiation with respect to $\chi$, mathematical setting for perturbation of the drift and diffusion terms of the state process should be made first. Let \textit{I} be an open interval containing 0. Let $b_{\cdot}(\cdot):I\times E\to\mbr$ and $\sigma_{\cdot}(\cdot):I\times E\to\mbr$ be continuously differentiable functions such that $b_0(z)=b(z)$ and $\sigma_0(z)=\sigma(z)$, and $\sigma_{\cdot}(\cdot)$ is positive. Suppose that the SDE
\begin{equation} \label{eqn:perturb:Yeps}
	dZ_t\supe=b\sube(Z_t\supe)\,dt+\sigma\sube(Z_t\supe)\,dB_t,\quad Z_0\supe=z\in E,\quad\mbp_z\supe\mbox{ - a.s.}
\end{equation}
has a unique strong solution for each $\epsilon\in I$ and $\mbp_z^0=\mbp_z$. We call $\left((Z\supe,\mbp_z\supe)\right)_{\epsilon\in I}$ a family of perturbed processes of \textit{Z} under $\mbp_z$.

For the family of perturbed processes $((Z\supe,\mbp_z\supe))_{\epsilon\in I}$ and a family of measurable functions $(f\supe)_{\epsilon\in I}$ on \textit{E} such that $\epsilon\mapsto f\supe(z)$ is differentiable for each $z\in E$, we develop a method for finding the derivative of $\epsilon\mapsto\mathbb{E}_z^{\mbp\supe}$ at $\epsilon=0$, namely
\begin{equation} \label{val:perturb}
	\frac{\partial}{\partial\epsilon}\biggl|_{\epsilon=0}\mathbb{E}_z^{\mbp\supe}\left[f\supe(Z_t\supe)\right].
\end{equation}
Then, the partial derivative of \eqref{ftn:expec(init,para)} with respect to $\chi$ follows. Indeed, by setting
\begin{equation}
	f\supe(z;\chi):=f(z;\chi+\epsilon),\quad b\sube(z;\chi):= b(z;\chi+\epsilon),\quad \sigma\sube(z;\chi):=\sigma(z;\chi+\epsilon),
\end{equation}
we have
\begin{equation}
	\frac{\partial}{\partial\chi}\EP_z[f(Z_t;\chi)]=\frac{\partial}{\partial\epsilon}\biggl|_{\epsilon=0}\mathbb{E}_z^{\mbp\supe}\left[f\supe(Z_t\supe)\right].
\end{equation}

The first step involves finding the value \eqref{val:perturb} in the case where perturbations only occur on the drift term. The following proposition is a special case of \citealp[Proposition A.1]{park2018sensitivity}. 
\begin{prop} \label{prop:perturb:drift}
	Let the SDE
	\begin{equation}
		dZ_t^{\epsilon}=b_{\epsilon}(Z_t^{\epsilon})\,dt+\sigma(Z_t^{\epsilon})\,dB_t,\, 	Z_0^{\epsilon}=z\in E,\quad\mbp_z\supe\mbox{ - a.s.}
	\end{equation}
	has a unique strong solution for each $(\epsilon,z)\in I\times E$, and $(f\supe)_{\epsilon\in I}$ be a family of measurable functions on \textit{E} such that $\epsilon\mapsto f\supe(z)$ is differentiable for each $z\in E$. Assume that there exist functions $g,\psi:E\to\mbr$ such that
	\begin{align}
		&\left|\sigma^{-1}\frac{\partial b\sube}{\partial\epsilon}\right|\leq g, \label{ineqn:b_eps} \\
		&\left|f\supe\right|\leq \psi\,, \label{ineqn:psi}
	\end{align}
	for every $\epsilon\in I$. Suppose that for each $t>0$ 
	\begin{equation} \label{cond:unif_int:thm:drift}
		\mbox{a family }\left(\frac{\partial }{\partial\epsilon}f\supe(Z_t)\right)_{\epsilon\in I} \mbox{ is uniformly integrable,}
	\end{equation}
	and there exist positive constants $\epsilon_0,\epsilon_1,p,q$ with $p\geq 2$ and $1/p+1/q=1$ such that 
	\begin{driftcondi}
		\item $\EP_z\left[e^{\epsilon_0\int_0^tg^2(Z_s)\,ds}\right]<\infty,$ \label{cond1:thm:drift}
		\item $\EP_z\left[\int_0^tg^{p+\epsilon_1}(Z_s)\,ds\right]<\infty,$ \label{cond2:thm:drift}
		\item $\EP_z\left[\psi^q(Z_t)\right]<\infty.$ \label{cond3:thm:drift}
	\end{driftcondi}
	Then, a function $\epsilon\mapsto\mathbb{E}_z^{\mbp\supe}[f\supe(Z_t^{\epsilon})]$ is continuously differentiable and 
	\begin{equation}
		\frac{\partial}{\partial\epsilon}\biggl|_{\epsilon=0}\mathbb{E}_z^{\mbp\supe}\left[f\supe(Z_t\supe)\right]=\EP_z\left[f(Z_t)\int_0^t\left(\frac{\partial}{\partial\epsilon}\biggl|_{\epsilon=0}\sigma^{-1}b\sube\right)(Z_s)dB_s\right]+\EP_z\left[\left(\frac{\partial }{\partial\epsilon}\biggl|_{\epsilon=0}f\supe\right)(Z_t)\right]\,.
	\end{equation}
\end{prop}
\noindent Proposition \ref{prop:perturb:drift} can also be extended to the general case using a transformation technique, which will be introduced.
\begin{cor} \label{cor:perturb:united}
	Let $\left((Z\supe,\mbp_z\supe)\right)_{\epsilon\in I}$ be a family of perturbed processes of \textit{Z} under $\mbp_z$, and $\sigma\sube(z)$ is positive on $I\times E$.
	Let us define
	\begin{equation}
		h\sube(x):=\int_{\cdot}^x\frac{1}{\sigma\sube(y)}dy,\quad\widetilde{Z}_t\supe:=h\sube(Z_t\supe).
	\end{equation}
	Suppose that all the assumptions in Proposition \ref{prop:perturb:drift} hold for $\widetilde{Z}$ and $f\supe\circ h\sube^{-1}$ replacing \textit{Z} and $f\supe$, respectively. Then, a function $\epsilon\mapsto\mathbb{E}_z^{\mbp\supe}[f\supe(Z_t^{\epsilon})]$ is continuously differentiable and 
	\begin{align}
		\frac{\partial}{\partial\epsilon}\biggl|_{\epsilon=0}\mathbb{E}_z^{\mbp\supe}[f\supe(Z_t^{\epsilon})]&=\EP_{z}\!\left[f(Z_t)\int_0^t\left(\frac{\partial}{\partial\epsilon}\biggl|_{\epsilon=0}\!\tilde{b}\sube\right)\left(h(Z_s)\right)dB_s\right]+\EP_{z}\left[\left(\frac{\partial }{\partial\epsilon}\biggl|_{\epsilon=0}\! f\supe\circ h\sube^{-1}\right)\left(h(Z_t)\right)\right] \\ &\quad+\left(\de\!\tilde{z}\supe\right)\EP_{z}\!\left[(f\circ h^{-1})'\left(h(Z_t)\right)\frac{\sigma(z)}{\sigma(Z_t)}e^{\int_0^tb'(Z_s)-\frac{1}{2}\sigma'(Z_s)^2ds+\int_0^t\sigma'(Z_s)dB_s}\right]\,,
	\end{align}
	where
	\begin{equation}
		\tilde{b}\sube=\left(\frac{b\sube}{\sigma\sube}-\frac{1}{2}\sigma\sube'\right)\circ h\sube^{-1},\quad\tilde{z}\supe=h\sube(z),\quad h=h_0.
	\end{equation}
\end{cor}

\begin{proof}
	Direct application of Ito formula shows that
	\begin{equation}
		d\widetilde{Z}_t\supe=\tilde{b}\sube(\widetilde{Z}_t\supe)\,dt+dB_t,\quad \widetilde{Z}_0\supe=\tilde{z}\supe\,.
	\end{equation}
Thus,
\begin{equation}
	\mathbb{E}_z^{\mbp\supe}\left[f\supe(Z_t\supe)\right]=\mathbb{E}_{\widetilde{Z}_0\supe=\tilde{z}\supe}^{\mbp\supe}\!\left[(f\supe\circ h\sube^{-1})(\widetilde{Z}_t\supe)\right]\,,
\end{equation}
and applying chain rule together with Corollary \ref{cor:perturb:flow} and Proposition \ref{prop:perturb:drift} yields
\begin{align}
	\frac{\partial}{\partial\epsilon}\biggl|_{\epsilon=0}&\mathbb{E}_z^{\mbp\supe}\left[f\supe(Z_t\supe)\right]=\frac{\partial}{\partial\epsilon}\biggl|_{\epsilon=0}\mathbb{E}_{\widetilde{Z}_0=\tilde{z}\supe}^{\mbp\supe}\!\left[(f\supe\circ h\sube^{-1})(\widetilde{Z}_t\supe)\right]\\
	&=\EP_{\widetilde{Z}_0=\tilde{z}}\!\left[(f\circ h^{-1})(\widetilde{Z}_t)\int_0^t\left(\frac{\partial}{\partial\epsilon}\biggl|_{\epsilon=0}\!\tilde{b}\sube\right)(\widetilde{Z}_s)dB_s\right]+\EP_{\widetilde{Z}_0=\tilde{z}}\left[\left(\frac{\partial }{\partial\epsilon}\biggl|_{\epsilon=0}\! f\supe\circ h\sube^{-1}\right)(\widetilde{Z}_t)\right] \\ &\quad+\left(\de\!\tilde{z}\supe\right)\EP_{\widetilde{Z}_0=\tilde{z}}\!\left[(f\circ h^{-1})'(\widetilde{Z}_t)\frac{\partial\widetilde{Z}_t}{\partial\tilde{z}}\right]\hspace{3.5cm} (\mbox{here, }\widetilde{Z}_t=\widetilde{Z}_t^0,\, \tilde{z}=\tilde{z}^0) \\
	&=\EP_{z}\!\left[f(Z_t)\int_0^t\left(\frac{\partial}{\partial\epsilon}\biggl|_{\epsilon=0}\!\tilde{b}\sube\right)\left(h(Z_s)\right)dB_s\right]+\EP_{z}\left[\left(\frac{\partial }{\partial\epsilon}\biggl|_{\epsilon=0}\! f\supe\circ h\sube^{-1}\right)\left(h(Z_t)\right)\right] \\ &\quad+\left(\de\!\tilde{z}\supe\right)\EP_{z}\!\left[(f\circ h^{-1})'\left(h(Z_t)\right)\frac{\sigma(z)}{\sigma(Z_t)}\frac{\partial Z_t}{\partial z}\right]\\
	&=\EP_{z}\!\left[f(Z_t)\int_0^t\left(\frac{\partial}{\partial\epsilon}\biggl|_{\epsilon=0}\!\tilde{b}\sube\right)\left(h(Z_s)\right)dB_s\right]+\EP_{z}\left[\left(\frac{\partial }{\partial\epsilon}\biggl|_{\epsilon=0}\! f\supe\circ h\sube^{-1}\right)\left(h(Z_t)\right)\right] \\ &\quad+\left(\de\!\tilde{z}\supe\right)\EP_{z}\!\left[(f\circ h^{-1})'\left(h(Z_t)\right)\frac{\sigma(z)}{\sigma(Z_t)}e^{\int_0^tb'(Z_s)-\frac{1}{2}\sigma'(Z_s)^2ds+\int_0^t\sigma'(Z_s)dB_s}\right].
\end{align}
\end{proof}
\noindent Such transformation of \textit{Z} into $\widetilde{Z}$ is called the \textit{Lamperti transformation}.

\section{3/2 state process model} \label{sec:3/2:append}

Consider the 3/2 process solving
\begin{equation} \label{SDE:append:3/2}
	dY_t=(b-aY_t)Y_tdt+\sigma Y_t^{3/2}dB_t,\quad y\in (0,\infty)\,,
\end{equation}
where $b,\sigma>0$ and $a>-\sigma^2/2$. The SDE \eqref{SDE:append:3/2} can be obtained by applying It\^{o} formula to the inverse of a CIR process. Thus, the 3/2 process is well-defined and remains in $(0,\infty)$ for all $t\geq 0$.

This section proves Theorem \ref{thm:static_fund:3/2} and \ref{thm:perturb:3/2} in Section \ref{sec:3/2}. First, we introduce three lemmas regarding the 3/2 process. They are not only used to prove the theorems but are also interesting in themselves.

\begin{lemma} \label{lem:diff:3/2}
	Let $Y^y$ be the solution to \eqref{SDE:append:3/2} with $Y_0=y$. Then for each $t\geq 0$, a function $y\mapsto Y_t^y$ is differentiable a.s., and
	\begin{equation} \label{eqn:deriv:3/2}
		\frac{\partial Y_t^y}{\partial y}=\left(\frac{Y_t^y}{y}\right)^{3/2}\exp{\left\{-\frac{b}{2}t-\left(\frac{a}{2}+\frac{3}{8}\sigma^2\right)\int_{0}^{t}Y_s^yds\right\}}\,.
	\end{equation}
\end{lemma}
\begin{proof}
	It is easy to check that all the conditions in Proposition \ref{prop:perturb:flow} are met with $b(y)=(b-ay)y,\,\sigma(y)=\sigma y^{3/2},$ and $E=(0,\infty)$. Thus, we can apply Proposition \ref{prop:perturb:flow} to derive that
	\begin{equation}
		\frac{\partial Y_t^y}{\partial y}=\exp{\left\{\int_0^tb-\left(2a+\frac{9}{8}\sigma^2\right)Y_s^yds+\frac{3}{2}\int_0^t\sigma \sqrt{Y_s^y}dB_s\right\}}\,.
	\end{equation}
	On the other hand, applying It\^{o} formula to $\log{Y_t^y}$ yields
	\begin{equation}
		Y_t^y=y\exp{\left\{\int_0^tb-\left(a+\frac{1}{2}\sigma^2\right)Y_s^yds+\int_0^t\sigma \sqrt{Y_s^y}dB_s\right\}}\,.
	\end{equation}
	Therefore, the result follows from the above two equations.
\end{proof}

\begin{lemma} \label{lem:expect_conti:3/2}
	Let \textit{Y} be a 3/2 process satisfying \eqref{SDE:append:3/2}. For any $0<\nu<2a/\sigma^2+2$, a map $s\mapsto\EP[Y_s^{\nu}]$ is continuous on $[0,\infty)$.
\end{lemma}
\begin{proof}
	The distribution of a 3/2 process is studied in \citet{ahn1999parametric}. From the literature, the continuity of $s\mapsto\EP_z[Z_s^{1+\epsilon_1}]$ on $(0,\infty)$ can be easily derived; however, the continuity at $s=0$ is not clear and must be proved.
	
	It follows from \citet{ahn1999parametric} that
	\begin{equation} \label{eqn:moment:append:3/2}
		\EP[Y_t^{\nu}]=\alpha_t^{\nu}e^{-\beta_t}\frac{1}{\Gamma(\nu)}\int_0^1e^{\beta_tx}x^{\kappa-\nu}(1-x)^{\nu-1}dx,\quad 0<\nu<\kappa+1\,,
	\end{equation}
	where
	\begin{equation}
		\alpha_t=\frac{2b}{\sigma^2(1-e^{-bt})},\quad\beta_t=\frac{\alpha_t}{y}e^{-bt},\quad\kappa=\frac{2a}{\sigma^2}+1\,.
	\end{equation}
	Fix $\epsilon\in(0,1)$ and rewrite \eqref{eqn:moment:append:3/2} as
	\begin{align}
		\EP[Y_t^{\nu}]&=\frac{y^{\nu}e^{\nu bt}}{\Gamma(\nu)}\left(\int_0^{1-\epsilon}\!\beta_t^{\nu}e^{-\beta_t(1-x)}x^{\kappa-\nu}(1-x)^{\nu-1}dx+\int_{1-\epsilon}^1\!(\beta_t(1-x))^{\nu-1}e^{-\beta_t(1-x)}x^{\kappa-\nu}\beta_tdx\right) \\
		&=:\frac{y^{\nu}e^{\nu bt}}{\Gamma(\nu)}(\textrm{\upperRomannumeral{1}}+\textrm{\upperRomannumeral{2}})\,.
	\end{align}
	It is clear that $\upperRomannumeral{1}\to 0$ as $t\to 0$. Now set $w=\beta_t(1-x)$ and rewrite \textit{\upperRomannumeral{2}} as
	\begin{equation}
		\textrm{\upperRomannumeral{2}}=\int_0^{\epsilon\beta_t}w^{\nu-1}e^{-w}\left(1-\frac{w}{\beta_t}\right)^{\kappa-\nu}dw\,.
	\end{equation}
	Then, we have
	\begin{equation}
		c\sube\gamma(\nu,\epsilon\beta_t)\leq\textrm{\upperRomannumeral{2}}\leq C\sube\gamma(\nu,\epsilon\beta_t)\,,
	\end{equation}
	where $c\sube=\min{\{(1-\epsilon)^{\kappa-\nu},1\}}$ and $C\sube=\max{\{(1-\epsilon)^{\kappa-\nu},1\}}$, and $\gamma$ is the incomplete gamma function. Since $\beta_t\to\infty$ as $t\to 0$, we have
	\begin{equation}
		c\sube\Gamma(\nu)\leq\liminf_{t\to 0}{\textrm{\upperRomannumeral{2}}}\leq\limsup_{t\to 0}{\textrm{\upperRomannumeral{2}}}\leq C\sube\Gamma(\nu)\,,
	\end{equation}
	and hence
	\begin{equation}
		c\sube y^{\nu}\leq\liminf_{t\to 0}{\EP[Y_t^{\nu}]}\leq\limsup_{t\to 0}{\EP[Y_t^{\nu}]}\leq C\sube y^{\nu}\,.
	\end{equation}
	Note that $\epsilon$ was arbitrarily chosen among $(0,1)$. Therefore, letting $\epsilon\to 0$ yields
	\begin{equation}
		\lim_{t\to 0}\EP[Y_t^{\nu}]=y^{\nu}.
	\end{equation}
\end{proof}

\begin{lemma} \label{lem:exp_finite:3/2}
	Let \textit{Y} be a 3/2 process solving the SDE \eqref{SDE:append:3/2} with $b,\sigma>0$ and $a>-\sigma^2/2$. Then for any $\epsilon_0$ satisfying $0<\epsilon_0<\frac{1}{2}(\frac{\sigma}{2}+\frac{a}{\sigma})^2$,
	\begin{equation}
		\EP\left[e^{\epsilon_0\int_0^tY_u\,du}\right]<\infty,\quad\forall t>0\,.
	\end{equation} 
\end{lemma}

\begin{proof}
	Set $\tilde{\epsilon}_0:=\frac{1}{2}(\frac{\sigma}{2}+\frac{a}{\sigma})^2$ and $\gamma:=\sqrt{2\tilde{\epsilon}_0}/\sigma$. We first show that
	\begin{equation}
		\EP\left[Y_t^{\gamma}e^{\tilde{\epsilon}_0\int_0^tY_u\,du}\right]<\infty\,,
	\end{equation}
	for any fixed $t> 0$. The process
	\begin{equation}
		\widetilde{Y}_s:=Y_s^{\gamma}e^{b\gamma(t-s)+\tilde{\epsilon}_0\int_0^sY_u\,du},\quad 0\leq s\leq t\,,
	\end{equation}
	is a local martingale, and hence a supermartingle. Indeed, using It\^{o} formula,
	\begin{align}
		d\widetilde{Y}_s&=\left(-b\gamma+\tilde{\epsilon}_0Y_s+\gamma(b-aY_s)+\frac{1}{2}\sigma^2\gamma(\gamma-1)Y_s\right)\widetilde{Y}_sds+\sigma\gamma\sqrt{Y_s}\widetilde{Y}_sdB_s \\
		&=\left(\tilde{\epsilon}_0+\frac{1}{2}\sigma^2\gamma(\gamma-1)-a\gamma\right)Y_s\widetilde{Y}_sds+\sigma\gamma\sqrt{Y_s}\widetilde{Y}_sdB_s\\
		&=\sigma\gamma\sqrt{Y_s}\widetilde{Y}_sdB_s\,,
	\end{align}
	because $\tilde{\epsilon}_0+\sigma^2\gamma(\gamma-1)/2-a\gamma=0$. Thus,
	\begin{equation} \label{temp:lem_proof:3/2}
		\EP\left[Y_t^{\gamma}e^{\tilde{\epsilon}_0\int_0^tY_u\,du}\right]=\EP\left[\widetilde{Y}_t\right]\leq\widetilde{Y}_0=e^{b\gamma t}y^{\gamma}<\infty\,.
	\end{equation}
	We now prove the lemma. It suffices to show that for any $r\in(0,1)$
	\begin{equation}
		\EP\left[e^{r\tilde{\epsilon}_0\int_0^tY_u\,du}\right]<\infty\,.
	\end{equation}
	Note that $Y_s>0,\,0\leq s\leq t$, almost surely. Using H\"{o}lder's inequality,
	\begin{align}
		\EP\left[e^{r\tilde{\epsilon}_0\int_0^tY_u\,du}\right]&=\EP\left[\left(Y_t^{\gamma}e^{\tilde{\epsilon}_0\int_0^tY_u\,du}\right)^rY_t^{-r\gamma}\right] \\
		&\leq\EP\left[Y_t^{\gamma}e^{\tilde{\epsilon}_0\int_0^tY_u\,du}\right]^r\EP\left[Y_t^{-\frac{r}{1-r}\gamma}\right]^{1-r}\,.
	\end{align}
	The inequality \eqref{temp:lem_proof:3/2} and the fact that the \textit{n}-th moment of a CIR process is finite for arbitrary $n>0$ show that the right-hand side in the above equation is finite. This completes the proof.
\end{proof}

\subsection{Proof of Theorem \ref{thm:static_fund:3/2}} \label{subsec:3/2:fund_sep}

Hereinafter, the symbols to be defined are valid until the end of this section. Assumption \ref{assume:3/2} exactly corresponds to the Feller condition for the CIR process $1/Z_t$. Hence, the SDE
\begin{equation} \label{eqn:3/2:u:SDE}
	dZ_t=(b-\theta Z_t)Z_tdt+\sigma Z_t^{3/2}dB_t,\quad Z_0=z>0,\quad \mbp_z\mbox{-a.s.}
\end{equation}
has a unique strong solution and the solution remains in $(0,\infty)$ for all $t\geq 0$ for each $z\in(0,\infty)$. In addition, it can be directly shown that all the conditions in \citealp[Theorem 1]{heath2000martingales} hold. Therefore, a function \textit{u} defined as
\begin{equation}
	u(t,z):=\EP_z\left[\exp{\left\{\int_{0}^{t}\dfrac{1}{\delta}\left(p\,r-\dfrac{q}{2}\lVert\mu\rVert^2Z_s\right)\,ds\right\}} \right]\,,\quad (t,z)\in[0,T)\times(0,\infty),
\end{equation}
is a $C^{1,2}((0,T)\times (0,\infty))$ solution to the PDE \eqref{eqn:3/2:PDE}, which proves (i).

It can be directly shown that a pair
\begin{equation} \label{eigpair:3/2}
	(\lambda,\phi(z)):=\left(b\eta-\frac{p\,r}{\delta},z^{-\eta} \right)\,,
\end{equation}
where
\begin{equation} \label{eigpair:3/2:const}
	\eta:=-\left(\frac{1}{2}+\frac{\theta}{\sigma^2}\right)+\sqrt{\left(\frac{1}{2}+\frac{\theta}{\sigma^2}\right)^2\!+\frac{q\lVert\mu\rVert^2}{\delta\sigma^2}}\,,
\end{equation}
is an eigenpair of
\begin{equation}
	\mathscr{L}=\frac{1}{2}\sigma^2z^3\frac{d^2}{dz^2}+\left(bz-\theta z^2\right)\frac{d}{dz}+\frac{1}{\delta}\left(p\,r-\frac{q}{2}\lVert\mu\rVert^2z\right)\cdot.
\end{equation}
This is indeed the positive recurrent eigenpair (see Remark \ref{rmk:1d:positive_recur} and Proposition \ref{prop:pos_recur}). Thus, we can define a new measure $\tmbp_z$ on $\mathcal{F}_t$ by
\[ \frac{d\tmbp_z}{d\mbp_z}=\exp{\left\{\lambda t+\int_{0}^{t}\dfrac{1}{\delta}\left(p\,r-\dfrac{q}{2}\lVert\mu\rVert^2Z_s\right)\,ds\right\}}\left(\frac{Z_t}{z}\right)^{-\eta} \]
for each $t\geq 0$, and by Girsanov's theorem
\begin{equation} \label{BM:3/2:2}
	\widetilde{B}_s:=B_s+\int_{0}^{s}\sigma\eta\sqrt{Z_u}du,\quad 0\leq s\leq t
\end{equation}
is a Brownian motion under $\tmbp_z$. Then, the Hansen--Scheinkman decomposition applied to \textit{u} gives the form
\begin{equation} \label{eqn:3/2:u:sol2}
	u(t,z)=e^{-\lambda t}z^{-\eta}\ETP_z\left[Z_t^{\eta}\right]\,,
\end{equation}
where $(Z_s)_{0\leq s\leq t}$ is the solution to the SDE 
\begin{equation} \label{SDE:append:3/2:2}
	dZ_s=\Bigl(b-(\theta+\sigma^2\eta)Z_s\Bigl)Z_s\,ds+\sigma Z_s^{3/2}d\widetilde{B}_s,\quad 0\leq s\leq t\quad Z_0=z,\quad \tmbp_z\mbox{-a.s.}
\end{equation}
This proves (ii).

As discussed in Subsection \ref{subsec:verification}, the verification argument holds because the SDE \eqref{SDE:append:3/2:2} is well-posed. Thus, (iii) and the dynamic fund separation illustrated in (vi) hold. It remains to show that the intertemporal hedging weight $f_z(T-t,z)/f(T-t,z)$ vanishes exponentially fast as terminal time goes to infinity.

The invariant probability density $\tilde{\phi}$ of \textit{Z} is
\begin{equation} \label{eqn:invar_density:append:3/2}
	\tilde{\phi}(z)=\frac{(2b/\sigma^2)^{2(\theta+\sigma^2\eta)/\sigma^2+2}}{\Gamma(2(\theta+\sigma^2\eta)/\sigma^2+2)}z^{-2(\theta+\sigma^2\eta)/\sigma^2-3}e^{-2b/(\sigma^2z)}
\end{equation}
(\citealp[Theorem 5.1.10]{pinsky1995positive}) and $z^{\eta}\in L^1(\tilde{\phi})$. By the ergodic theorem, for each $z\in(0,\infty),\,\,f(t,z)=\ETP_z\left[Z_t^{\eta}\right]\to\int_Ex^{\eta}\tilde{\phi}(x)dx$ as $t\to\infty$. Moreover, \textit{f} is continuous at $t=0$ for each $z\in(0,\infty)$ by Lemma \ref{lem:expect_conti:3/2}. Thus, $f(\cdot,z)$ is bounded and bounded away from zero for each \textit{z}. Now we investigate the behavior of $f_z$. To emphasize the initial value we also denote \textit{Z} by $Z^z$ for a moment. By Lemma \ref{lem:diff:3/2}, the derivative of a map $z\mapsto (Z_t^z)^{\eta}$ is given by
\begin{equation} \label{eqn:3/2:Z_partial}
	\frac{\partial}{\partial z}(Z_t^z)^{\eta}=\frac{\eta}{z^{3/2}}(Z_t^z)^{\eta+1/2}\exp{\left\{-\frac{b}{2}t-\left(\frac{\theta+\sigma^2\eta}{2}+\frac{3\sigma^2}{8}\right)\int_0^t(Z_s^z)\,ds\right\}}\,.
\end{equation}
Since $z^{\eta+1/2}\in L^1(\tilde{\phi})$, as discussed in Appendix \ref{sec:recur} $Z_t^{\eta+1/2}\in L^1(\tmbp_z)$. We also note by \eqref{eqn:3/2:Z_partial} that the map $z\mapsto (Z_t^z)$ is increasing. Therefore, we can apply Corollary \ref{cor:perturb:flow} to calculate the derivative
\begin{align}
	f_z(t,z)&=\frac{\partial}{\partial z}\ETP_z\left[Z_t^{\eta}\right] \\&=\ETP_z\left[\frac{\partial}{\partial z}Z_t^{\eta}\right] \\&=\frac{\eta}{z^{3/2}}\ETP_z\left[Z_t^{\eta+1/2}\exp{\left\{-\left(\frac{\theta+\sigma^2\eta}{2}+\frac{3\sigma^2}{8}\right)\int_0^tZ_s\,ds-\frac{b}{2}t \right\}}\right]\,,
\end{align}
and (iv) follows. The positive recurrent eigenpair of the operator $\mathscr{L}^{\phi}$ (see \eqref{operator:1d:Lphi}) is \begin{equation}
	\left(\hat{\lambda},\hat{\phi}(z)\right)=\left(\frac{b}{2},z^{-1/2}\right)\,.
\end{equation}
As before, we define a probability measure $\hmbp_z$ on $\mathcal{F}_t$ by
\[ \frac{d\hmbp_z}{d\tmbp_z}=\exp{\left\{\frac{b}{2}t-\left(\frac{\theta+\sigma^2\eta}{2}+\frac{3\sigma^2}{8}\right)\int_0^tZ_s\,ds \right\}}\frac{\hat{\phi}(Z_t)}{\hat{\phi}(z)}\]
for each $t\geq 0$. Then under $\hmbp_z,\,f_{z}$ is written as
\begin{equation} \label{eqn:f_z:append:3/2}
	f_z(t,z)=\eta e^{-bt}z^{-1/2}\mathbb{E}_z^{\hmbp}\left[Z_t^{\eta+1}\right]\,,
\end{equation}
where \textit{Z} is the solution of the following SDE with a Brownian motion $(\widehat{B}_s)_{0\leq s\leq t}$:
\begin{equation}
	dZ_s=\left(b-\left(\theta+\sigma^2(\eta+1/2)\right)Z_s\right)Z_s\,ds+\sigma Z_s^{3/2}d\widehat{B}_s,\quad 0\leq s\leq t,\quad\hmbp_z\mbox{-a.s.}
\end{equation}
It can be directly shown that a map $z\mapsto z^{\eta+\xi+1/2}$ is integrable with respect to the invariant probability measure of the transition measure of $Z_t,\,t\geq 0$. Then, by the ergodic property, the expectation term on the right hand side converges to some positive value as $t\to\infty$. In addition, $f_z(\cdot,z)$ is continuous on $[0,\infty)$ for each $z\in E$ by
Lemma \ref{lem:expect_conti:3/2}. So far we have shown that both $f(\cdot,z)$ and $f_z(\cdot,z)$ are bounded and bounded away from zero for each $z\in E$. Therefore, there exists positive functions $C_1,C_2:(0,\infty)\to\mbr$ such that
\begin{equation}
	C_1(z)e^{-bt}\leq\left|\frac{f_z(t,z)}{f(t,z)}\right|\leq C_2(z)e^{-bt}\,,
\end{equation}
and, hence, (v) and the static fund separation described in (vi) are now clear.

\subsection{Proof of Theorem \ref{thm:perturb:3/2}} \label{subsec:3/2:stability}

\textbf{Proof of (i).} Recall that in the proof of Theorem \ref{thm:static_fund:3/2} it was derived that
\begin{equation}
	\left|\frac{f_z(t,z)}{f(t,z)}\right|\leq C(z)e^{-bt}\,,
\end{equation}
for some positive function $C:(0,\infty)\to\mbr$. Since
\begin{equation}
	\frac{\partial}{\partial z}\left(\frac{f_z(t,z)}{f(t,z)}\right)=\frac{f_{zz}(t,z)}{f(t,z)}-\left(\frac{f_z(t,z)}{f(t,z)}\right)^2\,,
\end{equation}
it suffices to analyze the first term on the right-hand side. But we also recall that for each $z\in E$ the function $f(\cdot,z)$ is bounded and bounded away from zero. Thus, it remains to show that $f_{zz}(t,z)\to 0$ as $t\to\infty$ for each $z\in E$.

Lemma \ref{lem:diff:3/2} and Corollary \ref{cor:perturb:flow} allow \eqref{eqn:f_z:append:3/2} to be differentiated with respect to \textit{z} as follows:
\begin{align}
	f_{zz}(t,z)&=\frac{\partial}{\partial z}\left(\eta e^{-bt}z^{-1/2}\mathbb{E}_z^{\hmbp}\left[Z_t^{\eta+1}\right]\right)\\
	&=-\eta e^{-bt}z^{-3/2}\left(\frac{1}{2}+bz\right)\mathbb{E}_z^{\hmbp}\left[Z_t^{\eta+1}\right]\\
	&\quad+\eta(\eta+1) e^{-\frac{3}{2}bt}z^{-3/2}\EHP_z\left[Z_t^{\eta+3/2}
	e^{-\left(\theta/2+\left(\eta/2+5/8\right)\sigma^2\right)\int_0^tZ_s\,ds}\right] \\
	&=-\eta e^{-bt}z^{-3/2}\left(\frac{1}{2}+bz\right)\mathbb{E}_z^{\hmbp}\left[Z_t^{\eta+1}\right]+\eta(\eta+1)e^{-2bt}z^{-2}\EBP_z\left[Z_t^{\eta+2}\right]\,, \label{temp:initial:append:3/2}
\end{align}
where $(Z,\bmbp_z)$ satisfies
\begin{equation}
	dZ_s=\Bigl(b-(\theta+\sigma^2(\eta+1))Z_s\Bigl)Z_s\,ds+\sigma Z_s^{3/2}d\bar{B}_s,\, 0<s\leq t,\quad Z_0=z,\quad \bmbp_z\mbox{-a.s.}
\end{equation}
As in the proof of Theorem \ref{thm:static_fund:3/2}, we can show that the expectation terms in \eqref{temp:initial:append:3/2} converge to some positive numbers for each $z\in E$. Thus, we have
\begin{equation}
	\left|\frac{\partial}{\partial z}\left(\frac{u_z(t,z)}{u(t,z)}\right)-\frac{\eta}{z^2}\right|=\left|\frac{\partial}{\partial z}\left(\frac{f_z(t,z)}{f(t,z)}\right)\right|\leq C(z)e^{-bt}
\end{equation}
for some positive function $C:(0,\infty)\to\mbr$. Finally, based on the result in the previous subsection, we conclude that there exists a positive function $C:(0,\infty)\to\mbr$ such that
\begin{align}
	\left|\frac{\partial}{\partial z}\hat{\pi}_T(t,z)\right|&=\left|\frac{(\Sigma')^{-1}\rho\,\sigma\delta}{1-p}\left(\frac{u_z(T-t,z)}{u(T-t,z)}-z\frac{\partial}{\partial z}\!\left(\frac{u_z(T-t,z)}{u(T-t,z)}\right)\right)\right| \\
	&=\left|\frac{(\Sigma')^{-1}\rho\,\sigma\delta}{1-p}\left(\left(\frac{u_z(T-t,z)}{u(T-t,z)}-\frac{\phi'(z)}{\phi(z)}\right)-z\left(\frac{\partial}{\partial z}\!\left(\frac{u_z(T-t,z)}{u(T-t,z)}-\frac{\phi'(z)}{\phi(z)}\right)\right)\right)\right| \\
	&\leq\left|\frac{(\Sigma')^{-1}\rho\,\sigma\delta}{1-p}\right|\left(\left|\frac{f_z(T-t,z)}{f(T-t,z)}\right|+z\left|\frac{\partial}{\partial z}\!\frac{f_z(T-t,z)}{f(T-t,z)}\right|\right) \\
	&\leq C(z)e^{-b(T-t)}\,,
\end{align}
for any $T>0,\,0\leq t<T$.\vspace{1cm}

\noindent\textbf{Proof of (ii).} Since
\begin{equation}
	\frac{\partial}{\partial b}\!\left(\frac{f_z(t,z)}{f(t,z)}\right)=\frac{f_z(t,z)}{f(t,z)}\left(\frac{\partial}{\partial b}\!\log{f_z(t,z)}-\frac{\partial}{\partial b}\!\log{f(t,z)}\right)\,.
\end{equation}
the large time behaviors of the logarithmic derivatives of \textit{f} and $f_z$ with respect to \textit{b} should be of our interest. Consider a family of pairs $(Z\supe,\tmbp_z\supe)_{z\in E,\epsilon\in I}$ satisfying
\begin{equation} \label{eqn:drift_b_Yeps:append:3/2}
	dZ_t^{\epsilon}=\left(b+\epsilon-(\theta+\sigma^2\eta) Z_t^{\epsilon}\right)Z_t\supe dt+\sigma(Z_t^{\epsilon})^{3/2}d\tB_t\supe,\quad Z_0^{\epsilon}=z\in E,\quad\tmbp_z\supe\mbox{ - a.s.}
\end{equation}
where $\tB\supe$ is a Brownian motion under $\tmbp_z\supe$ for each $\epsilon\in I$. Then clearly
\begin{equation}
	\frac{\partial}{\partial b}f(t,z)=\de\ETPe_z[(Z_t\supe)^{\eta}]\,.
\end{equation}
To apply Proposition \ref{prop:perturb:drift}, it is necessary to show that the conditions \ref{cond1:thm:drift} - \ref{cond3:thm:drift} therein hold with
\begin{equation}
	g(z)=\frac{1}{\sigma\sqrt{z}}\quad\mbox{and}\quad\psi(z)=z^{\eta}\,.
\end{equation} 

\noindent\textbf{\ref{cond1:thm:drift}}: Since $1/Z$ is a CIR process, it follows from \citealp[Lemma 3.1]{wong2006changes} that
\begin{equation}
	\ETP_z\left[\exp{\left\{\epsilon_0\!\int_0^t\frac{1}{\sigma^2Z_s}\,ds\right\}}\right]<\infty
\end{equation}
for $\epsilon_0\leq b^2/2$. \\

\noindent\textbf{\ref{cond2:thm:drift}} : The distribution of a CIR process is well-known (e.g., \citet{cox1985theory}), and it is easily derived that
$\ETP_z[1/Z_s^{1+\epsilon_1/2}]<\infty$ for all $\epsilon_1>0$ and $s\geq 0$, and $s\mapsto\ETP_z[1/Z_s^{1+\epsilon_1/2}]$ is continuous on $[0,\infty)$. Therefore, we have
\begin{equation}
	\ETP_z\left[\int_0^t\left(\frac{1}{\sigma\sqrt{Z_s}}\right)^{2+\epsilon_1}\!ds\right]<\infty
\end{equation}
for all $\epsilon_1>0$. \\

\noindent\textbf{\ref{cond3:thm:drift}} : Recall that the invariant probability density of \textit{Z} is given by \eqref{eqn:invar_density:append:3/2}. It is straightforward to see that $z^{2\eta}\in L^1(\tilde{\phi})$. Therefore,
\begin{equation}
	\ETP_z\left[Z_t^{2\eta}\right]<\infty\,.
\end{equation}

Now we can apply Proposition \ref{prop:perturb:drift}:
\begin{align}
	\left|\de\ETPe_z[(Z_t\supe)^{\eta}]\right|&=\left|\ETP_z\left[Z_t^{\eta}\int_0^t\frac{1}{\sigma\sqrt{Z_s}}d\tB_s\right]\right| \\
	&\leq\left(\ETP_z\left[Z_t^{2\eta}\right]\right)^{1/2}\left(\ETP_z\left[\int_0^t\frac{1}{\sigma^2Z_s}\,ds\right]\right)^{1/2}\,\,\mbox{(H\"{o}lder's inequality, It\^{o} isometry)} \\
	&\leq C(z\,;\,b,a,\sigma)\sqrt{t}\,, \qquad\mbox{(Ergodic property)}
\end{align}
for some positive function $C:(0,\infty)\to\mbr$ depending on \textit{b,a}, and $\sigma$. Recall that the function \textit{f} is bounded and bounded away from zero. Therefore, we conclude that
\begin{equation}
	\left|\frac{\partial}{\partial b}\!\log{f(t,z)}\right|=\left|\frac{\frac{\partial}{\partial b}f(t,z)}{f(t,z)}\right|\leq C(z)\sqrt{t}\,,
\end{equation}
for some positive function $C:(0,\infty)\to\mbr$. Next, we take the logarithmic derivative of \eqref{eqn:f_z:append:3/2} with respect to \textit{b}:
\begin{equation}
	\frac{\frac{\partial}{\partial b}f_z(t,z)}{f_z(t,z)}=\frac{\partial}{\partial b}\log{f_z(t,z)}=-t+\de\!\log{\EHP_z\left[Z_t^{\eta+1}\right]}\,.
\end{equation}
Then, we can proceed with the same process as before and have
\begin{equation}
	\de\!\EHP_z\left[Z_t^{\eta+1}\right]\leq C(z\,;\,b,a,\sigma)\sqrt{t}\,,
\end{equation}
for some positive function $C:(0,\infty)\to\mbr$ depending on \textit{b,a}, and $\sigma$. Finally, we conclude that
\begin{align}
	\left|\frac{\partial}{\partial b} \hat{\pi}_T(t,z)\right|&=\left|\frac{(\Sigma')^{-1}\rho\,\sigma z\delta}{1-p}\frac{\partial}{\partial b}\left(\frac{u_z(T-t,z)}{u(T-t,z)}\right)\right| \\
	&\leq \left|\frac{(\Sigma')^{-1}\rho\,\sigma z\delta}{1-p}\right|\left|\frac{f_z(T-t,z)}{f(T-t,z)}\right|\left(\left|\frac{\partial}{\partial b}\log{f(T-t,z)}\right|+\left|\frac{\partial}{\partial b}\log{f_z(T-t,z)}\right|\right) \\
	&\leq C(z\,;\,b,a,\sigma)(1+T)e^{-b(T-t)}
\end{align}
for any $0\leq t<T$ and some positive function $C:(0,\infty)\to\mbr$ depending on \textit{b,a}, and $\sigma$.\vspace{1cm}

\noindent\textbf{Proof of (iii).} Since the parameter \textit{a} is also included in the drift term, the main idea is similar to the proof of (ii). The main difference here is that $\eta$ depends on \textit{a}. As before, we have
\begin{equation}
	\frac{\partial}{\partial a}\!\left(\frac{f_z(t,z)}{f(t,z)}\right)=\frac{f_z(t,z)}{f(t,z)}\left(\frac{\partial}{\partial a}\!\log{f_z(t,z)}-\frac{\partial}{\partial a}\!\log{f(t,z)}\right)\,.
\end{equation}
Consider a family of pairs $(Z\supe,\tmbp_z\supe)_{z\in E,\epsilon\in I}$ satisfying
\begin{equation} \label{eqn:drift_a:append:Yeps}
	dZ_t^{\epsilon}=\left(b-(\theta\sube+\sigma^2\eta\sube) Z_t^{\epsilon}\right)Z_t\supe dt+\sigma(Z_t^{\epsilon})^{3/2}d\tB_t\supe,\quad Z_0^{\epsilon}=z\in E,\quad\tmbp_z\supe\mbox{ - a.s.}
\end{equation}
where $\tB\supe$ is a Brownian motion under $\tmbp_z\supe,\,\theta\sube:=a+\epsilon+q\sigma\rho'\mu=\theta+\epsilon$, and
\begin{equation}
	\eta\sube:=-\left(\frac{1}{2}+\frac{\theta\sube}{\sigma^2}\right)+\sqrt{\left(\frac{1}{2}+\frac{\theta\sube}{\sigma^2}\right)^2\!+\frac{q\lVert\mu\rVert^2}{\delta\sigma^2}}\,.
\end{equation}
for each $\epsilon\in I$. We note that $\eta$ and function \textit{f} depend on \textit{a} only through $\theta$. Thus,
\begin{equation}
	\frac{d\eta}{da}=\frac{\partial\eta}{\partial\theta},\quad\mbox{and}\quad\frac{\partial}{\partial a}f(t,z)=\frac{\partial}{\partial\theta}f(t,z)=\de\ETPe_z[(Z_t\supe)^{\eta\sube}]\,.
\end{equation}
We verify that the conditions in Proposition \ref{prop:perturb:drift} hold. First, we take \textit{I} to be sufficiently small so that $1\not\in I$. Since $(z^{\eta+1}\log{z})^2\in L^1(\tilde{\phi})$, the family
\begin{equation}
	\left(\frac{\partial}{\partial \epsilon}Z_t^{\eta\sube}\right)_{\epsilon\in I}\!=\left(\frac{\partial\eta\sube}{\partial\epsilon}Z_t^{\eta\sube}\log{Z_t}\right)_{\epsilon\in I}
\end{equation}
is uniformly integrable. Then, we verify that \ref{cond1:thm:drift} - \ref{cond3:thm:drift} hold with
\begin{equation}
	g(z)=C\sqrt{z}\quad\mbox{and}\quad\psi(z)=z^{\eta+1}\,,
\end{equation} 
where \textit{C} is a positive constant such that $C>\theta/\sigma+\sigma\eta$. \\

\noindent\textbf{\ref{cond1:thm:drift}}: This immediately follows from Lemma \ref{lem:exp_finite:3/2}.\\

\noindent\textbf{\ref{cond2:thm:drift}}: It suffices to show that for some $\epsilon_1>0,\,\ETP_z[Z_s^{1+\epsilon_1}]<\infty$ for $s\geq 0$ and $s\mapsto\ETP_z[Z_s^{1+\epsilon_1}]$ is continuous on $[0,\infty)$. Using Lemma \ref{lem:expect_conti:3/2}, it can be directly shown that for any $\epsilon_1<2\eta,\,s\mapsto\ETP_z[Z_s^{1+\epsilon_1}]$ is continuous on $[0,\infty)$, and hence continuous on $[0,\infty)$. Therefore, we conclude that
\begin{equation}
	\ETP_z\left[\int_0^tZ_s^{1+\epsilon_1}ds\right]<\infty,\quad\forall\epsilon_1<2\eta\,.
\end{equation}

\noindent\textbf{\ref{cond3:thm:drift}}: It is straightforward to see that $z^{2\eta+2}\in L^1(\tilde{\phi})$. Therefore,
\begin{equation}
	\ETP_z\left[Z_t^{2(\eta+1)}\right]<\infty\,.
\end{equation}
Thus, we can apply Proposition \ref{prop:perturb:drift}:
\begin{align}
	\de\!\ETPe_z\left[(Z_t\supe)^{\eta\sube}\right]=\frac{\partial\eta}{\partial\theta}\ETP_z[Z_t^{\eta}\log{Z_t}]+\left(\frac{\theta}{\sigma}+\sigma\eta\right)\ETP_z\left[Z_t^{\eta}\int_0^t\sqrt{Z_s}d\tB_s\right]\,.
\end{align}
Now the remaining tasks are the same as in (ii). Likewise, the term
\begin{equation}
	\frac{\partial}{\partial b}\log{f_z(t,z)}
\end{equation}
can be dealt with in the same way. Therefore,
\begin{align}
	\left|\frac{\partial}{\partial a} \hat{\pi}_T(t,z)-\frac{\partial}{\partial a}\hat{\pi}_{\infty}(z)\right|&=\left|\frac{(\Sigma')^{-1}\rho\,\sigma z\delta}{1-p}\right|\left|\frac{\partial}{\partial a}\frac{f_z(T-t,z)}{f(T-t,z)}\right| \\
	&\leq \left|\frac{(\Sigma')^{-1}\rho\,\sigma z\delta}{1-p}\right|\left|\frac{f_z(T-t,z)}{f(T-t,z)}\right|\left(\left|\frac{\partial}{\partial a}\log{f(T-t,z)}\right|+\left|\frac{\partial}{\partial a}\log{f_z(T-t,z)}\right|\right) \\
	&\leq C(z\,;\,b,a,\sigma)(1+T)e^{-b(T-t)}
\end{align}
for any $0\leq t<T$ and some positive function $C:(0,\infty)\to\mbr$ depending on \textit{b,a}, and $\sigma$.\vspace{1cm}

\noindent \textbf{Proof of (iv).} As before, we investigate the large-time behavior of the logarithmic derivative of \textit{f} and $f_z$ with respect to $\sigma$.
Consider a family of pairs $(Z\supe,\tmbp_z\supe)_{z\in E,\epsilon\in I}$ satisfying
\begin{equation} \label{eqn:diffusion_Yeps:append:3/2}
	dZ_t^{\epsilon}=\left(b-(\theta\sube+(\sigma+\epsilon)^2\eta\sube) Z_t^{\epsilon}\right)Z_t\supe dt+(\sigma+\epsilon)(Z_t^{\epsilon})^{3/2}d\tB_t\supe,\quad Z_0^{\epsilon}=z\in E,\quad\tmbp_z\supe\mbox{ - a.s.}
\end{equation}
where $\tB\supe$ is a Brownian motion under $\tmbp_z\supe,\,\theta\sube:=a+q(\sigma+\epsilon)\rho'\mu$, and
\begin{equation}
	\eta\sube:=-\left(\frac{1}{2}+\frac{\theta\sube}{(\sigma+\epsilon)^2}\right)+\sqrt{\left(\frac{1}{2}+\frac{\theta\sube}{(\sigma+\epsilon)^2}\right)^2\!+\frac{q\lVert\mu\rVert^2}{\delta(\sigma+\epsilon)^2}}\,.
\end{equation}
for each $\epsilon\in I$. Let us define
\begin{equation}
	\widetilde{Z}_t\supe:=\frac{2}{(\sigma+\epsilon)\sqrt{Z_t\supe}},\,\tilde{z}\supe:=\frac{2}{(\sigma+\epsilon)\sqrt{z}},\quad t\geq 0,\,\epsilon\in I\,.
\end{equation}
Then, $\widetilde{Z}$ satisfies
\begin{equation}
	d\widetilde{Z}_t\supe=\left(\left(\frac{2\theta\sube}{(\sigma+\epsilon)^2}+2\eta\sube+\frac{3}{2}\right)\frac{1}{\widetilde{Z}_t\supe}-\frac{b}{2}\widetilde{Z}_t\supe\right)\,dt-d\tB_t,\quad \widetilde{Z}_0\supe=\tilde{z}\supe.
\end{equation}
We note that $(\widetilde{Z}\supe)^2$ is a CIR process and $(1/\widetilde{Z}\supe)^2$ is a 3/2 process. Thus, combining the demonstrations in (ii) and (iii) directly yields \ref{cond1:thm:drift} - \ref{cond3:thm:drift} in Proposition \ref{prop:perturb:drift}. Since $\eta$ depends on $\sigma$ not only through $\sigma$ itself but also through $\theta$, to reduce confusion we use the notation "$d/d\sigma$" as a differentiation with respect to $\sigma$. That is,
\begin{equation}
	\frac{d\eta}{d\sigma}=\frac{\partial\eta}{\partial\sigma}+\frac{\partial\theta}{\partial\sigma}\frac{\partial\eta}{\partial\theta}=q\rho'\mu\frac{\partial\eta}{\partial\theta}+\frac{\partial\eta}{\partial\sigma}\,.
\end{equation}
Using Corollary \ref{cor:perturb:united}, explicit calculations can be conducted as follows:
\begin{align}
	\de\!\ETPe_z\!\left[(Z_t\supe)^{\eta\sube}\right]&=2\left(\frac{d\eta}{d\sigma}\log{\frac{2}{\sigma}}-\frac{\eta}{\sigma}\right)\ETP_z[Z_t^{\eta}] \\
	&\quad+\left(\frac{2\theta}{\sigma^2}-\frac{1}{\sigma}\frac{\partial\theta}{\partial\sigma}-\sigma\frac{d\eta}{d\sigma}\right)\ETP_z\!\left[Z_t^{\eta}\int_0^t\sqrt{Z_s}d\tB_s\right]\\
	&\quad+\frac{d\eta}{d\sigma}\ETP_z\left[Z_t^{\eta}\log{\left(\frac{\sigma^2}{4}Z_t\right)}\right] \\
	&\quad-\frac{2\eta}{\sigma\sqrt{z}}\,\ETP_z\left[Z_t^{\eta+1/2}\left(\frac{Z_t}{z}\right)^{-3/2}\frac{\partial Z_t}{\partial z}\right]\,.
\end{align}
Each term on the right-hand side can now be dealt with as before. Likewise, the term
\begin{equation}
	\frac{\partial}{\partial\sigma}\log{f_z(t,z)}
\end{equation}
can be dealt with in the same way. Therefore, we conclude that
\begin{align}
	\left|\frac{\partial}{\partial\sigma} \hat{\pi}_T(t,z)-\frac{\partial}{\partial\sigma} \hat{\pi}_{\infty}(t,z)\right|&=\left|\frac{(\Sigma')^{-1}\rho z\delta}{1-p}\right|\left|\frac{f_z(T-t,z)}{f(T-t,z)}+\sigma\frac{\partial}{\partial\sigma}\!\left(\frac{f_z(T-t,z)}{f(T-t,z)}\right)\right| \\
	&\leq\left|\frac{(\Sigma')^{-1}\rho z\delta}{1-p}\right|\left(\left|\frac{f_z(T-t,z)}{f(T-t,z)}\right|+\sigma\left|\frac{\partial}{\partial\sigma}\!\frac{f_z(T-t,z)}{f(T-t,z)}\right|\right) \\
	&\leq C(z\,;\,b,a,\sigma)(1+T)e^{-b(T-t)}
\end{align}
for any $0\leq t<T$ and some positive function $C:(0,\infty)\to\mbr$ depending on \textit{b,a}, and $\sigma$.

\section{Inverse Bessel state process model} \label{sec:invB:append}

The Inverse Bessel process is defined to be the solution to the SDE:
\begin{equation} \label{SDE:append:invB}
	dY_t=(b-aY_t)Y_t^2dt+\sigma Y_t^2dB_t,\quad y\in (0,\infty)\,,
\end{equation}
where $b,\sigma>0$, and $a>-\sigma^2/2$. At a first glance, the well-posedness of the above SDE is not clear, and it needs to be verified.

The following SDE
\begin{equation} \label{squarediff:recip}
	dZ_t=\left(2a+3\sigma^2-2b\sqrt{Z_t}\right)\,dt-2\sigma\sqrt{Z_t}dB_t,\quad Z_0=1/y\,.
\end{equation}
satisfies the condition of linear growth, and hence the weak existence holds (e.g., \citealp[Chapter 4, Theorem 2.3]{ikeda1989stochastic}). Moreover, it can be easily shown that condition (i) in Remark \ref{rmk:1d:positive_recur} holds with $\alpha=0$ and $\beta=\infty$. Thus, a weak solution of the SDE \eqref{squarediff:recip} remains in $(0,\infty)$ for all $t\geq 0$(e.g. \citealp[Chapter 5, Proposition 5.22 (a)]{karatzas1991brownian}). Then, process $Y_t:=1/\sqrt{Z_t},\,t\geq 0$ is well-defined and by It\^{o} formula \textit{Y} satisfies \eqref{SDE:append:invB}. Consequently, the SDE \eqref{SDE:append:invB} has a weak solution and any weak solutions stay in $(0,\infty)$. In addition, since $b(y)=(b-ay)y^2$ and $\sigma(y)=\sigma y^2$ in \eqref{SDE:append:invB} are locally Lipschitz continuous, the pathwise uniqueness holds for \eqref{SDE:append:invB}(e.g. \citealp[Chapter 4, Theorem 3.1]{ikeda1989stochastic}). Therefore, we conclude that the SDE \eqref{SDE:append:invB} with initial value $y\in(0,\infty)$ has a unique strong solution(\citet{yamada1971uniqueness}).

This section proves Theorem \ref{thm:static_fund:invB} and Theorem \ref{thm:perturb:invB} in Section \ref{sec:invB}. We first prove the following three lemmas regarding the inverse Bessel process, which will be frequently used later. The second lemma, which provides a relationship with a 3/2 process, is quite important in that it allows us to capture the behavior of the inverse Bessel process to some extent, even if its explicit distribution is unknown.

\begin{lemma} \label{lem:diff:invB}
	For each $t>0$, let $Y_t^y$ be the unique strong solution of \eqref{SDE:append:invB} with $Y_0=y$. Then, $Y_t^y$ is differentiable with respect to the initial value \textit{y} a.s., and
	\begin{equation} \label{eqn:deriv:invB}
		\frac{\partial Y_t^y}{\partial y}=\left(\frac{Y_t^y}{y}\right)^2\exp{\left\{-(\sigma^2+a)\int_{0}^{t}\left(Y_s^y\right)^2ds\right\}}\,.
	\end{equation}
\end{lemma}
\begin{proof}
	It is easy to see that all the conditions in Proposition \ref{prop:perturb:flow} are satisfied with $b(y)=(b-ay)y^2,\,\sigma(y)=\sigma y^{2},$ and $E=(0,\infty)$. Thus, we can apply Proposition \ref{prop:perturb:flow} to derive that
	\begin{equation}
		\frac{\partial Y_t^y}{\partial y}=\exp{\left\{\int_0^t2bY_s^y-\left(3a+2\sigma^2\right)(Y_s^y)^2ds+2\sigma\int_0^tY_s^ydB_s\right\}}\,.
	\end{equation}
	On the other hand, applying It\^{o} formula to $\log{Y_t^y}$ yields
	\begin{equation}
		Y_t^y=y\exp{\left\{\int_0^tbY_s^y-\left(a+\frac{1}{2}\sigma^2\right)(Y_s^y)^2ds+\sigma\int_0^tY_s^ydB_s\right\}}\,.
	\end{equation}
	Therefore, the result follows from the above two equations.
\end{proof}

\begin{lemma} \label{lem:comparison:invB}
	Let us define a 3/2 process by
	\begin{equation} \label{SDE:3/2:append:invB}
		d\widetilde{Y}_t=\left(2b\beta-(2a-\sigma^2-2b\alpha)\widetilde{Y}_t\right)\widetilde{Y}_tdt+2\sigma\widetilde{Y}_t^{3/2}dB_t,\quad\widetilde{Y}_0=y^2,
	\end{equation}
	where $\alpha,\beta$ are positive real numbers such that $\alpha<(2a+\sigma^2)/2b$ and $\alpha\beta\leq 1/4$.
	Then, for the inverse Bessel process \eqref{SDE:append:invB},
	\begin{equation}
		Y_t^2\leq \widetilde{Y}_t,\quad\forall t\in[0,\infty),\quad\mbox{a.s.}
	\end{equation}
\end{lemma}
\begin{proof}
	By the condition $\alpha<(2a+\sigma^2)/2b$, the 3/2 process \eqref{SDE:3/2:append:invB} is well-defined. Put $\widetilde{Z}:=1/\widetilde{Y}$. Then, $\widetilde{Z}$ is a CIR process satisfying
	\begin{equation} \label{SDE:temp:append:invB}
		d\widetilde{Z}_t=(2a+3\sigma^2-2b\alpha-2b\beta\widetilde{Z}_t)\,dt-2\sigma\sqrt{\widetilde{Z}_t}dB_t,\quad \widetilde{Z}_0=1/y^2.
	\end{equation}
	Note that $\sqrt{x}\leq\alpha+\beta x,\,\forall x>0$, because $\alpha\beta\geq 1/4$. Thus, the drift terms of \eqref{squarediff:recip} and \eqref{SDE:temp:append:invB} have the following relationship:
	\begin{equation}
		2a+3\sigma^2-2b\sqrt{x}\geq 2a+3\sigma^2-2b\alpha-2b\beta,\quad \forall x>0.
	\end{equation}
	By \citealp[Chapter 5, Proposition 2.18]{karatzas1991brownian}, $\widetilde{Z}_t\leq Z_t,\,\forall t\geq 0$, a.s., where \textit{Z} is a process defined by \eqref{squarediff:recip}. Since $Y=1/\sqrt{Z}$ and $\widetilde{Y}=1/\widetilde{Z}$, the proof is completed.
\end{proof}

\begin{lemma} \label{lem:expect_conti:invB}
	Let \textit{Y} be an inverse Bessel process solving \eqref{SDE:append:invB}. For any $1<\nu<2a/\sigma^2+3$, a map $s\mapsto\EP[Y_s^{\nu}]$ is continuous on $[0,\infty)$.
\end{lemma}
\begin{proof}
	Fix $\nu\in(1,2a/\sigma^2+3)$ and define a 3/2 process $\widetilde{Y}$ solving \eqref{SDE:3/2:append:invB} with sufficiently small $\alpha$ and sufficiently large $\beta$ such that $\nu<(2a-2b\alpha)/\sigma^2+3$ and $\alpha\beta\leq 1/4$. Then, by Lemma \ref{lem:comparison:invB}, $Y_t^{\nu}\leq\widetilde{Y}_t^{\nu/2},\,\forall t\geq 0$, a.s. Thus,
	\begin{equation}
		\EP[Y_t^{\nu}]\leq\EP[\widetilde{Y}_t^{\nu/2}],\quad \forall t\geq 0\,.
	\end{equation}
	By Lemma \ref{lem:expect_conti:3/2}, a map $s\mapsto\EP[Y_s^{\nu/2}]$ is continuous on $[0,\infty)$. Thus, the ergodicity of the 3/2 process implies that there exists a positive constant \textit{C} such that
	\begin{equation}
		\sup_{t\geq 0}\EP[Y_t^{\nu}]\leq\sup_{t\geq 0}\EP[\widetilde{Y}_t^{\nu/2}]\leq C\,.
	\end{equation}
	Since $\nu>1$ and $Y_t\to Y_{t_0}$ as $t\to t_0$ a.s. for any $t_0\geq 0$, it is known that $Y_t\rightharpoonup Y_{t_0}$ in $L^{\nu}(\mbp)$. Therefore,
	\begin{equation}
		\EP[Y_t^{\nu}]\to\EP[Y_{t_0}^{\nu}]\quad\mbox{as }t\to t_0,\quad \forall t_0\geq 0\,.
	\end{equation}
\end{proof}

\subsection{Proof of Theorem \ref{thm:static_fund:invB}}
\noindent Hereinafter, the symbols to be defined are valid until the end of this section. Since $\theta=a+q\sigma\rho'\mu>-\sigma^2/2$, the SDE
\begin{equation} \label{eqn:invB:u:SDE}
	dZ_t=(b-\theta Z_t)Z_t^2dt+\sigma Z_t^2dB_t,\quad z>0,\quad \mbp_z\mbox{-a.s.}
\end{equation}
has a unique strong solution. It is straightforward to check that all the conditions in \citealp[Theorem 1]{heath2000martingales} hold, and hence the function
\begin{equation}\label{eqn:invB:u:sol}
	u(t,z)=\EP\left[\exp{\left\{\int_{0}^{t}\dfrac{1}{\delta}\left(p\,r-\dfrac{q}{2}\lVert\mu\rVert^2Z_s^2\right)\,ds\right\}} \right]\,,\quad (t,z)\in[0,T)\times(0,\infty),
\end{equation}
is a $C^{1,2}((0,T)\times(0,\infty))$ to the PDE \eqref{eqn:PDE:invB}, which proves (i).

Direct calculation shows that a pair
\begin{equation} \label{eigpair:invB}
	(\lambda,\phi(z)):=\left(-\dfrac{1}{2}\sigma^2\xi^2+b\xi-\dfrac{p\,r}{\delta},z^{-\eta}e^{\xi/z} \right)\,,
\end{equation}
where
\begin{equation} \label{eigpair:invB:const}
	\eta:=\frac{-\left(\sigma^2/2+\theta\right)+\sqrt{\left(\sigma^2/2+\theta\right)^2+q\sigma^2\lVert\mu\rVert^2/\delta}}{\sigma^2},\quad\xi:=\dfrac{b\eta}{\sigma^2(\eta+1)+\theta}\,,
\end{equation}
is an eigenpair of 
\begin{equation}
	\mathscr{L}=\frac{1}{2}\sigma^2z^4\frac{d^2}{dz^2}+\left(bz^2-\theta z^3\right)\frac{d}{dz}+\frac{1}{\delta}\left(p\,r-\frac{q}{2}\lVert\mu\rVert^2z^2\right)\cdot.
\end{equation}
This is indeed a positive recurrent eigenpair (see Remark \ref{rmk:1d:positive_recur} and Proposition \ref{prop:pos_recur}). Thus, the process
\begin{equation} \label{mart:invB:u}
	M_t=\exp{\left\{\lambda t+\int_{0}^{t}\dfrac{1}{\delta}\left(p\,r-\dfrac{q}{2}\lVert\mu\rVert^2Z_s^2\right)\,ds+\xi\left(\frac{1}{Z_t}-\frac{1}{z}\right)\right\}}\left(\frac{Z_t}{z}\right)^{-\eta}
\end{equation}
is a $\mbp_z$-martingale, and by Girsanov's theorem the process
\begin{equation} \label{invB:BM:2}
	\widetilde{B}_s:=B_s+\int_{0}^{s}\sigma(\eta Z_u+\xi)\,du,\quad 0\leq s\leq t
\end{equation}
is a standard Brownian motion under a new measure $\tmbp_z$ on $\mathcal{F}_t$ defined by
\[ \frac{d\tmbp_z}{d\mbp_z}=M_t \]
for each $t\geq 0$. Then, \textit{u} is rewritten in the form
\begin{equation} \label{eqn:invB:u:sol2}
	u(t,z)=e^{-\lambda t}z^{-\eta}e^{\xi/z}\ETP_z\left[Z_t^{\eta}e^{-\xi/Z_t}\right]\,,
\end{equation}
where $\tmbp_z$-dynamics of $(Z_s)_{0\leq s\leq t}$ is given by
\begin{equation} \label{eqn:invB:f:SDE}
	dZ_s=\Bigl((b-\sigma^2\xi)-(\theta+\sigma^2\eta)Z_s\Bigl)Z_s^2ds+\sigma Z_s^2d\widetilde{B}_s,\,0<s\leq t\,.
\end{equation}
This proves (ii).

Since the SDE \eqref{eqn:invB:f:SDE} has a unique strong solution, (iii) is proven via the argument described in Subsection \ref{subsec:verification}. Thus, the dynamic fund separation follows. The static fund separation will be completed by showing that the intertemporal hedging weight $f_z(T-t,z)/f(T-t,z)$ vanishes exponentially fast as terminal time goes to infinity.

Since the invariant probability density $\tilde{\phi}$ of \textit{Z} is
\begin{equation} \label{eqn:invar_density:append:invB}
	\tilde{\phi}(z)=\frac{(2(b-\sigma^2\xi)/\sigma^2)^{2(\theta+\sigma^2\eta)/\sigma^2+3}}{\Gamma(2(\theta+\sigma^2\eta)/\sigma^2+3)}z^{-2\theta/\sigma^2-2\eta-4}e^{-2(b-\sigma^2\xi)/(\sigma^2z)}
\end{equation}
(refer to \citealp[Theorem 5.1.10]{pinsky1995positive}) and clearly $\eta<2(\theta+\sigma^2\eta)/\sigma^2+3,\,z^{\eta}e^{-\xi/z}\in L^1(\tilde{\phi})$. Thus, for each $z\in(0,\infty),\,f(t,z)=\ETP_z\left[Z_t^{\eta}e^{-\xi/Z_t}\right]\to\int_Ex^{\eta}e^{-\xi/x}\tilde{\phi}(x)dx>0$ as $t\to\infty$. Moreover, by Lemma \ref{lem:expect_conti:invB}, the map $t\mapsto f(t,z)$ is continuous on $[0,\infty)$ for each $z\in(0,\infty)$. These show that \textit{f} is bounded and bounded away from zero. Next, we investigate the behavior of $f_z$. To emphasize the initial value we also denote \textit{Z} by $Z^z$ for a moment. By Lemma \ref{lem:diff:invB}, the derivative of a map $z\mapsto (Z_t^z)^{\eta}e^{-\xi/Z_t^z}$ is given by
\begin{equation} \label{eqn:temp:append:invB}
	\frac{\partial}{\partial z}\left\{\left(Z_t^z\right)^{\eta}e^{-\xi/Z_t^z} \right\}=\frac{1}{z^2}\left(\eta\left(Z_t^z\right)^{\eta+1}+\xi\left(Z_t^z\right)^{\eta}\right)\exp{\left\{-\frac{\xi}{Z_t^z}-\left(\sigma^2(1+\eta)+\theta\right)\int_{0}^{t}(Z_s^z)^2ds\right\}}\,.
\end{equation}
Recall that the $\tmbp_z$-dynamics of $Z_t$ is given by \eqref{eqn:invB:f:SDE}. Since
\[\eta+1<\frac{2\theta}{\sigma^2}+2\eta+3, \]
it is easy to see that
\[\ETP_z\left[Z_t^{\eta+1}\right]<\infty \]
for each $t>0$ and $z>0$ (see Subsection \ref{sec:recur}). Moreover, the derivative of $z\mapsto Z_t$ given in \eqref{eqn:temp:append:invB} is positive. Therefore, we can apply Corollary \ref{cor:perturb:flow} to calculate that
\begin{align}
	f_z(t,z)&=\ETP_z\left[\frac{\partial}{\partial z}\left\{Z_t^{\eta}e^{-\xi/Z_t} \right\}\right] \\&=\frac{1}{z^2}\ETP_z\left[\left(\eta Z_t^{\eta+1}+\xi Z_t^{\eta}\right)\exp{\left\{-\frac{\xi}{Z_t}-\left(\sigma^2(1+\eta)+\theta\right)\int_{0}^{t}Z_s^2ds\right\}}\right]\,,
\end{align}
which gives (iv). We then seek a positive recurrent eigenpair of the operator $\mathscr{L}^{\phi}$ (see \eqref{operator:1d:Lphi}). It can be shown similarly as before that the following pair
\begin{equation}
	\left(\hat{\lambda},\hat{\phi}(z)\right)=\left(-\frac{1}{2}\sigma^2\zeta^2+(b-\sigma^2\xi)\zeta,z^{-1}e^{\zeta/z}\right)\,,
\end{equation}
where \[\zeta=\frac{b-\sigma^2\xi}{\theta+\sigma^2(\eta+2)}, \]
is the positive recurrent eigenpair. We also note here that $\mathcal{W}\geq0,\,\mathcal{W}\not\equiv 0$; thus, $\hat{\lambda}$ must be positive, as discussed in Remark \ref{rmk:1d:positive_recur}. Indeed,
\[\hat{\lambda}=(b-\sigma^2\xi)^2\frac{\theta+\sigma^2(\eta+3/2)}{(\theta+\sigma^2(\eta+2))^2}>0. \]
We again employ the change of measure technique; then for a probability measure $\hmbp_z$ on $\mathcal{F}_t$ defined by
\[ \frac{d\hmbp_z}{d\tmbp_z}=\exp{\left\{\hat{\lambda}t-\left( \sigma^2(1+\eta)+\theta \right)\int_0^tZ_s^2ds \right\}}\frac{\hat{\phi}(Z_t)}{\hat{\phi}(z)}\]
for each $t\geq 0$, and for a Brownian motion $(\widehat{B}_s)_{0\leq s\leq t}$ under $\hmbp_z$
\begin{equation} \label{eqn:f_z:append:invB}
	f_z(t,z)=e^{-\hat{\lambda} t}\frac{1}{z^3}e^{\zeta/z}\mathbb{E}_z^{\hmbp}\left[\left(\eta Z_t^{\eta+2}+\xi Z_t^{\eta+1}\right)\exp{\left\{-\frac{\xi+\zeta}{Z_t}\right\}} \right]
\end{equation}
with
\begin{equation}
	dZ_s=\left(b-\sigma^2(\xi+\zeta)-\left(\theta+\sigma^2(\eta+1)\right)Z_s\right)Z_s^2ds+\sigma Z_s^2d\widehat{B}_s,\,0<s\leq t,\quad\hmbp_z\mbox{-a.s.}
\end{equation}
Hence, the expectation term on the right-hand side converges to some positive constant as $t\to\infty$. In addition, using the same argument as in the proof of Lemma \ref{lem:expect_conti:invB}, we can show that $f_z(\cdot,z)$ is continuous on $[0,\infty)$ for each $z\in E$. Thus, $f_z(\cdot,z)$ is also bounded and bounded away from zero for each $z\in E$ and converges to 0 exponentially fast. Consequently, we conclude that there exists positive functions $C_1,C_2:(0,\infty)\to\mbr$ such that
\begin{equation}
	C_1(z)e^{-\hlambda t}\leq\left|\frac{f_z(t,z)}{f(t,z)}\right|\leq C_2(z)e^{-\hat{\lambda} t}\,,
\end{equation}
and hence (v) and the static fund separation hold.

\subsection{Proof of Theorem \ref{thm:perturb:invB}}
\textbf{Proof of (i).} As in Appendix \ref{subsec:3/2:stability}, it suffices to show that $f_{zz}(t,z)\to 0$ as $t\to\infty$ for each $z\in E$. Using Lemma \ref{lem:diff:invB} and Corollary \ref{cor:perturb:flow},
\begin{align}
	f_{zz}(t,z)&=\frac{\partial}{\partial z}\left(e^{-\hat{\lambda} t}\frac{1}{z^3}e^{\zeta/z}\mathbb{E}_z^{\hmbp}\left[\left(\eta Z_t^{\eta+2}+\xi Z_t^{\eta+1}\right)\exp{\left\{-\frac{\xi+\zeta}{Z_t}\right\}} \right]\right)\\
	&=-e^{-\hat{\lambda} t}\left(\frac{3}{z^4}+\frac{\zeta}{z^5}\right)e^{\zeta/z}\mathbb{E}_z^{\hmbp}\left[\left(\eta Z_t^{\eta+2}+\xi Z_t^{\eta+1}\right)\exp{\left\{-\frac{\xi+\zeta}{Z_t}\right\}} \right]\\
	&\quad+e^{-\hat{\lambda} t}\frac{1}{z^5}e^{\zeta/z}\EHP_z\left[\left(\eta(\eta+2)Z_t^{\eta+3}+(2\eta\xi+\eta\zeta+\xi)Z_t^{\eta+2}+\xi(\xi+\zeta)Z_t^{\eta+1}\right)\right. \\
	&\hspace{4cm}\left.\cdot\exp{\scriptstyle\left\{-\frac{\xi+\zeta}{Z_t}-(\theta+\sigma^2(\eta+2))\int_0^tZ_s^2ds\right\}} \right] \\
	&=-e^{-\hat{\lambda} t}\left(\frac{3}{z^4}+\frac{\zeta}{z^5}\right)e^{\zeta/z}\mathbb{E}_z^{\hmbp}\left[\left(\eta Z_t^{\eta+2}+\xi Z_t^{\eta+1}\right)\exp{\left\{-\frac{\xi+\zeta}{Z_t}\right\}} \right]\\
	&\,+e^{-(\hlambda+\bar{\lambda}) t}\frac{1}{z^4}e^{\zeta/z}\EBP_z\left[\left(\eta(\eta+2)Z_t^{\eta+4}+(2\eta\xi+\eta\zeta+\xi)Z_t^{\eta+3}+\xi(\xi+\zeta)Z_t^{\eta+2}\right)e^{-\frac{\xi+\zeta+\varsigma}{Z_t}} \right] \\ \label{temp:initial:append:invB}
\end{align}
where 
\begin{equation}
	\bar{\lambda}=-\frac{1}{2}\sigma^2\varsigma^2+(b-\sigma^2(\xi+\zeta))\varsigma,\quad \varsigma=\frac{b-\sigma^2(\xi+\zeta)}{\theta+\sigma^2(\eta+4)}\,,
\end{equation}
and $(Z,\bmbp_z)$ satisfies
\begin{equation}
	dZ_s=\Bigl(b-\sigma^2(\xi+\zeta+\varsigma)-(\theta+\sigma^2(\eta+2))Z_s\Bigl)Z_s^2ds+\sigma Z_s^{2}d\bar{B}_s,\,0<s\leq t,\quad Z_0=z,\quad \bmbp_z\mbox{-a.s.}
\end{equation}
where the probability measure $\bmbp_z$ on $\mathcal{F}_t$ is self-explanatory. As in the proof of Theorem \ref{thm:static_fund:invB}, we can show that the expectation terms in \eqref{temp:initial:append:invB} converge to some positive numbers for each $z\in E$. Thus, we have
\begin{equation}
	\left|\frac{\partial}{\partial z}\frac{f_z(t,z)}{f(t,z)}\right|\leq C(z)e^{-\hlambda t}
\end{equation}
for some positive function $C:(0,\infty)\to\mbr$. Finally, together with the result in the previous subsection, we conclude that there exists a positive function $C:(0,\infty)\to\mbr$ such that
\begin{align}
	\left|\frac{\partial}{\partial z}\left(\hat{\pi}_T(t,z)-\hat{\pi}_{\infty}(z)\right)\right|&=\left|\frac{(\Sigma')^{-1}\rho\,\sigma\delta}{1-p}\right|\left|\frac{f_z(T-t,z)}{f(T-t,z)}+z\frac{\partial}{\partial z}\frac{f_z(T-t,z)}{f(T-t,z)}\right| \\
	&\leq\left|\frac{(\Sigma')^{-1}\rho\,\sigma\delta}{1-p}\right|\left(\left|\frac{f_z(T-t,z)}{f(T-t,z)}\right|+z\left|\frac{\partial}{\partial z}\frac{f_z(T-t,z)}{f(T-t,z)}\right|\right) \\
	&\leq C(z)e^{-\hlambda (T-t)}\,,
\end{align}
for any $T>0,\,0\leq t<T$.\vspace{1cm}

\noindent\textbf{Proof of (ii).} As discussed earlier, we have
\begin{equation}
	\frac{\partial}{\partial b}\!\left(\frac{f_z(t,z)}{f(t,z)}\right)=\frac{f_z(t,z)}{f(t,z)}\left(\frac{\partial}{\partial b}\!\log{f_z(t,z)}-\frac{\partial}{\partial b}\!\log{f(t,z)}\right)\,,
\end{equation}
and hence the large time behaviors of the logarithmic derivatives of \textit{f} and $f_z$ with respect to \textit{b} should be of our interest. Consider a family of pairs $(Z\supe,\tmbp_z\supe)_{z\in E,\epsilon\in I}$ satisfying
\begin{equation} \label{eqn:drift_b_Yeps:append:invB}
	dZ_t^{\epsilon}=\left(b+\epsilon-\sigma^2\xi\sube-(\theta+\sigma^2\eta)Z_t^{\epsilon}\right)(Z_t\supe)^2 dt+\sigma(Z_t^{\epsilon})^{2}d\tB_t\supe,\quad Z_0^{\epsilon}=z\in E,\quad\tmbp_z\supe\mbox{ - a.s.}
\end{equation}
where $\xi\sube=(b+\epsilon)\eta/(\sigma^2(\eta+1)+\theta)$, and $\tB\supe$ is a Brownian motion under $\tmbp_z\supe$ for each $\epsilon\in I$. Then clearly
\begin{equation}
	\frac{\partial}{\partial b}f(t,z)=\de\ETPe_z[(Z_t\supe)^{\eta}e^{-\xi\sube/Z_t\supe}]\,.
\end{equation}
Following the notation in Proposition \ref{prop:perturb:drift},
\begin{equation}
	\left(\sigma^{-1}\de b\sube\right)(x)\equiv\sigma^{-1}-\sigma\frac{\partial\xi}{\partial b}\,.
\end{equation}
Thus, conditions \ref{cond1:thm:drift} and \ref{cond2:thm:drift} trivially hold. Since $\eta$ does not depend on \textit{b}, it suffices to show that $\ETP_z[Z_t^{2\eta}]<\infty$ for the condition \ref{cond3:thm:drift} to hold. This follows directly from observing the invariant density \eqref{eqn:invar_density:append:invB}. Therefore, we can apply Proposition \ref{prop:perturb:drift} to deduce
\begin{align}
	\left|\de\ETPe_z[(Z_t\supe)^{\eta}e^{-\xi\sube/Z_t\supe}]\right|&=\left|\frac{\partial\xi}{\partial b}\ETP_z\left[Z_t^{\eta-1}e^{-\xi/Z_t}\right]+\ETP_z\left[Z_t^{\eta}e^{-\xi/Z_t}\left(\sigma^{-1}-\sigma\frac{\partial\xi}{\partial b}\right)B_t\right]\right| \\
	&\leq\left|\frac{\partial\xi}{\partial b}\ETP_z\left[Z_t^{\eta-1}e^{-\xi/Z_t}\right]\right|+\left|\ETP_z\left[Z_t^{2\eta}e^{-2\xi/Z_t}\right]^{1/2}\right|\left|\left(\sigma^{-1}-\sigma\frac{\partial\xi}{\partial b}\right)\sqrt{t}\right| \\
	&\leq C(z\,;\,b,a,\sigma)(1+\sqrt{t}), \hspace{4.5cm}\mbox{(Ergodic property)}
\end{align}
for some positive function $C:(0,\infty)\to\mbr$ depending on \textit{b,a}, and $\sigma$. Since the function \textit{f} is bounded and bounded away from zero, we have
\begin{equation}
	\left|\frac{\partial}{\partial b}\!\log{f(t,z)}\right|=\left|\frac{\frac{\partial}{\partial b}f(t,z)}{f(t,z)}\right|\leq C(z\,;\,b,a,\sigma)(1+\sqrt{t})\,,
\end{equation}
for some positive function $C:(0,\infty)\to\mbr$ depending on \textit{b,a}, and $\sigma$. We now take the logarithmic derivative of \eqref{eqn:f_z:append:invB} with respect to \textit{b}. Here, we also use the notation ``$d/db$'' as a differentiation with respect to \textit{b} (see Appendix \ref{subsec:3/2:stability}). Thus, we write
\begin{equation}
	\frac{d\zeta}{db}=\frac{\partial\zeta}{\partial b}+\frac{\partial\xi}{\partial b}\frac{\partial\zeta}{\partial\xi}\,.
\end{equation}
Similarly,
\begin{equation}
	\frac{d\hlambda}{db}=\frac{\partial\hlambda}{\partial b}+\frac{\partial\xi}{\partial b}\frac{\partial\hlambda}{\partial\xi}+\frac{d\zeta}{db}\frac{\partial\hlambda}{\partial\zeta}\,.
\end{equation}
Then, the logarithmic derivative of \eqref{eqn:f_z:append:invB} is written as
\begin{equation}
	\frac{\frac{\partial}{\partial b}f_z(t,z)}{f_z(t,z)}=\frac{\partial}{\partial b}\log{f_z(t,z)}=-\frac{d\hlambda}{db}t+\frac{d\xi}{db}\frac{1}{z}+\de\!\log{\EHP_z\left[\left(\eta Z_t^{\eta+2}+\xi Z_t^{\eta+1}\right)\exp{\left\{-\frac{\xi+\zeta}{Z_t}\right\}} \right]}\,.
\end{equation}
We then proceed with the same calculations as before:
\begin{equation}
	\de\!\EHP_z\left[\left(\eta Z_t^{\eta+2}+\xi Z_t^{\eta+1}\right)\exp{\left\{-\frac{\xi+\zeta}{Z_t}\right\}} \right]\leq C(z\,;\,b,a,\sigma)(1+\sqrt{t})\,,
\end{equation}
for some positive function $C:(0,\infty)\to\mbr$ depending on \textit{b,a}, and $\sigma$. Finally, we conclude that there exists a positive function $C:(0,\infty)\to\mbr$ depending on \textit{b,a}, and $\sigma$ such that
\begin{align}
	\left|\frac{\partial}{\partial b}\hat{\pi}_T(t,z)-\frac{\partial}{\partial b}\hat{\pi}_{\infty}(z)\right|&=\left|\frac{(\Sigma')^{-1}\rho\,\sigma z \delta}{1-p}\right|\left|\frac{\partial}{\partial b}\frac{f_z(T-t,z)}{f(T-t,z)}\right|\\
	&\leq C(z\,;\,b,a,\sigma)(1+T)e^{-\hlambda (T-t)}\,,
\end{align}
for any $T>0,\,0\leq t<T$.\vspace{1cm}

\noindent\textbf{Proof of (iii).} As before, the large time behaviors of the logarithmic derivatives of \textit{f} and $f_z$ with respect to \textit{a} are of our interest. Consider a family of pairs $(Z\supe,\tmbp_z\supe)_{z\in E,\epsilon\in I}$ satisfying
\begin{equation}
	dZ_t^{\epsilon}=\left(b-\sigma^2\xi\sube-(\theta\sube+\sigma^2\eta\sube)Z_t^{\epsilon}\right)(Z_t\supe)^2 dt+\sigma(Z_t^{\epsilon})^{2}d\tB_t\supe,\quad Z_0^{\epsilon}=z\in E,\quad\tmbp_z\supe\mbox{ - a.s.}
\end{equation}
where $\tB\supe$ is a Brownian motion under $\tmbp_z\supe$,
\begin{equation}
	\theta\sube:=a+\epsilon+q\sigma\rho'\mu=\theta+\epsilon,\quad\eta\sube:=-\left(\frac{1}{2}+\frac{\theta\sube}{\sigma^2}\right)+\sqrt{\left(\frac{1}{2}+\frac{\theta\sube}{\sigma^2}\right)^2\!+\frac{q\lVert\mu\rVert^2}{\delta\sigma^2}},\quad\xi\sube=\frac{b\eta\sube}{\sigma^2(\eta\sube+1)+\theta\sube}
\end{equation}
for each $\epsilon\in I$. Then,
\begin{equation}
	\frac{\partial}{\partial a}f(t,z)=\de\ETPe_z[(Z_t\supe)^{\eta\sube}e^{-\xi\sube/Z_t\supe}]\,.
\end{equation}
We verify that the conditions in Proposition \ref{prop:perturb:drift} hold. First, we take \textit{I} to be sufficiently small so that the family
\begin{equation}
	\left(\frac{\partial}{\partial \epsilon}\left(Z_t^{\eta\sube}e^{-\xi\sube/Z_t\supe}\right)\right)_{\epsilon\in I}\!=\left(\frac{\partial\eta\sube}{\partial\epsilon}Z_t^{\eta\sube}e^{-\xi\sube/Z_t\supe}\log{Z_t}-\frac{\partial\xi\sube}{\partial\epsilon}Z_t^{\eta\sube-1}e^{-\xi\sube/Z_t\supe}\right)_{\epsilon\in I}
\end{equation}
is uniformly integrable. Then, we verify that \ref{cond1:thm:drift} - \ref{cond3:thm:drift} hold with
\begin{equation}
	g(z)=C(1+z)\quad\mbox{and}\quad\psi(z)=z^{\eta+1}\,,
\end{equation} 
where \textit{C} is a positive constant such that $C>\max{\left\{|\sigma\cdot d\xi/da|,|\sigma^{-1}+\sigma\cdot\partial\eta/\partial\theta|\right\}}$. \\

\noindent\textbf{\ref{cond1:thm:drift}}: It is equivalent to
\begin{equation}
	\ETP_z\left[\exp{\left\{\epsilon_0\int_0^tZ_s^2ds\right\}}\right]<\infty\,,
\end{equation}
for sufficiently small $\epsilon_0>0$. By Lemma \ref{lem:comparison:invB}, there exists a 3/2 process $\widetilde{Z}$ such that $Z_t^2\leq\widetilde{Z}_t,\,\forall t\geq 0$, a.s. Thus, by Lemma \ref{lem:exp_finite:3/2} we have, for some sufficiently small $\epsilon_0>0$,
\begin{equation}
	\ETP_z\left[\exp{\left\{\epsilon_0\int_0^tZ_s^2ds\right\}}\right]\leq\ETP_z\left[\exp{\left\{\epsilon_0\int_0^t\widetilde{Z}_s\,ds\right\}}\right]<\infty\,.
\end{equation}

\noindent\textbf{\ref{cond2:thm:drift}} : Similarly, by Lemma \ref{lem:expect_conti:3/2} and Lemma \ref{lem:comparison:invB}, we have
\begin{equation}
	\ETP_z\left[\int_0^tZ_s^{2+\epsilon_1}ds\right]<\infty\,,
\end{equation}
for sufficiently small $\epsilon_1>0$.

\noindent\textbf{\ref{cond3:thm:drift}}: It is straightforward to see that $z^{2\eta+2}\in L^1(\tilde{\phi})$. Therefore,
\begin{equation}
	\ETP_z\left[Z_t^{2(\eta+1)}\right]<\infty\,.
\end{equation}
Thus, we can apply Proposition \ref{prop:perturb:drift}. By H\"{o}lder's inequality and the ergodic property, there exists a positive function $C:(0,\infty)\to\mbr$ depending on $b,a,\sigma$ such that
\begin{align}
	\de\!\ETPe_z\left[(Z_t\supe)^{\eta\sube}e^{-\xi\sube/Z_t\supe}\right]&=\ETP_z\left[\left(\frac{\partial\eta}{\partial\theta}Z_t^{\eta}-\frac{d\xi}{da}Z_t^{\eta-1}\right)e^{-\xi/Z_t}\right] \\
	&\quad -\ETP_z\left[Z_t^{\eta}e^{-\xi/Z_t}\left(\sigma\frac{d\xi}{da}B_t-\left(\frac{1}{\sigma}+\sigma\frac{\partial\eta}{\partial\theta}\right)\int_0^tZ_s\,d\tB_s\right)\right]
	\\&\leq C(z\,;\,b,a,\sigma)(1+\sqrt{t})e^{-\hlambda t}\,.
\end{align}
Likewise, the logarithmic derivative of $f_z$ with respect to \textit{a}
\begin{equation}
	\frac{\partial}{\partial a}\log{f_z(t,z)}\,,
\end{equation}
can be dealt with in the same manner. Therefore, for some positive function $C:(0,\infty)\to\mbr$ depending on $b,a,\sigma$,
\begin{align}
	\left|\frac{\partial}{\partial a}\hat{\pi}_T(t,z)-\frac{\partial}{\partial a}\hat{\pi}_{\infty}(z)\right|&=\left|\frac{(\Sigma')^{-1}\rho\,\sigma z \delta}{1-p}\right|\left|\frac{\partial}{\partial a}\frac{f_z(T-t,z)}{f(T-t,z)}\right|\\
	&\leq C(z\,;\,b,a,\sigma)(1+T)e^{-\hlambda (T-t)}\,,
\end{align}
for any $T>0,\,0\leq t<T$.\vspace{1cm}

\noindent\textbf{Proof of (iv).} As before, we calculate the logarithmic derivatives of \textit{f} and $f_z$ with respect to $\sigma$. Consider a family of pairs $(Z\supe,\tmbp_z\supe)_{z\in E,\epsilon\in I}$ satisfying
\begin{equation} \label{eqn:diffusion_Yeps:append:invB}
	dZ_t^{\epsilon}=\left(b-(\sigma+\epsilon)^2-(\theta\sube+(\sigma+\epsilon)^2\eta\sube) Z_t^{\epsilon}\right)(Z_t\supe)^2 dt+(\sigma+\epsilon)(Z_t^{\epsilon})^{2}d\tB_t\supe,\quad Z_0^{\epsilon}=z\in E,\quad\tmbp_z\supe\mbox{ - a.s.}
\end{equation}
where $\tB\supe$ is a Brownian motion under $\tmbp_z\supe,\,\theta\sube:=a+q(\sigma+\epsilon)\rho'\mu$,
\begin{equation}
	\eta\sube:=-\left(\frac{1}{2}+\frac{\theta\sube}{(\sigma+\epsilon)^2}\right)+\sqrt{\left(\frac{1}{2}+\frac{\theta\sube}{(\sigma+\epsilon)^2}\right)^2\!+\frac{q\lVert\mu\rVert^2}{\delta(\sigma+\epsilon)^2}},\quad\mbox{and}\quad\xi\sube:=\frac{b\eta\sube}{(\sigma+\epsilon)^2(\eta\sube+1)+\theta\sube}\,,
\end{equation}
for each $\epsilon\in I$. Let us define
\begin{equation}
	\widetilde{Z}_t\supe:=\frac{1}{(\sigma+\epsilon)Z_t\supe},\,\tilde{z}\supe:=\frac{1}{(\sigma+\epsilon)z}\quad t\geq 0,\,\epsilon\in I\,.
\end{equation}
Then, $\widetilde{Z}\supe$ satisfies
\begin{equation}
	d\widetilde{Z}_t\supe=\left(\left(\frac{\theta\sube}{(\sigma+\epsilon)^2}+\eta\sube+1\right)\frac{1}{\widetilde{Z}_t\supe}-\frac{b}{\sigma+\epsilon}+(\sigma+\epsilon)\xi\sube\right)\,dt-d\tB_t,\quad \widetilde{Z}_0\supe=\tilde{z}\supe.
\end{equation}
Since the process $1/\widetilde{Z}\supe$ is an inverse Bessel process, the demonstration in the proof of (iii) directly implies that the conditions \ref{cond1:thm:drift} - \ref{cond3:thm:drift} in Proposition \ref{prop:perturb:drift} hold. Thus, Corollary \ref{cor:perturb:united} provides an explicit form of the derivative
\begin{align}
	\de\!\ETPe_z\!\left[(Z_t\supe)^{\eta\sube}e^{-\xi\sube/Z_t\supe}\right]&=-\left(\eta\sigma^{-1}+\frac{d\eta}{d\sigma}\log{\sigma}\right)\ETP_z[Z_t^{\eta}e^{-\xi/Z_t}] \\
	&\quad+\ETP_z\!\left[Z_t^{\eta}e^{-\xi/ Z_t}\left(\int_0^t{\scriptstyle\left(\frac{2\theta}{\sigma^2}-\frac{1}{\sigma}\frac{\partial\theta}{\partial\sigma}-\sigma\frac{d\eta}{d\sigma}\right)}Z_s\,d\tB_s-{\scriptstyle\left(\frac{b}{\sigma^2}+\xi+\sigma\frac{d\xi}{d\sigma}\right)}\tB_t\right)\right]\\
	&\quad+\ETP_z\left[\frac{d\eta}{d\sigma}Z_t^{\eta}\log{\sigma Z_s}e^{-\xi/Z_t}-\left(\frac{\xi}{\sigma}+\frac{d\xi}{d\sigma}\right)Z_t^{\eta+1}e^{-\xi/Z_t}\right] \\
	&\quad-\frac{1}{\sigma z}\ETP_z\left[\left(\eta Z_t^{\eta+1}+\xi Z_t^{\eta}\right)e^{-\xi/Z_t}\left(\frac{Z_t}{z}\right)^{-2}\frac{\partial Z_t}{\partial z}\right]\,.
\end{align}
Each term on the right-hand side can now be handled as before. Likewise, the term
\begin{equation}
	\frac{\partial}{\partial\sigma}\log{f_z(t,z)}
\end{equation}
can be dealt with in the same manner. Therefore, we conclude that
\begin{align}
	\left|\frac{\partial}{\partial\sigma} \hat{\pi}_T(t,z)-\hat{\pi}_{\infty}(t,z)\right|&=\left|\frac{(\Sigma')^{-1}\rho z\delta}{1-p}\right|\left|\frac{f_z(T-t,z)}{f(T-t,z)}+\sigma\frac{\partial}{\partial\sigma}\frac{f_z(T-t,z)}{f(T-t,z)}\right| \\
	&\leq\left|\frac{(\Sigma')^{-1}\rho z\delta}{1-p}\right|\left(\left|\frac{f_z(T-t,z)}{f(T-t,z)}\right|+\sigma\left|\frac{\partial}{\partial\sigma}\frac{f_z(T-t,z)}{f(T-t,z)}\right|\right) \\
	&\leq C(z\,;\,b,a,\sigma)(1+T)e^{-\hlambda(T-t)}
\end{align}
for any $T>0,\,0\leq t<T$ and some positive function $C:(0,\infty)\to\mbr$ depending on \textit{b,a} and $\sigma$.	

\section{Filtered Ornstein-Uhlenbeck state process}\label{sec:FOU:append}
Consider the Ornstein-Uhlenbeck process satisfying
\begin{equation} \label{SDE:append:FOU}
	dY_t=(b-aY_t)dt+\sigma\,dB_t,\quad y\in\mbr\,,
\end{equation}
where $b,a,\sigma>0$. It is well-known that the SDE \eqref{SDE:append:FOU} has a unique strong solution and that the distribution of $Y_t$ follows the normal distribution with mean $e^{-at}y+\dfrac{b}{a}(1-e^{-at})$ and variance $\dfrac{\sigma^2}{2a}(1-e^{-2at})$.

\begin{lemma} \label{lem:diff:FOU}
	Let $Y^y$ be the solution to \eqref{SDE:append:FOU} with $Y_0=y$. Then for each $t\geq 0$, a function $y\mapsto Y_t^y$ is differentiable a.s., and
	\begin{equation} \label{eqn:deriv:FOU}
		\frac{\partial Y_t^y}{\partial y}=e^{-at}\,.
	\end{equation}
\end{lemma}
\begin{proof}
	The proof is very simple. For $\epsilon\neq 0$, we have $d(Y_t^{y+\epsilon}-Y_t^y)=-a(Y_t^{t+\epsilon}-Y_t^y)dt$. Thus, $Y_t^{y+\epsilon}-Y_t^y=\epsilon e^{-at}$.
\end{proof}

\begin{lemma} \label{lem:exp_finite:FOU}
	Let \textit{Y} be an Ornstein-Uhlenbeck process solving the SDE \eqref{SDE:append:FOU} with $b,a,\sigma>0$. Then for any $\epsilon_0$ satisfying $0<\epsilon_0<\frac{a^2}{4\sigma^2}$,
	\begin{equation}
		\EP\left[e^{\epsilon_0\int_0^tY_u^2\,du}\right]<\infty,\quad\forall t>0\,.
	\end{equation} 
\end{lemma}

\begin{proof}
	Define
	\begin{equation}
		\eta=\frac{-a+\sqrt{a^2-4\sigma^2\epsilon_0}}{2\sigma^2},\quad\xi=\frac{2b\eta}{a+2\sigma^2\eta},\quad\lambda=\sigma^2\eta-\frac{1}{2}\sigma^2\eta+b\xi.
		\end{equation}
	Then, it can be directly shown by It\^{o} formula that the process
	\begin{equation}
		\widetilde{Y}_s:=e^{\lambda s-\eta Y_s^2-\xi Y_s+2\epsilon_0 \int_0^sY_u^2\,du},\quad 0\leq s\leq t\,,
	\end{equation}
	is a local martingale, and hence a supermartingale. Thus,
	\begin{equation} \label{temp:lem_proof:FOU}
		\EP\left[e^{\lambda t-\eta Y_t^2-\xi Y_t+2\epsilon_0 \int_0^tY_u^2\,du}\right]=\EP\left[\widetilde{Y}_t\right]\leq\widetilde{Y}_0=e^{-\eta y^2-\xi y}<\infty\,.
	\end{equation}
	We also note that \begin{equation}
		\EP\left[e^{\eta Y_t^2+\xi Y_t}\right]<\infty,
	\end{equation}
	for $\eta<a/\sigma^2$ and $Y_t\sim\mathcal{N}\bigl(e^{-at}y+b(1-e^{-at})/a,\sigma^2(1-e^{-2at})/(2a)\bigl)$. By H\"{o}lder's inequality,
	\begin{align}
		\EP\left[e^{\epsilon_0\int_0^tY_u^2\,du}\right]&=e^{-\frac{\lambda}{2}t}\EP\left[e^{\frac{1}{2}(\eta Y_t^2+\xi Y_t)}e^{\frac{1}{2}(\lambda t-\eta Y_t^2-\xi Y_t+2\epsilon_0 \int_0^tY_u^2\,du)}\right] \\
		&\leq e^{-\frac{\lambda}{2}t}\EP\left[e^{\eta Y_t^2+\xi Y_t}\right]^\frac{1}{2}\EP\left[e^{\lambda t-\eta Y_t^2-\xi Y_t+2\epsilon_0 \int_0^tY_u^2\,du}\right]^\frac{1}{2}\\ &<\infty.
	\end{align}
\end{proof}

\subsection{Proof of Theorem \ref{thm:static_fund:FOU}}
\noindent Let \textit{Z} be the solution to the SDE
\begin{equation}
	dZ_t=(b-(a+q\theta'\mu) Z_t)dt+\norm{\theta}dB_t,\quad Z_0=z\in\mbr.
\end{equation}
By Feynman-Kac Theorem(e.g. \citealp[Chapter 5, Theorem 7.6]{karatzas1991brownian}), the function
\begin{equation}\label{eqn:FOU:u:sol}
	u(t,z)=\EP\left[\exp{\left\{\int_{0}^{t}\dfrac{1}{\delta}\left(p\,r-\dfrac{q}{2}\lVert\mu\rVert^2Z_s^2\right)\,ds\right\}} \right]\,,\quad (t,z)\in[0,T)\times(0,\infty),
\end{equation}
is a $C^{1,2}((0,T)\times(0,\infty))$ to the PDE \eqref{eqn:PDE:FOU}, which proves (i).

Direct calculation shows that a pair
\begin{equation} \label{eigpair:FOU}
	(\lambda,\phi(z)):=\left(\left(\eta-\frac{1}{2}\xi^2\right)\norm{\theta}^2+b\,\xi-\dfrac{p\,r}{\delta},e^{-\eta z^2-\xi z}\right)\,,
\end{equation}
where
\begin{equation} \label{eigpair:FOU:const}
	\eta:=\frac{-(a+q\theta'\mu)+\sqrt{(a+q\theta'\mu)^2+\dfrac{q}{\delta}\norm{\theta}^2\norm{\mu}^2}}{2\norm{\theta}^2},\quad\xi:=\dfrac{2b\,\eta}{a+q\theta'\mu+2\norm{\theta}^2\eta}\,,
\end{equation}
is an eigenpair of 
\begin{equation}
	\mathscr{L}=\frac{1}{2}\norm{\theta}^2\frac{d^2}{dz^2}+\left(b-(a+q\theta'\mu)z\right)\frac{d}{dz}+\frac{1}{\delta}\left(p\,r-\frac{q}{2}\lVert\mu\rVert^2z^2\right)\cdot.
\end{equation}
It is easy to see that a process
\begin{equation} \label{mart:FOU:u}
	M_t=\exp{\left\{\lambda t+\int_{0}^{t}\dfrac{1}{\delta}\left(p\,r-\dfrac{q}{2}\lVert\mu\rVert^2Z_s^2\right)\,ds-\eta Z_t^2-\xi Z_t+\eta z^2+\xi z\right\}}
\end{equation}
is a $\mbp_z$-martingale, and by Girsanov's theorem
\begin{equation} \label{FOU:BM:2}
	\widetilde{B}_s:=B_s+\norm{\theta}\int_{0}^{s}2\eta Z_u+\xi\,du,\quad 0\leq s\leq t
\end{equation}
is a standard Brownian motion under a new measure $\tmbp_z$ on $\mathcal{F}_t$ defined by
\[ \frac{d\tmbp_z}{d\mbp_z}=M_t \]
for each $t\geq 0$. Then, \textit{u} is rewritten in the form
\begin{equation} \label{eqn:FOU:u:sol2}
	u(t,z)=e^{-\lambda t-\eta z^2-\xi z}\ETP_z\left[e^{\eta Z_t^2+\xi Z_t}\right]=e^{-\lambda t-\eta z^2-\xi z}f(t,z)\,,
\end{equation}
where $\tmbp_z$-dynamics of $(Z_s)_{0\leq s\leq t}$ is given by
\begin{equation} \label{eqn:FOU:f:SDE}
	dZ_s=\Bigl((b-\norm{\theta}^2\xi)-(a+q\theta'\mu+2\norm{\theta}^2\eta)Z_s\Bigl)ds+\norm{\theta} d\widetilde{B}_s,\,0<s\leq t\,.
\end{equation}
This proves (ii).

Since the SDE \eqref{eqn:FOU:f:SDE} has a unique strong solution, (iii) is proven by the argument described in Subsection \ref{subsec:verification}. Thus, the dynamic fund separation follows. The static fund separation will be completed by showing that the intertemporal hedging weight $f_z(T-t,z)/f(T-t,z)$ vanishes exponentially fast as terminal time approaches infinity.

Since \textit{Z} is an Ornstein-Uhlenbeck process and $\eta<a+q\theta'\mu+2\norm{\theta}^2\eta/\norm{\theta}^2$, it can be directly observed that the function $t\mapsto f(t,z)$ is continuous on $[0,\infty)$ and converges to some positive constant for each \textit{z}. Thus, $f(\cdot,z)$ is bounded and bounded away from zero for each \textit{z}. Next, we investigate the behavior of $f_z$. We can apply Lemma \ref{lem:diff:FOU} to calculate the derivative 
\begin{align}
	f_z(t,z)=e^{-\hlambda\,t}\ETP_z\left[(2\eta Z_t+\xi)e^{\eta Z_t^2+\xi Z_t}\right]\,,
\end{align}
where
\begin{equation}
	\hlambda:=a+q\theta'\mu+2\norm{\theta}^2\eta,
\end{equation}
which gives (iv). In addition, as a function of \textit{t}, the expectation on the right-hand side is continuous on $[0,\infty)$ and converges to some positive constant. Therefore, we conclude that there exists positive functions $C_1,C_2:\mbr\to\mbr$ such that
\begin{equation}
	C_1(z)e^{-\hlambda t}\leq\left|\frac{f_z(t,z)}{f(t,z)}\right|\leq C_2(z)e^{-\hat{\lambda} t}\,,
\end{equation}
and hence (v) and the static fund separation hold.

\subsection{Proof of Theorem \ref{thm:perturb:FOU}}
\textbf{Proof of (i).} Recall that
\begin{equation}
	\frac{\partial}{\partial z}\left(\frac{f_z(t,z)}{f(t,z)}\right)=\frac{f_{zz}(t,z)}{f(t,z)}-\left(\frac{f_z(t,z)}{f(t,z)}\right)^2.
\end{equation}
Here, $f_{zz}$ can be also computed easily:
\begin{equation}
	f_{zz}(t,z)=e^{-2\hlambda\,t}\,\ETP_z\left[(4\eta^2Z_t^2+4\eta\xi Z_t+2\eta+\xi^2)e^{\eta Z_t^2+\xi Z_t}\right].
\end{equation} Thus, we have
\begin{equation}
	\left|\frac{\partial}{\partial z}\frac{f_z(t,z)}{f(t,z)}\right|\leq C(z)e^{-2\hlambda\,t}
\end{equation}
for some positive function $C:\mbr\to\mbr$. Finally, in conjunction with the result in the previous subsection, we conclude that there exists a positive function $C:\mbr\to\mbr$ such that
\begin{align}
	\left|\frac{\partial}{\partial z}\left(\hat{\pi}_T(t,z)-\hat{\pi}_{\infty}(z)\right)\right|&=\left|\frac{\delta}{1-p}(\Sigma')^{-1}\theta\right|\left|\frac{\partial}{\partial z}\frac{f_z(T-t,z)}{f(T-t,z)}\right| \\
	&\leq C(z)e^{-2\hlambda\,(T-t)}\,,
\end{align}
for any $T>0,\,0\leq t<T$.\vspace{1cm}

\noindent\textbf{Proof of (ii).} Since $\xi$ depends on \textit{b}, we write $\xi=\xi(b)$ as a function of \textit{b}. As discussed earlier, we have
\begin{equation}
	\frac{\partial}{\partial b}\!\left(\frac{f_z(t,z)}{f(t,z)}\right)=\frac{f_z(t,z)}{f(t,z)}\left(\frac{\partial}{\partial b}\!\log{f_z(t,z)}-\frac{\partial}{\partial b}\!\log{f(t,z)}\right)\,.
\end{equation}
Consider a family of pairs $(Z\supe,\tmbp_z\supe)_{z\in E,\epsilon\in I}$ satisfying
\begin{equation} \label{eqn:drift_b_Yeps:append:FOU}
	dZ_t^{\epsilon}=\left(b+\epsilon-\norm{\theta}^2\xi(b+\epsilon)-(a+q\theta'\mu+2\norm{\theta}^2\eta)Z_t^{\epsilon}\right)dt+\norm{\theta}\,d\tB_t\supe,\quad Z_0^{\epsilon}=z\in E,\quad\tmbp_z\supe\mbox{ - a.s.}
\end{equation}
where $\tB\supe$ is a Brownian motion under $\tmbp_z\supe$ for each $\epsilon\in I$. Then clearly
\begin{equation}
	\frac{\partial}{\partial b}f(t,z)=\de\ETPe_z\left[e^{\eta (Z_t\supe)^2+\xi(b+\epsilon) Z_t\supe}\right]\,.
\end{equation}
Following the notation in Proposition \ref{prop:perturb:drift},
\begin{equation}
	\left(\sigma^{-1}\de b\sube\right)(x)\equiv\norm{\theta}^{-1}-\norm{\theta}\frac{d\xi}{d b}\,.
\end{equation}
Thus, the conditions \ref{cond1:thm:drift} and \ref{cond2:thm:drift} trivially hold. Since $\eta$ does not depend on \textit{b}, it suffices to show that $\ETP_z[Z_t^{2\eta}]<\infty$ for the condition \ref{cond3:thm:drift} to hold. This follows directly by observing the invariant density. Therefore, one can apply Proposition \ref{prop:perturb:drift} to deduce
\begin{align}
	\left|\de\ETPe_z[e^{\eta (Z_t\supe)^2+\xi(b+\epsilon) Z_t\supe}]\right|&=\left|\frac{\partial\xi}{\partial b}\ETP_z\left[e^{\eta Z_t^2+\xi Z_t}\right]+\ETP_z\left[e^{\eta Z_t^2+\xi Z_t}\left(\norm{\theta}^{-1}-\norm{\theta}\frac{d\xi}{d b}\right)B_t\right]\right| \\
	&\leq\left|\frac{\partial\xi}{\partial b}\ETP_z\left[e^{\eta Z_t^2+\xi Z_t}\right]\right|+\left|\ETP_z\left[e^{2\eta Z_t^2+2\xi Z_t}\right]^{1/2}\right|\left|\left(\sigma^{-1}-\sigma\frac{d\xi}{d b}\right)\sqrt{t}\right| \\
	&\leq C(z\,;\,b,a,\sigma)(1+\sqrt{t}), 
\end{align}
for some positive function $C:\mbr\to\mbr$ depending on \textit{b,a} and $\sigma$. Since the function \textit{f} is bounded and bounded away from zero, we have
\begin{equation}
	\left|\frac{\partial}{\partial b}\!\log{f(t,z)}\right|=\left|\frac{\frac{\partial}{\partial b}f(t,z)}{f(t,z)}\right|\leq C(z\,;\,b,a,\sigma)(1+\sqrt{t})\,,
\end{equation}
for some positive function $C:(0,\infty)\to\mbr$ depending on \textit{b,a} and $\sigma$. Similarly, it can be also derived that
\begin{equation}
	\left|\frac{\partial}{\partial b}\!\log{f_z(t,z)}\right|\leq C(z\,;\,b,a,\sigma)(1+\sqrt{t}).
\end{equation}
Therefore, there exists a positive function $C:\mbr\to\mbr$ depending on \textit{b,a}, and $\sigma$ such that
\begin{align}
	\left|\frac{\partial}{\partial b}\hat{\pi}_T(t,z)-\frac{\partial}{\partial b}\hat{\pi}_{\infty}(z)\right|&=\left|\frac{\delta}{1-p}(\Sigma')^{-1}\theta\right|\left|\frac{\partial}{\partial b}\frac{f_z(T-t,z)}{f(T-t,z)}\right|\\
	&\leq C(z\,;\,b,a,\sigma)(1+\sqrt{T})e^{-\hlambda\,(T-t)}\,,
\end{align}
for any $T>0,\,0\leq t<T$.\vspace{1cm}

\noindent\textbf{Proof of (iii).} We regard $\theta,\,\eta,\,\xi$, and $\hlambda$ as functions of \textit{a} and denote $\theta\sube=\theta(a+\epsilon),\,\eta\sube=\eta(a+\epsilon,\theta\sube),\,\xi\sube=\xi(a+\epsilon,\theta\sube,\eta\sube),\,\hlambda\sube=\hlambda(a+\epsilon,\theta\sube,\eta\sube)$ for $\epsilon\in I$. Consider a family of pairs $(Z\supe,\tmbp_z\supe)_{z\in E,\epsilon\in I}$ satisfying
\begin{equation}
	dZ_t^{\epsilon}=\left(b-\norm{\theta\sube}^2\xi\sube-(a+\epsilon+q\theta\sube'\mu+2\norm{\theta\sube}^2\eta\sube)Z_t^{\epsilon}\right)dt+\norm{\theta\sube}\,d\tB_t\supe,\quad Z_0^{\epsilon}=z\in E,\quad\tmbp_z\supe\mbox{ - a.s.}
\end{equation}
where $\tB\supe$ is a Brownian motion under $\tmbp_z\supe$ for each $\epsilon\in I$. Let us efine
\begin{equation}
	\widetilde{Z}_t\supe:=\frac{Z_t\supe}{\norm{\theta\sube}},\,\tilde{z}\supe:=\frac{z}{\norm{\theta\sube}}\quad t\geq 0,\,\epsilon\in I\,.
\end{equation}
Then, $\widetilde{Z}\supe$ satisfies
\begin{equation}
	d\widetilde{Z}_t^{\epsilon}=\left(b\norm{\theta\sube}^{-1}-\norm{\theta\sube}\xi\sube-(a+\epsilon+q\theta\sube'\mu+2\norm{\theta\sube}^2\eta\sube)\widetilde{Z}_t^{\epsilon}\right)dt+d\tB_t\supe,\quad \widetilde{Z}_0^{\epsilon}=\tilde{z}\in E,\quad\tmbp_z\supe\mbox{ - a.s.}
\end{equation}
The conditions \ref{cond1:thm:drift} - \ref{cond3:thm:drift} hold for $\widetilde{Z}\supe$. Indeed, \ref{cond2:thm:drift} and \ref{cond3:thm:drift} are trivial, and \ref{cond1:thm:drift} immediately follows from Lemma \ref{lem:exp_finite:FOU}. Thus, applying Corollary \ref{cor:perturb:united} is available. Proceeding with the same calculations as before, we obtain
\begin{align}
	\left|\frac{\partial}{\partial a}\hat{\pi}_T(t,z)-\frac{\partial}{\partial a}\hat{\pi}_{\infty}(z)\right|&\leq\left|\frac{\delta}{1-p}(\Sigma')^{-1}\theta\right|\left|\frac{\partial}{\partial a}\frac{f_z(T-t,z)}{f(T-t,z)}\right|+\left|\frac{\delta}{1-p}(\Sigma')^{-1}\frac{d\theta}{da}\frac{f_z(T-t,z)}{f(T-t,z)}\right|\\
	&\leq C(z\,;\,b,a,\sigma)(1+T)e^{-\hlambda\,(T-t)}\,,
\end{align}
for any $T>0,\,0\leq t<T$.\vspace{1cm}

\noindent\textbf{Proof of (iv).} Since $\theta,\,\eta,\,\xi$, and $\hlambda$ depend also on $\sigma$, the proof is exactly the same as the proof of (iii). In conclusion, we have
\begin{align}
	\left|\frac{\partial}{\partial a}\hat{\pi}_T(t,z)-\frac{\partial}{\partial a}\hat{\pi}_{\infty}(z)\right|&\leq\left|\frac{\delta}{1-p}(\Sigma')^{-1}\theta\right|\left|\frac{\partial}{\partial\sigma}\frac{f_z(T-t,z)}{f(T-t,z)}\right|+\left|\frac{\delta}{1-p}(\Sigma')^{-1}\frac{d\theta}{d\sigma}\frac{f_z(T-t,z)}{f(T-t,z)}\right|\\
	&\leq C(z\,;\,b,a,\sigma)(1+T)e^{-\hlambda\,(T-t)}\,,
\end{align}
for any $T>0,\,0\leq t<T$.

\bibliographystyle{abbrvnat}
	\bibliography{References}
\end{document}